\documentclass[12pt]{article}
\usepackage{amsmath,bbm,amssymb}
\usepackage{amsthm,amscd}
\usepackage{graphicx, graphics}
\usepackage{comment}
\usepackage{subfig}
\usepackage{multirow}
\usepackage{color}
\usepackage{multicol}
\usepackage{array}
\usepackage{natbib}

\newtheorem{theorem}{Theorem}
\newtheorem{lemma}{Lemma}

\usepackage[utf8]{inputenc}
\usepackage{amsmath}
\usepackage{amsfonts}
\usepackage{amssymb}
\usepackage{longtable, lscape}
\usepackage{rotating}

\newtheorem{algorithm}{Algorithm}

\numberwithin{equation}{section}
\numberwithin{proposition}{section} \numberwithin{lemma}{section}
\numberwithin{theorem}{section}

\newcommand{\bm}[1]{\mbox{\boldmath$ #1 $\unboldmath}}

\newcolumntype{x}[1]{>{\raggedright}p{#1}}
\newcolumntype{V}{>{\centering\arraybackslash} m{1cm}}
\newcolumntype{C}{>{\centering\arraybackslash} m{6cm}}

\def\tr{{\rm tr}}

\DeclareMathOperator*{\argmin}{arg\,min}

\newcommand{\nx}[1]{{\color{black}{#1}}}

\begin{document}
\baselineskip=22pt
\vskip 10pt

\begin{center}
\nx{\Large \bf Robust Estimation of Sparse Precision Matrix using Adaptive Weighted Graphical Lasso Approach}\\
%(\today)\\
\vskip 5pt
Peng Tang$^{\dag}$, Huijing Jiang$^{\ddag}$, Heeyoung Kim$^{\S}$,  and Xinwei Deng$^{\star}$\footnote{Address for correspondence: Xinwei Deng, Associate Professor, Department of Statistics,
Virginia Tech, Blacksburg, VA 24061 (E-mail: xdeng@vt.edu).} \\
\vskip 5pt
$^{\dag}$AI, Cruise LLC\\
$^{\ddag}$Global Biometrics and Data Sciences, Bristol Myers Squibb\\
$^{\S}$Department of Industrial \& Systems Engineering, KAIST\\
%$^{\P}$Department of Industrial \& Systems Engineering, Korea Advanced Institute of Science and Technology (KAIST) \\ %Daejeon, Republic of Korea, Email:heeyoungkim@kaist.ac.kr\\
$^{\star}$Department of Statistics, Virginia Tech
\end{center}

%\spacing{1.4}

\begin{abstract}
Estimation of a precision matrix (i.e., inverse covariance matrix) is widely used to exploit conditional independence among continuous variables.
The influence of abnormal observations is exacerbated in a high dimensional setting as the dimensionality increases.
In this work, we propose robust estimation of the inverse covariance matrix based on an $l_1$ regularized objective function with a weighted sample covariance matrix.
The robustness of the proposed objective function can be justified by a nonparametric technique of the integrated squared error criterion.
To address the non-convexity of the objective function, we develop an efficient algorithm in a similar spirit of majorization-minimization.
Asymptotic consistency of the proposed estimator is also established.
The performance of the proposed method is compared with several existing approaches via numerical simulations.
We further demonstrate the merits of the proposed method with application in genetic network inference.
\end{abstract}


\section{Introduction}
Estimation of the inverse covariance matrix has attracted wide attention in various applications such as genetic network inference.
Under the assumption of Gaussian distribution, it is well known that the estimation of undirected Gaussian graphs is equivalent to the estimation of inverse covariance matrices (Whittaker, 1990; Lauritzen, 1996).
Let ${\bf x}=(X^{(1)}, \ldots, X^{(p)})^T \sim N_p(\bm \mu, \bm \Sigma)$ be a $p$-dimensional Gaussian random vector with mean $\bm \mu$ and covariance matrix $\bm \Sigma$.
The inverse covariance matrix ${\bm\Omega}= {\bm \Sigma}^{-1}\equiv(c_{ij})_{1\leq i,j\leq p}$ contains information related to the conditional dependency between variables, which are referred to as nodes in a graph representation.
Each off-diagonal entry of $\bm \Omega$ represents an edge of two nodes in the corresponding graph.
If $c_{ij}=c_{ji}=0$, then $X^{(i)}$ and $X^{(j)}$ are conditionally independent given the other variables $X^{(k)},$ $k\neq i,j$; otherwise, $X^{(i)}$ and $X^{(j)}$ are conditionally dependent.

Various estimation methods for a sparse inverse covariance matrix have been proposed in the literature.
One direction is to exploit the connections between entries of the inverse covariance matrix and  coefficients of the corresponding multivariate linear regressions (Pourahmadi, 2000; Huang et al., 2006).
For example, Meinshausen and B{\"u}hlmann (2006) applied a neighborhood selection scheme by regressing each variable on the remaining variables.
Recently, regularized likelihood-based approaches have been proposed by different researchers, see Banerjee et al.\,(2006), Yuan and Lin\,(2007), and Friedman et al.\,(2008), Lam and Fan (2009), Deng and Yuan\,(2009), Kang and Deng\,(2020), among many others.
These methods identify zeros in the inverse covariance matrix through the penalized log-likelihood functions.
Rather than directly using the likelihood function, Cai et al.\,(2011) investigated a constrained $l_1$ minimization method to estimate the sparse inverse covariance matrix as well as its convergence properties.

The aforementioned approaches often can work effectively when there are no outliers in the observations.
However, when abnormal observations contaminate the data, the estimation of inverse covariance matrix can be sensitive to the presence of outliers.
Especially in the high-dimensional setting, where dimension $p$ is relatively large compared to the sample size $n$, it is non-trivial to check whether the observations follow the underlying model assumptions.
Several methods were proposed to address this concern.
Pe{\'n}a and Prieto\,(2001) presented a robust estimator for the covariance matrix based on the analysis of the projected data.
Miyamura and Kano\,(2006) introduced a robustified likelihood function to estimate the inverse covariance matrix, and they also proposed corresponding test statistics associated with the robustified estimators. The above two methods did not consider the sparsity in the estimates. Later,
Finegold and Drton\,(2011) proposed a so-called tlasso method by adopting multivariate $t$-distribution for robust inferences of graphical models.
Sun and Li (2012) developed an $l_1$ regularization procedure based on density power divergence and used coordinate descent algorithms for robust estimation of the inverse covariance matrix.
Other robust approaches (\"{O}llerer and Croux, 2015,  Tarr, M\"{u}ller, and Weber, 2016, and Loh and Tan, 2018) consider to use a robust estimate of the covariance matrix to replace the sample covariance matrix in solving the glasso (Friedman et al., 2008).

In this work, we propose robust estimation of a sparse precision matrix via maximizing an $l_1$ regularized objective function with weighted sample covariance matrix. The robustness of our proposed method can be justified from a nonparametric perspective.
Specifically, the proposed loss function can be viewed as approximation of the integrated squared error (ISE) criterion, which aims at minimizing the $L_2$ distance between the probability densities using the estimated and true inverse covariance matrices.
This property naturally enables the proposed estimate more resistent to outliers compared to likelihood-based approaches.
Moreover, we encourage the sparsity by adding an $l_1$ penalty in the objective function, in a similar spirit as the Lasso (Tibshirani, 1996) for sparse regression.
Note that the proposed objective function is not convex and thus can be computationally inefficient.
To address this challenge, we develop an efficient algorithm from the majorization-minimization aspects to find an upper bound for the original objective function.
Consequently, the estimation problem can be transformed into an iterative estimation procedure in the content of graphical lasso (Friedman et al., 2008) with an appropriate weighting scheme on observations.
The sequence of weights in each iteration is adaptively updated to automatically account for potential abnormal observations.
We call the proposed method ``\textit{weighted Graphical Lasso}" \nx{(WGLasso)}.

The rest of the work is organized as follows. In Section \ref{sec:method}, we introduce the robust inverse covariance matrix estimate and develop a computationally efficient algorithm for it.
Section \ref{sec:simulation} provides a series of simulations to evaluate the performance of the proposed method.
We continue in Section \ref{sec:case} with a real application using the proposed robust estimate.
A discussion and concluding remarks are given in Section \ref{sec:concl}. The proof of the theorem is deferred to Appendix.

\section{Robust Estimation for Sparse Precision Matrix}\label{sec:method}
In this section, we describe in details the proposed robust estimator of the sparse inverse covariance matrix.
Suppose that ${\bf x}_{1},\ldots, {\bf x}_{n}$ are a random sample of $n$ observations on ${\bf x}=(X^{(1)}, \ldots, X^{(p)})^T$, each of which follows a multivariate normal distribution $N_p(\bm \mu, \bm \Sigma)$ with mean $\bm \mu$ and nonsingular covariance matrix $\bm \Sigma$. Without loss of generality, we assume that $\bm \mu = \bf 0$.

\subsection{Proposed Weighted Graphical Lasso}
To estimate the inverse covariance matrix ${\bf \Omega} = {\bf \Sigma}^{-1} \equiv (c_{ij})_{p \times p}$, we consider a minimization problem as
\begin{align}\label{eq:obj-wgl-3rd}
\min_{\bf \Omega} \left[-\log |{\bf \Omega}| + \tr[{\bf \Omega S}^{*}] + \rho\|{\bf \Omega}\|_{1} \right],
\end{align}
where $\rho \ge 0$ is a tuning parameter, which is to balance the goodness of fit and the sparsity of the estimate. Note that ${\bf S}^{*} = {\bf S}^{*}({\bf \Omega}^{(0)})\equiv \frac{1}{n}  \sum_{i = 1}^{n} w_{i}{\bf x}_{i}{\bf x}_{i}^{T}$ can be viewed as a weighted sample covariance matrix given an initial estimate ${\bf \Omega}^{(0)}$. Here the weight $w_i$ can be written as
\begin{align}\label{eq: weight}
w_{i} = w_{i}({\bf \Omega}^{(0)}) \equiv \frac{ \exp(-\frac{1}{2} {\bf x}_{i}^{T} {\bf \Omega}^{(0)}
	{\bf x}_{i})}{\frac{1}{n} \sum_{j =1}^{n} \exp(-\frac{1}{2} {\bf
		x}_{j}^{T} {\bf \Omega}^{(0)} {\bf x}_{j})}.
\end{align}
The form of the objective function in \eqref{eq:obj-wgl-3rd} is similar to the negative log-likelihood function but with a weighted sample covariance matrix.
The optimization of \eqref{eq:obj-wgl-3rd} can be efficiently solved using the $glasso$ algorithm developed by Friedman et al. (2008).
Therefore, we refer to the solution of (\ref{eq:obj-wgl-3rd}) as the {\it weighted Graphical Lasso} (WGLasso).
It is worth pointing out that the derived objective function in (\ref{eq:obj-wgl-3rd}) bears some similarity to the objective function in the M-step of the EM algorithm of the tlasso method (Fingegold and Drton, 2011).
The tlasso considered the weighted sample covariance with the weights derived based on the EM algorithm, whereas the proposed method provides a weighted Graphical Lasso algorithm with the weights derived based on the ISE criterion (Rudemo 1982) as discussed in Section \ref{sec: robust-motivation}.

The weight in \eqref{eq: weight} also has a meaningful interpretation.
One can see that the weights $w_{i}$ can be rewritten as $w_{i} = f({\bm x}_{i}) / [\frac{1}{n} \sum_{j=1}^{n} f({\bm x}_{j})]$, where $f({\bm x}_{i})$ is the probability density function (pdf) of $N_p(\bm  0,  (\bm\Omega^{(0)})^{-1})$ when ${\bm x}={\bm x}_i$.
The form of ${\bf S}^{*}$ assigns different weights to the observations ${\bm x}_{i}$ based on the ratio of its pdf value relative to the average density.
If an observation is an outlier from the underlying distribution, its pdf value tends to be smaller compared to that of the average; thus its assigned weight would be smaller.
This helps us to reduce the influence from outliers on the proposed estimate.

Note that the expression in \eqref{eq:obj-wgl-3rd} depends on the initial estimate ${\bm \Omega}^{(0)}$ in order to calculate the weights and corresponding ${\bm S}^*$. We iteratively minimize the objective function (\ref{eq:obj-wgl-3rd}) to obtain the estimate of ${\bm \Omega}$. The detailed description of the iterative algorithm is as follows:

\begin{algorithm}\label{wgl}
\noindent

{\bf Step 1}: Obtain an initial estimate $\bm \Omega^{(0)}$, a positive definite matrix, and set the stopping threshold $\delta=10^{-6}$.

{\bf Step 2}: Compute $w_{i}$ according to \eqref{eq: weight} and obtain ${\bm S}^{*} = \frac{1}{n}  \sum_{i = 1}^{n} w_{i}{\bf x}_{i}{\bf x}_{i}^{T}$.

{\bf Step 3}: Estimate $\bm \Omega$ by minimizing \eqref{eq:obj-wgl-3rd}, i.e.,
\begin{align*}
\hat{\bf \Omega} = \arg \min_{\bf \Omega} \left [-\log |{\bf \Omega}| + \tr[{\bf \Omega S}^{*}]+ \rho\|{\bf \Omega}\|_{1} \right ].
\end{align*}

{\bf Step 4}: Stop if the Frobenius norm $\|\hat{\bf \Omega} - {\bf \Omega}^{(0)}\|_{F}^{2} \le \delta$; otherwise, set ${\bf \Omega}^{(0)} = \hat{\bf \Omega}$ and go back to  {\bf Step 2}.
\end{algorithm}

We generally take the inverse of the sample covariance matrix as an initial estimate ${\bf \Omega}^{(0)}$.
If $p>n$, the sample covariance matrix is singular, and we add a small perturbation to its diagonal elements.

\subsection{Justification of Robustness via Integrated Squared Error}\label{sec: robust-motivation}
In this section, we provide some justification of robustness for the proposed estimate from the scope of using the ISE as a loss function.
The ISE criterion (Rudemo, 1982) is to measure the $L_2$ distance between the estimated probability density function and the true density.
Scott (2011) investigated various parametric statistical models by using this criterion as a theoretical and practical estimation tool. The key idea in Scott (2011) is to apply nonparametric techniques to parametric models for the robustness of the estimation.
Based on the ISE as a loss function, one can consider the following minimization for the estimation of the inverse covariance matrix.
\begin{align}\label{eq:L2E-der}
&\int (f({\bf x}|{\bf \Omega}) - f({\bf x}|{\bf \Omega}_{*}))^2 \, d{\bf x} \nonumber \\
=& \int f^2({\bf x}|{\bf \Omega}) \, d{\bf x} - 2 \int f({\bf x}|{\bf \Omega})f({\bf x}|{\bf \Omega}_{*})\,d{\bf x} + \int f^2({\bf x}|{\bf \Omega}_{*}) \, d{\bf x},
\end{align}
where $f({\bf x}|{\bf \Omega})$ is the pdf of the normal distribution $N_p(\bm 0, \bf \Omega^{-1})$ and $\bf \Omega_{{*}}$ is the true but unknown inverse covariance matrix.
Note that the third term on the right hand side of \eqref{eq:L2E-der} does not depend on $\bf \Omega$.
Additionally, the second integral in $\eqref{eq:L2E-der}$ can be viewed as $\mathbb{E}(f(\bf x|\bf \Omega))$ with respect to ${\bf x}$.
With the observed data available, one can approximate $\mathbb{E}(f(\bf x|\bf \Omega))$ by the empirical mean $\frac{1}{n}\sum_{i=1}^n f({\bf x}_{i}|\bf \Omega)$ (Scott (2001), Chi and Scott, (2014)).
Now we reformulate the ISE as
\begin{eqnarray}\label{eq:L2E-loss}
L_{n}(\bf \Omega) &\equiv& \int f^2({\bf x}|{\bf \Omega}) \, d{\bf x} - 2\int f({\bf x}|{\bf \Omega})f({\bf x}|{\bf \Omega}_{*})\, d{\bf x}  \nonumber \\
                  & \approx & \int f^2({\bf x} |{ \bf \Omega}) \, d{\bf x} -  \frac{2}{n}\sum_{i = 1}^{n}f({\bf x}_{i}|{\bf \Omega}) \nonumber \\
                  &=& \frac{1}{2^{p/2}}|{\bf \Omega}|^{1/2} - \frac{2}{n}|{\bf \Omega} |^{1/2}\sum_{i = 1}^{n}  \exp\bigg(-\frac{1}{2} {\bf x}^{T}_i {\bf \Omega} {\bf x}_i \bigg ),
\end{eqnarray}
up to constant $(2\pi)^{-p/2}$.
To encourage the sparsity in the estimate, one can consider the $l_1$ regularization for estimating $\bf \Omega$ as
\begin{align}\label{eq:obj-fun}
L_{n,\rho}({\bf \Omega}) &\equiv L_{n}({\bf \Omega})+ \rho \|{\bf \Omega}\|_{1} \nonumber \\
&\propto -\frac{2}{n} |{\bf \Omega}|^{1/2}\sum_{i =1}^{n}\exp\bigg (-\frac{1}{2} {\bf x}^{T}_{i} {\bf \Omega} {\bf x}_{i}\bigg ) + \tau |{\bf \Omega}|^{1/2} +  \rho \|{\bf \Omega} \|_{1},
\end{align}
where $\|{\bf \Omega}\|_{1} = \sum_{i,j} |v_{ij}|$ is the $l_{1}$ penalty for ${\bf \Omega}=(v_{ij})_{p\times p}$. Here $\tau = 2^{-p/2}$.

It is also feasible to convert the problem into a constrained optimization by the use of Lagrange multiplier.
That is, we can reformulate the above optimization in \eqref{eq:obj-fun} as
\begin{align}\label{eq:obj-ub-con}
\underset{{\bf \Omega}}{\min} & -\frac{2}{n} |{\bf \Omega}|^{1/2}\sum_{i =1}^{n}\exp\left(-\frac{1}{2} {\bf x}^{T}_{i} {\bf \Omega  x}_{i}\right ), \\
s.t. \quad  &  |{\bf \Omega}|^{1/2} \le M_{1}, \nonumber \\
&  \|{\bf \Omega} \|_{1} \le M_{2}. \nonumber
\end{align}
The following lemma provides an inequality between the matrix determinant and its $l_1$ norm.
\begin{lemma}\label{them:bound}
For a positive definite matrix $\bf \Omega$ with dimension $p$, the relationship between its determinant value and $l_1$ norm can be described in the following inequality
\begin{align*}
|{\bf \Omega}|^{1/2} \le  \left ( \frac{ \|{\bf \Omega} \|_{1} }{p} \right )^{p/2}.
\end{align*}
\end{lemma}
The proof of Lemma \ref{them:bound} is based on Gershgorin's circle theorem (Horn and Johnson, 1996).
Using Lemma 2.1, we can transform the nonlinear constraint on matrix determinant to a linear constraint.
Specifically, the minimization under consideration becomes
\begin{align}\label{eq:obj-ub-con}
\underset{{\bf \Omega}}{\min} & -\frac{2}{n} |{\bf \Omega}|^{1/2}\sum_{i =1}^{n}\exp\left(-\frac{1}{2} {\bf x}^{T}_{i} {\bf \Omega  x}_{i}\right ), \\
s.t. \quad  &  \left ( \frac{ \|{\bf \Omega} \|_{1} }{p} \right )^{p/2} \le M_{1}, \nonumber \\
&  \|{\bf \Omega} \|_{1} \le M_{2}, \nonumber
\end{align}
where $M_{1} \ge 0$ and $M_{2} \ge 0$. This is a commonly used technique in optimization for formulating the constraints (Boyd and Vandenberghe, 2004).
The optimization in \eqref{eq:obj-ub-con} can be written as
\begin{align}\label{eq:obj-ub-con2}
\underset{{\bf \Omega}}{\min} & -\frac{2}{n} |{\bf \Omega}|^{1/2}\sum_{i =1}^{n}\exp\left(-\frac{1}{2} {\bf x}^{T}_{i} {\bf \Omega  x}_{i}\right ), \\
s.t. &  \|{\bf \Omega} \|_{1} \le M_{*}, \nonumber
\end{align}
where $M_{*} = \min\{ M_{1}^{-p/2}, M_{2}\}$ is a tuning parameter, which will be determined by proper cross-validation as described in Section 2.4.
We remark that the problems (2.5) and (2.8) are not equivalent since the matrix determinant constraint in (2.6) is approximated by the linear constraint via the use of Lemma 1.
The underlying reasoning for approximating (2.5) by (2.8) is to mitigate the nonlinear constraint of the matrix determinant term.

From the fact that the logarithm is a monotonically increasing function, solving \eqref{eq:obj-ub-con2} is equivalent to minimizing the objective function in \eqref{eq:obj-wgl-1st} as
\begin{align}\label{eq:obj-wgl-1st}
& \min_{{\bf \Omega}}   - \log \left [  |{\bf \Omega}|^{1/2} \times \frac{1}{n} \sum_{i =1}^{n}\exp\left(-\frac{1}{2} {\bf x}^{T}_{i} {\bf \Omega x}_{i}\right) \right ] \\
& s.t. \quad  \|{\bf \Omega} \|_{1} \le  M_{*}. \nonumber
\end{align}
Note that the objective function \eqref{eq:obj-wgl-1st} is to take logarithm of the objective function in \eqref{eq:obj-ub-con2} up to some constant.
Here the $l_1$ constraint is to encourage the sparsity for inverse covariance matrix estimate (Tibshirani, 1996; Yuan and Lin, 2007).

To shed some lights on the robustness of the proposed method, we note that since the pdf of a normal distribution is proportional to $|{\bf \Omega}|^{1/2}\exp\left(-\frac{1}{2} {\bf x}^{T} {\bf \Omega x} \right)$ up to some constant, the proposed method by solving \eqref{eq:obj-wgl-1st} essentially attempts to maximize the summation of the probabilities. In contrast, the traditional likelihood-based approach attempts to maximize the product of the probabilities. As an objective function, the summation of probabilities can be  more resistant than the product to outliers.

Converting the constraint in \eqref{eq:obj-wgl-1st} to a penalty term in the minimization, we have
\begin{align}\label{eq:obj-wgl-2nd}
\min_{\bf \Omega} \left [- \frac{1}{2}\log |{\bf \Omega}| - \log \left[ \frac{1}{n} \sum_{i  =1}^{n} \exp\left(-\frac{1}{2} {\bf x}_{i}^{T} {\bf \Omega  x}_{i}\right)\right] + \rho \|{\bf \Omega}\|_{1}\right ].
\end{align}
To further facilitate the computational advantage of the proposed method, we adopt a first-order approximation (Zou and Li, 2008)
to the second term in \eqref{eq:obj-wgl-2nd}, denoted as $g({\bf \Omega})$, with respect to the initial estimate $\bm \Omega^{(0)}$. Taking the first-order Taylor expansion on $g({\bf \Omega})$ leads to $g({\bf \Omega}) = g({\bf \Omega}^{(0)}) +  \frac{1}{2n} \sum_{i = 1}^{n} w_{i} {\bf x}_{i}^{T}  {\bf \Omega} {\bf x}_{i} -\frac{1}{2n} \sum_{i = 1}^{n} w_{i} {\bf x}_{i}^{T}  {\bf \Omega}^{(0)} {\bf x}_{i}$. %Such a technique can be viewed as a majorization-minimization approach, which operates by iteratively minimizing a surrogate function that majorizes the objective function (Hunter and Lange, 2004).
Such an approximation strategy leads us to the WGLasso optimization problem \eqref{eq:obj-wgl-3rd}.

\subsection{Asymptotic Property}\label{sec:asymptotic}
We derive an asymptotic property of the proposed estimator, which is analogous to the property for the likelihood-based estimator with $l_1$ penalty in Yuan and Lin (2007).
For simplicity as in Yuan and Lin (2007), we assume that $p$ is fixed as the sample size $n \rightarrow \infty$.
\begin{theorem}\label{thm:consistency}
	Let $\hat{\boldsymbol{\Omega}}$ denote the final output of Algorithm \ref{wgl} with an initial estimator ${\bf \Omega}^{(0)}$ which is a consistent estimator of $\boldsymbol{\Omega}$.
	If $\sqrt{n}\rho \rightarrow \rho_0 \geq 0 $ as $n \rightarrow \infty$,
\begin{equation*}
	\sqrt{n}(\hat{\boldsymbol{\Omega}}-\boldsymbol{\Omega}) \rightarrow_d \mbox{argmin}_{\textbf{U}} ~ V(U),
\end{equation*}
	where $V(\textbf{U})=\mbox{tr}(\textbf{U}\boldsymbol{\Sigma} \textbf{U} \boldsymbol{\Sigma})+\mbox{tr}(\textbf{U}\textbf{Q})+\rho_0\sum_{i,j}\{ u_{ij}\mbox{sign}(c_{ij})I(c_{ij}\neq 0)+|u_{ij}|I(c_{ij}=0) \}$, where $\textbf{Q}$ is a random symmetric $p \times p$ matrix such that $vec(\textbf{Q})\sim N(0, \boldsymbol{\Lambda})$, and $\boldsymbol{\Lambda}$ is such that $cov(q_{ij}, q_{i'j'})=\left(\frac{2}{\sqrt{3}}\right)^p cov(X^{(i)}X^{(j)}, X^{(i')}X^{(j')})$.
\end{theorem}
\noindent Theorem \ref{thm:consistency} implies the $\sqrt{n}$-consistency of $\hat{\boldsymbol{\Omega}}$.
The proof of Theorem \ref{thm:consistency} is provided in Appendix \ref{sec:app1}.
We remark that there can be other ways to prove the above result.
For instance, the above asymptotic property can be achieved by using a similar technique as that of Lam and Fan (2009).

It is also worth to pointing out that, due to the non-convexity of the objective function in (2.1), there would be multiple local optima. Different settings of initial estimate or threshold value $\delta$ could
potentially lead to the algorithm numerically converging to different estimates or local optima.
The asymptotic result here is established for a certain local optimum.

\subsection{Selection of Tuning Parameter}\label{sec:selection}
Note that $\rho$ is a tuning parameter in the proposed method.
Common methods for tuning parameter selection include cross-validation (Finegold and Drton, 2011), a validation-set approach (Levina et al., 2008)  and information-based criteria such as the Bayesian information criterion (Yuan and Lin, 2007).
Since the proposed method is based on a nonparametric criterion, the cross-validation or validation-set approaches are more appropriate in our situation.
However, the conventional cross-validation assumes data in every fold follow the same distribution. This assumption may no longer be valid with outliers presented in the data.
When randomly partitioning the data into $k$ folds for cross-validation, outlier observations may not be uniformly allocated into each fold.

To address this issue, we adopt a revised cross-validation so that the observations in every fold are more likely to be uniformly distributed with respect to the likelihood.
Specifically, we sort the observations based on the value of their likelihood functions with respect to ${\bf \Omega}$.
Suppose that there are $n$ observations for $k$-fold cross-validation.
The ordered observations are grouped into $\lceil n/k \rceil$ blocks, each of which contains $k$ observations.
Then every fold for cross validations is formed by randomly drawing one observation from each block without replacement.
Consequently, the data points in each fold tend to be uniformly distributed in terms of the likelihood values. For the cross-validation score, the values of ISE in \eqref{eq:L2E-loss} are used for the proposed method. The optimal tuning parameter is the one associated with the smallest ISE value.

\section{Simulation} \label{sec:simulation}
To assess the performance of the proposed method,
we conduct a set of simulation studies to compare the proposed method with several existing methods for estimating the inverse covariance matrix.
In particular, we compare the proposed estimates with the LW estimator (Ledoit and Wolf, 2004), glasso (Friedman et al., 2008), and tlasso (Finegold and Drton, 2011), spearman (\"{O}llerer and Croux 2015) and cellweise (Tarr et al., 2016).

Ledoit and Wolf (2004) proposed to estimate $\mathbf \Sigma$ by a linear combination of the sample covariance  matrix $\bf S$ and the identity matrix $\bf I$, i.e.,
\begin{align*}
\hat{\bf \Sigma}_{LW} = \nu \tau {\bf S} + (1-\nu){\bf I},
\end{align*}
where the optimal values of $\tau$ and $\nu$ are obtained by minimizing the expected quadratic loss $E(\|\hat{\bf \Sigma}_{LW} - {\bf \Sigma} \|_2)$.
The corresponding estimator of $\mathbf \Omega$ is $\hat{\bf \Omega}_{LW} = \hat{\bf \Sigma}_{LW}^{-1}$.
Here we include the LW estimator in comparison because it is not a likelihood-based approach; therefore, it can be robust to the model assumption.
The glasso estimates ${\bf \Omega}$ by solving
\begin{align*}
\underset{{\bf \Omega}}{\min} \quad -\log |{\bf \Omega}| + \tr[{\bf \Omega}{\bf S}] + \lambda \sum_{i \ne j} |v_{ij}|,
\end{align*}
where ${\bf \Omega} = (v_{ij})_{p\times p}$ and $\lambda$ is a tuning parameter.
The tlasso estimates ${\bf \Omega}$ by modeling the data with a multivariate $t$-distribution and solving
\begin{align*}
\underset{{\bf \Omega}}{\min} \quad \frac{n}{2}\log |{\bf \Omega}|-\frac{1}{2}\tr \left({\bf \Omega}\sum_{i=1}^n\tau_i {\bf y}_i {\bf y}_i^T \right )+{\bm \mu}^T{\bf \Omega}\sum_{i=1}^n\tau_i {\bf y}_i-\frac{1}{2}{\bm \mu}^T{\bf \Omega} \bm \mu \sum_{i=1}^n\tau_i,
\end{align*}
where ${\bf y}_i$'s are observed data following a multivariate $t$-distribution $t_{p,\nu}(\bm \mu, \bm \Omega^{-1})$ and $\tau_{i}$'s are an associated sequence of hidden gamma random variables assigned to the observations. An EM algorithm is applied to minimize the objective function to obtain the estimator. The tuning parameters in glasso and tlasso are selected using cross-validation.
The R codes of tlasso are provided by the authors of tlasso (Finegold and Drton, 2011).

Another widely studied class of robust and sparse precision matrix estimator (\"{O}llerer  and Croux, 2015, Liu et al., 2012; Tarr et al., 2016) is to take a robust correlation/covaraince matrix  as an initial, such as spearman correlation matrix, to convert into robust and sparse precision matrix via regularization routines, such as the graphical lasso, QUIC and CLIME.

We evaluate and compare the performance of the proposed method and existing methods via (i) {\it selection accuracy} measured by  $F_1$ score ($F_{1}$),
and (ii) {\it estimation accuracy} measured by Frobenius norm ($Fnorm$) and Kullback-Leibler loss ($KL$).
Specifically, the $F_{1}$ score, a measure of selection accuracy of $\bm \Omega$, is defined as
\begin{eqnarray*}
F_1=\frac{2PR}{P+R},
\end{eqnarray*}
where $P={tp}/{(tp+fp)}$ is the precision and $R={tp}/{(tp+fn)}$ is the recall. The larger $F_1$ score implies more selection accuracy for the estimate.
Here $tp$ is the true positive, $fp$ is the false positive, and $fn$ is the false negative (Davis and Goadrich, 2006).
The $Fnorm$ is a metric to measure the distance between the estimated and true inverse covariance matrices, and is defined as
\begin{eqnarray*}
Fnorm=\|\hat{\bm \Omega}-{\bm \Omega}\|_F=\sqrt{\sum_{i,j}(\hat{c}_{ij}-c_{ij})^{2}},
\end{eqnarray*}
where ${\bf \Omega} = (c_{ij})_{p\times p}$.
The $KL$ loss is a likelihood-based measurement for the accuracy of the estimate, and is defined as
\begin{eqnarray*}
KL=-\log(|\hat{\bm \Omega}|)+\tr(\hat{\bm \Omega}{\bm \Omega}^{-1})-(-\log(|\hat{\bm \Omega}|)+p).
\end{eqnarray*}
Clearly, smaller $Fnorm$ and $KL$ loss indicate better estimation accuracy for the estimate.

\subsection{Estimation Accuracy Evaluation}
In this simulation, we consider that observations with sample size $n_{1}$ are generated from a multivariate normal $N_{p}({\bf 0},{\bf \Omega}^{-1})$
but they are contaminated by $n_2$ outlier observations generated from $N_{p}(\bm \mu, \bm I)$.
That is, the total sample size of the data set is $n = n_{1}+n_{2}$. Here we fix $n_1 = 50$ and vary the outlier-to-signal ratios $\gamma \equiv \frac{n_{2}}{n_1}= 0\%$, $6\%$, and $10\%$.
Two different cases of the dimensionality $p$ are considered: $p=55$ and $p=100$.
Moreover, we investigate the performance of the proposed method under two different scenarios of mean shifts: $\bm \mu = {\bf 2} $ and ${\bf 5}$.
Three inverse covariance matrices for ${\bf \Omega}$ are considered as follows.

{\bf Model 1.}
Identity matrix, i.e., ${\bf \Omega} = {\bm I}$.
Denote by M1-I and M1-II the models under the mean shifts of ${\bm \mu} = {\bm 2}$ and ${\bm \mu} = {\bm 5}$, respectively.

{\bf Model 2.}
Autoregressive covariance $AR(1)$, i.e., ${\bf \Omega} = (c_{ij})_{p\times p}$ with $c_{ii}=1$, $c_{i+1,i}=c_{i,i+1}=0.2$, $c_{ij}=0,$ where
$|i-j|>1$.
Denote by M2-I and M2-II the models under the mean shifts of ${\bm \mu} = {\bm 2}$ and ${\bm \mu} ={\bm 5}$, respectively.

{\bf Model 3.}
Randomly permuted $AR(1)$, i.e., ${\bf \Omega} = {\bm Q \bm M \bm Q}^{T}$, where $\bm M$ is the inverse covariance matrix defined in Model 2, and $\bm Q$ is a $p \times p$ matrix
obtained by randomly permuting the rows of a $p \times p$ identity matrix.
Denote by M3-I and M3-II the models under the mean shifts of $\bm \mu = \bm 2$ and $\bm \mu = \bm 5$, respectively.

The numerical results based on $100$ simulation runs are reported in Tables \ref{tab:non-out}--\ref{tab:10out}, corresponding to $0\%$, $6\%$, and $10\%$ outlier-to-signal ratios, respectively.
The standard errors are in parentheses.
Overall, the simulation results show that our proposed method outperforms the other five methods in most cases when outliers are present.

Table \ref{tab:non-out} reports the situation of no-outlier observations. From the results in Table \ref{tab:non-out},
it can be seen that the performance of our proposed method is comparable to the glasso method in terms of the selection accuracy and estimation accuracy. Compared with the LW method, the proposed method gives better selection accuracy, while the LW method provides best estimation accuracy in $Fnorm$ and $KL$ loss among all methods. The tlasso method performs well in selection accuracy but worst in estimation accuracy among all methods. The two robust initialization methods (spearman and cellwise) have highest selection accuracy in most cases.

\begin{table}
	\begin{center}
		{ \caption{Simulation results under outlier-to-signal ratio = $0\%$}
			\label{tab:non-out}
			\footnotesize   \begin{tabular}{ cc | ccc| ccc |ccc  } % creating 27 columns
				\hline \hline & & \multicolumn{3}{c|}{Model 1} & \multicolumn{3}{c|}{Model 2} & \multicolumn{3}{c}{Model 3} \\
				[0.5ex] $p$ & Method & $F_{1}$ & $Fnorm$ & $KL$ & $F_{1}$ & $Fnorm$ & $KL$
				& $F_{1}$ & $Fnorm$ & $KL$ \\
				\hline
				&WGLasso  & 0.31 & 8.52& 6.11 & 0.31 & 8.59 & 6.13 & 0.38 & 4.03 & 3.81 \\
				&          &(0.04) &(0.62) & (0.35)  &(0.03) &(0.66) & (0.39) & (0.07)&(0.64)&(0.40)   \\
				&glasso    & 0.32 & 7.16&  3.94 & 0.33 & 7.48 & 4.04 & 0.27 & 3.24 & 2.04    \\
				&          &(0.07) & (0.92)& (0.56)  &(0.07) &(1.41) & (0.80) &(0.11) &(0.84) & (0.45)  \\
				55      &tlasso      & 0.53 & 40.10 & 59.54 & 0.44 & 45.10 & 60.89 & 0.45 & 45.13 & 60.93 \\
				&          &(0.07) & (0.36)& (1.09)  &(0.04) &(0.43) & (1.30) &(0.03) &(0.32) & (1.07)  \\
				&LW       & 0.10 &  {\bf 4.30} &  {\bf 2.34} & 0.10 &  {\bf 4.31} &  {\bf 2.35} & 0.04 &  {\bf 0.27} &  {\bf 0.13}  \\
				&          &(0.00) &(0.19) &(0.11)  &(0.00) &(0.16) &(0.11)  &(0.00) &(0.23) &(0.11)   \\
				& spearman & {\bf 0.56} & 4.40 & 3.00 &  {\bf 0.45} & 7.03 & 3.56 &  {\bf 0.47} & 7.24 & 3.63 \\
				&          & (0.13) & (0.80) & (0.59) & (0.04) & (0.69) & (0.32) & (0.04) & (0.67) & (0.31) \\
				& cellwise & 0.52 & 4.92 &3.34  &  {\bf 0.45} & 11.21 & 5.88 & 0.45 & 11.10 & 5.81 \\
				&           & (0.14) & (0.86) & (0.62) & (0.05) & (1.15) & (0.68) & (0.05) & (1.09) & (0.65) \\ \hline
				&Proposed  & 0.24 & 15.96 & 10.66& 0.25 & 16.04 & 10.70& 0.30 & 7.06 & 6.62 \\
				&    & (0.04)         &(1.42)  &(0.79) & (0.04)  &(1.12) &(0.64)& (0.06)  & (0.92) & (0.81)  \\
				&glasso   & 0.29 & 14.03 & 7.46 & 0.29 & 13.95 & 7.48 & 0.29 & 5.76 & 3.80  \\
				&          &(0.06) &(1.23) & (0.56)  &(0.07) &(1.12) & (0.49) & (0.10)&(0.75) &(0.52)  \\
				100      &tlasso   & 0.51 & 84.95 & 164.60 &  {\bf 0.50} & 95.44 & 178.20&  {\bf 0.50} & 95.45 &  178.23 \\
				&          &(0.24) &(1.67) &(10.81)   & (0.02)& (0.13)& (1.07) & (0.02)&(0.13) & (1.12) \\
				&LW       & 0.06 &  {\bf 8.16} &  {\bf 4.44} & 0.06 &  {\bf 8.15} &  {\bf 4.43} & 0.06 &  {\bf 0.45} &  {\bf 0.22}  \\
				&          &(0.00) &(0.23) & (0.16)  &(0.00) &(0.24) &(0.16)  &(0.19) &(0.33) & (0.16)\\
				& spearman & 0.57 & 8.74 & 5.96 & 0.44 & 14.07 & 6.98 & 0.44 & 14.19 & 7.04\\
				&          & (0.10) & (1.21) & (0.92) & (0.03) & (1.06) & (0.52) & (0.03) & (1.03) & (0.51) \\
				& cellwise &  {\bf 0.66} & 12.77 & 8.90 &  0.45 & 25.40 & 13.26 &  0.46 & 24.58 & 13.37 \\
				&          & (0.10) & (1.25) & (0.97) & (0.03) & (1.59) & (1.14) & (0.03) & (1.40) & (0.99) \\
				\hline
				\hline
			\end{tabular}
		}
	\end{center}
\end{table}
%\end{sidewaystable}

%\begin{sidewaystable}
\begin{table}
	\begin{center}
		{\caption{Simulation results under outlier-to-signal ratio = $6\%$}
			\label{tab:6out}
			\tiny   \begin{tabular}{ cc | ccc| ccc |ccc  } % creating 27 columns
				\hline \hline & & \multicolumn{3}{c|}{M1-I} & \multicolumn{3}{c|}{M2-I} & \multicolumn{3}{c}{M3-I} \\
				[0.5ex] $p$ & Method & $F_{1}$ & $Fnorm$ & $KL$ & $F_{1}$ & $Fnorm$ & $KL$
				& $F_{1}$ & $Fnorm$ & $KL$ \\
				\hline
				&WGLasso  & {\bf 0.31} & 8.78 & 6.26 & {\bf 0.31} & {\bf 8.72} & 6.24 & {\bf 0.39} & {\bf 4.10} & {\bf 3.92}  \\
				&          &(0.03) &(1.10) & (0.58)  &(0.04) &(0.85) & (0.49) & (0.06)&(0.76)&(0.48)   \\
				&glasso    & 0.16   & 10.21 & 6.62 & 0.16  & 9.93 & 6.41 & 0.12  & 5.61 & 4.43 \\
				&          &(0.01) & (1.37)& (0.80)  &(0.01)  &(0.77) & (0.52) &(0.01) &(1.21) & (0.51)  \\
				55      &tlasso     & 0.28 & 40.11 & 61.37     & 0.22  & 45.02 & 63.06   & 0.21  & 44.68  &62.24  \\
				&          &(0.02) & (0.10)& (0.28)  &(0.01) &(0.11) & (0.31)  &(0.01) &(0.13)  &(0.31)  \\
				&LW       & 0.10  & 8.81 & 5.76 & 0.10  & 8.80 & {\bf 5.74}  &   0.04 & 8.19 & 4.89 \\
				&          &(0.00) &(0.76) &(0.51) &(0.00) &(0.77) & (0.50) &(0.00) & (1.63)&(0.68)   \\
				& spearman & 0.12 & 9.16 & 4.29 & 0.16 & 20.80 & 9.92 & 0.16 & 21.06 & 10.06 \\
				&          & (0.01) & (4.56) & (1.47) & (0.01) & (6.01) & (2.22) & (0.01) & (5.63) & (2.08) \\
				& cellwise & 0.11 & {\bf 6.42} & {\bf 5.04} & 0.16 & 11.83 & 7.77  & 0.16 & 11.49 & 7.55 \\
				&          & (0.01) & (1.13) & (0.61) & (0.01) & (2.68) & (1.49) & (0.01) & (2.49) & (1.42) \\
				\hline
				&WGLasso  & 0.25  & 16.10 & 10.80 & 0.25 & {\bf 16.20} & {\bf 10.80} & 0.30 & {\bf 7.10} & {\bf 6.50}\\
				&    & (0.04)         &(1.27)  &(0.68) & (0.04)  &(1.31) &(0.73)& (0.06)  & (1.36) & (0.77)  \\
				&glasso  & 0.12 & 17.90 & 11.41& 0.12 & 17.83 & 11.38& 0.10 & 9.21 & 7.52 \\
				&          &(0.01) &(0.86) & (0.36)  &(0.01) &(0.88) & (0.35) & (0.01)&(1.25)  &(0.56)  \\
				100      &tlasso     &{\bf 0.61}   &87.25  &180.61   &{\bf 0.41}  &95.76 &180.87  &{\bf 0.41}  &95.75 &180.94    \\
				&          &(0.02) &(0.04) &(0.19)   &(0.01)&(0.02)&(0.11)  &(0.01)&(0.02)&(0.12)  \\
				&LW     &   0.06 & 20.95 & 14.06 &  0.06 & 20.75 & 13.92 &   0.02 & 23.14 & 12.55 \\
				&          &(0.00) &(1.91) & (1.03) &(0.00) &(2.10) &(1.20)  &(0.00) &(3.33)& (1.18)\\
				& spearman & 0.09 & 16.80 &{\bf 8.25} & 0.12 & 26.13 & 14.44 & 0.12 & 26.21 & 14.47 \\
				&          & (0.00) & (1.47) & (0.54) & (0.00) & (1.06) & (0.44) & (0.00) & (1.12) & (0.45) \\
				& cellwise &  0.10  & {\bf 12.21} & 9.98 & 0.13 & 21.67 & 13.69 & 0.13 & 21.66 & 13.68 \\
				&          & (0.00) & (0.71) & (0.60) & (0.00) & (0.85) & (0.57) & (0.00) & (0.86) & (0.56) \\
				\hline
				& & \multicolumn{3}{c|}{M1-II} & \multicolumn{3}{c|}{M2-II} & \multicolumn{3}{c}{M3-II} \\
				[0.5ex]   &          & $F_{1}$ & $Fnorm$ & $KL$ & $F_{1}$ & $Fnorm$ & $KL$
				& $F_{1}$ & $Fnorm$ &$KL$  \\
				\hline
				&WGLasso  & {\bf 0.30} & 9.67 & 6.86 & {\bf 0.30} & {\bf 9.69} & {\bf 6.87} & {\bf 0.37} & {\bf 4.49} & {\bf 4.50}\\
				&          &(0.03) &(1.15) & (0.76)  &(0.04) &(1.09) &(0.71)   &(0.07) &(0.78) &(0.74)   \\
				&glasso    & 0.13 & 10.95 & 8.98 & 0.13  & 10.74  & 8.82  & 0.10& 6.44 & 6.76  \\
				&          &(0.02) &(1.27) &(0.66)   & (0.02)& (0.75)&(0.42)  &(0.01) & (1.35)&(0.50)   \\
				55      &tlasso      &0.12   &38.72  &61.27   &0.16   &43.49  &62.36     &0.15   &43.48  &62.30\\
				&          &(0.02) &(0.20) &(0.51)  &(0.00) &(0.22) &(0.50)    &(0.01) &(0.21) &(0.50)\\
				&LW       &   0.10 &    11.28 & 9.83 &   0.10  & 11.25 & 9.79 & 0.04 & 8.64& 8.56 \\
				&          &(0.00) & (0.25) & (0.18)  &(0.00) &(0.30) &(0.22)  & (0.00)&(0.20) &(0.19)  \\
				& spearman & 0.10 & 20.52 & 7. 63 & 0.15 & 31.20 & 13.71 & 0.15 & 31.98 & 13.98 \\
				&         & (0.01) & (6.50) & (1.96) & (0.01) & (3.16) & (1.04) & (0.01) & (2.26) & (0.68) \\
				& cellwise & 0.11 & {\bf 6.92} & {\bf 5.81} & 0.16 & 11.81 & 8.02 &0.16 & 11.94 & 8.03 \\
				&          & (0.01) & (0.48) & (0.43) & (0.01) & (0.77) & (0.56) & (0.01) & (0.99) &(0.66) \\
				\hline
				&WGLasso & {\bf 0.25}& 17.60 & 11.60 & {\bf 0.25}& {\bf 17.94} & {\bf 11.90}& {\bf 0.29} & {\bf 8.02} & {\bf 7.49}  \\
				&          &(0.03) &(1.98) & (52.38)  &(0.04) & (2.10)& (1.30) & (0.06)& (1.40)& (1.21)  \\
				&glasso   & 0.10 & 18.73 & 14.14& 0.10 & 18.91 & 14.19& 0.09  & 10.15 & 10.18   \\
				&          &(0.01) &(1.00) & (0.58)  & (0.01)&(1.00) & (0.50) &(0.01) &(1.30) & (0.59) \\
				100      &tlasso     &0.10   &85.69  &178.39  &0.12   &94.55  &182.42    &0.12   &94.49  &182.01  \\
				&          &(0.01) &(0.17) &(0.90)  &(0.00) &(0.15) &(0.77)    &(0.01) &(0.16) &(0.79) \\
				&LW        &   0.06 & 23.60 & 19.90 &   0.06  & 23.55 & 19.90 &   0.02 & 20.23 & 17.14 \\
				&          &(0.00) &(0.57) &(0.22)   & (0.00)&(0.56) & (0.22) &(0.00) &(0.73) & (0.22)\\
				& spearman & 0.09 & 22.47 & {\bf 9.60} & 0.11 & 37.55 & 18.81 & 0.11 & 36.56 & 18.54 \\
				&          & (0.00) & (1.25) & (0.36) & (0.00) & (12.86) & (4.08) & (0.00) & (10.32) & (3.26)\\
				& cellwsie & 0.10 & {\bf 14.45} & 12.16 &0.13 & 25.20 & 16.39 & 0.12 & 25.03 & 16.29 \\
				&          & (0.00) & (0.80) & (0.73) & (0.00) & (1.56) & (0.91) & (0.00) & (1.55) & (0.88) \\
				\hline
				\hline
			\end{tabular}
		}
	\end{center}
\end{table}
%\end{sidewaystable}

%\begin{sidewaystable}
\begin{table}
	\begin{center}
		{ \caption{Simulation results under outlier-to-signal ratio = $10\%$}
			\label{tab:10out}
			\tiny  \begin{tabular}{ cc | ccc| ccc |ccc  } % creating 27 columns
				\hline \hline & & \multicolumn{3}{c|}{M1-I} & \multicolumn{3}{c|}{M2-I} & \multicolumn{3}{c}{M3-I} \\
				[0.5ex] $p$ & Method & $F_{1}$ & $Fnorm$ & $KL$ & $F_{1}$ & $Fnorm$ & $KL$
				& $F_{1}$ & $Fnorm$ & $KL$ \\
				\hline
				&wGLasso & {\bf 0.30} & 9.01 & 6.42 & {\bf 0.31} & {\bf 9.03} & {\bf 6.46} & {\bf 0.35} & {\bf 4.44} & {\bf 4.27}   \\
				&          &(0.03) &(0.84) & (0.53) &(0.03) &(0.88)& (0.55) & (0.06)&(1.28)&(0.86)   \\
				&glasso   & 0.14 & 10.44 & 7.12 & 0.14 &  10.27 & 7.05 & 0.11  & 5.99 & 5.06 \\
				&          &(0.01) & (1.18)& (0.64)  &(0.01)  &(0.64) & (0.36) &(0.02) &(1.73)& (0.65)  \\
				55      &tlasso     & 0.10 & 39.26 & 60.82 & 0.15 & 44.14 &61.91  & 0.15  & 43.99 &  61.61 \\
				&          &(0.01) & (1.50)& (3.79)  &(0.01) &(1.47) & (3.52) &(0.01) &(1.55) & (3.66)\\
				&LW      &   0.10 &   21.49 & 11.61  &   0.10  & 21.34 & 11.57 &  0.04 & 27.88 & 11.65\\
				&          &(0.00) &(2.76) &(0.87) &(0.00) &(3.16)  & (1.06) &(0.00) & (4.56)&(1.10)   \\
				& spearman & 0.10 & 19.91 & 7.50 & 0.15 & 31.26 & 13.79 & 0.15 & 31.05 & 13.73\\
				&          & (0.01) & (6.53) & (1.96) & (0.01)& (2.16) & (0.64) & (0.01) & (1.99) & (0.63)\\
				& cellwise & 0.11 & {\bf 6.57} & {\bf 5.67} & 0.15 & 11.74 & 8.10  & 0.15 & 11.79 & 8.10 \\
				&          & (0.00) & (0.49) &(0.35) & (0.01) & (1.74) & (1.09) & (0.01) & (1.83) & (1.07) \\
				\hline
				&WGLasso  & {\bf 0.27} & 11.27 & {\bf 8.12}& {\bf 0.26} &  11.71 & {\bf 8.45}& {\bf 0.25} & 6.46 & {\bf 6.47}\\
				&    & (0.03)        &(1.84)  &(1.43) & (0.03)  &(2.21) &(1.98)& (0.05)  & (2.22) & (2.52) \\
				&glasso  & 0.13 & {\bf 10.75} & 9.26 & 0.13  & {\bf 10.67} & 9.24& 0.10 & {\bf 6.39} & 7.19  \\
				&          &(0.02) &(0.90) & (0.46) &(0.01) &(0.60) & (0.33) & (0.01)&(1.52) &(0.54) \\
				100      &tlasso   & 0.09 & 86.88 & 182.64 & 0.12  & 95.25 &183.38 &0.12  & 95.32 &  183.78   \\
				&          &(0.01) &(0.60) &(3.33)  & (0.01)&(0.59)& (3.16) & (0.01)&(0.36) & (1.96)   \\
				&LW    &   0.10 & 18.23 & 13.51 &   0.10  & 18.24 & 13.52 &   0.04 & 17.73 & 12.26 \\
				&          &(0.00) &(0.82) &(0.28) &(0.00) &(0.79) &(0.28) &(0.00) &(0.86)& (0.22)\\
				& spearman &  0.09 & 22.14 & 9.48 & 0.11 & 35.74 & 18.24 & 0.11 & 35.52 & 18.14 \\
				&          & (0.00) & (1.41) & (0.42) & (0.00) & (10.56) & (3.39) & (0.00) & (10.48) & (3.32) \\
				&  cellwise & 0.09 & 13.52 & 11.69 & 0.11 & 23.74 & 15.65 & 0.11 & 23.99 & 15.75 \\
				&           & (0.00) & (0.58) & (0.55) & (0.00) & (1.50) & (0.86) & (0.00) & (1.47) & (0.81)\\
				\hline
				& & \multicolumn{3}{c|}{M1-II} & \multicolumn{3}{c|}{M2-II} & \multicolumn{3}{c}{M3-II} \\
				[0.5ex]   &         & $F_{1}$ & $Fnorm$ & $KL$ & $F_{1}$ & $Fnorm$ & $KL$
				& $F_{1}$ & $Fnorm$ &$KL$  \\
				\hline
				&WGLasso  & {\bf 0.23}  & 20.99 & 13.95& {\bf 0.24}  & 16.69 & 11.11 & {\bf 0.29}  & {\bf 7.47} & {\bf 6.95} \\
				&          &(0.03)&(3.10) & (2.09)  &(0.04) &(1.72) &(0.98)   &(0.03)  &(1.05)  &(0.95)    \\
				&glasso     & 0.10 & 19.21 & 14.76& 0.11 & 18.33 & 12.21    & 0.09  & 9.58 & 8.37 \\
				&          &(0.01) &(3.17) &(1.55)  & (0.01) & (1.08)&(0.56) &(0.01) & (0.87)&(0.56)    \\
				55      &tlasso    & 0.09  & 38.89 &62.09 & 0.14 &44.13  &64.11 &0.15  & 43.83 &63.42 \\
				&          &(0.01) &(1.86) & (4.88) &(0.01) &(1.62) & (4.00) &(0.01)& (1.77)& (4.35) \\
				&LW      &   0.06 & 54.61 & 32.67&   0.06 &    68.18 & 32.22&   0.02  & 93.64 & 33.23 \\
				&          &(0.00) &(2.18)  & (0.52)  &(0.00) &(8.17)&(2.40)  & (0.00)&(11.11) &(2.40)  \\
				& spearman    & 0.11 & 40.12 & 12.74 & 0.12 & 52.21 & 20.41 & 0.12 & 52.24 & 20.39 \\
				&             & (0.00) & (3.02) & (0.69) & (0.00) & (3.51) & (0.96) & (0.00) & (3.54) & (0.96)\\
				& cellwise &  0.11 & {\bf 8.66} & {\bf 7.78} & 0.15 & {\bf 13.89} & {\bf 9.78} & 0.15 & 13.88 & 9.87 \\
				&          & (0.00) & (0.62) & (0.51) & (0.00) & (1.00) & (0.62) & (0.01) & (1.04) & (0.61) \\
				\hline
				&WGLasso & {\bf 0.24} & 21.23 & {\bf 14.06}& {\bf 0.23}  & 21.10 & {\bf 14.02} & {\bf 0.26}  & 10.91 & {\bf 10.06}  \\
				&          &(0.03) &(2.69) & (1.86)  &(0.03) & (3.13)& (2.09)& (0.03) &(2.03)& (1.68)  \\
				&glasso   & 0.10  & 18.87 & 14.59 & 0.10& {\bf 19.16} & 14.73 & 0.09 & {\bf 9.99} & 10.55  \\
				&          &(0.01) &(1.14) & (0.64)  & (0.01)&(3.20)& (1.55) &(0.01)&(0.90) & (0.55) \\
				100      &tlasso    & 0.09  & 86.84 & 185.19 &0.11 &95.32  &186.51  & 0.11  & 95.50 &  187.43 \\
				&          &(0.01) &(1.07)& (6.01)  & (0.01)& (0.82)& (4.32)&(0.01) &(0.37) & (2.02) \\
				&LW        & 0.06 & 54.46 & 32.63&   0.06  & 54.62 & 32.67&   0.02  & 56.82 & 29.25  \\
				&          &(0.00) &(2.09) &(0.49)  & (0.00)&(2.18) & (0.52) &(0.00) &(2.40) & (0.44)\\
				& spearman  & 0.08 & 87.42 & 26.51 & 0.10 & 109.95 & 41.27 & 0.11 &110.50 & 41.37 \\
				&           & (0.00) & (6.71) & (1.63) & (0.00) & (4.86) & (1.32) & (0.00) &(4.89) & (1.31) \\
				& cellwise & 0.09 & {\bf 18.61} & 16.37 & 0.11 &29.27 & 20.22 &0.11 & 29.35 &20.29 \\
				&          & (0.00) & (1.25) & (1.02) & (0.00) & (1.29) & (0.99)  & (0.00) & (1.43) & (1.09) \\
				\hline
				\hline
			\end{tabular}
		}
	\end{center}
\end{table}
%\end{sidewaystable}

Table \ref{tab:6out} shows the results for the $6\%$ outlier-to-signal ratio case.
When the dimension is $p = 55$ and the mean shift is ${\bm \mu} = {\bf 2}$, the proposed method has comparable performance to the LW method in terms of $KL$ loss for Models $1$ and $2$.
However, the $F_{1}$ score of the proposed method is much larger than that of the LW method, indicating its superiority in terms of selection accuracy.
When the dimension $p$ increases, as in the case of $p=100$,
the tlasso method has larger $F_{1}$ scores than the other five methods, but worst prediction accuracy indicted by the $Fnorm$ and $KL$ loss values. Our proposed method has highest prediction accuracy for M2 and M3 while the two robust initialization methods, spearman and cellwise, perform best in prediction accuracy for M1. For the case of large mean shift, i.e., ${\bm \mu} = {\bf 5}$, the proposed method in general outperforms all other methods in both selection and prediction accuracy except for M1 where two robust initialization methods perform best in prediction accuracy.

As an example of which abnormal observations become more significant, the results of the $10\%$ outlier-to-signal ratio case are reported in Table \ref{tab:10out}.
In this case, the proposed method always yields the largest $F_{1}$ score among the six methods under comparison.
For M1 model, the values of $KL$ loss and $Fnorm$ of our method are close to the cellwise method for $p=100$ but the cellwise method performs better when $p=55$. For M2 and M3 models, the overall prediction accuracy of our proposed method is superior to other competing methods.

To further validate our method, we also checked whether the final weights of outliers are much smaller than the weights of the normal observations.
Under the Model 2 setting with $p=55$ and 6\% outliers, the average of final weights of the normal observations and outliers were 1.06 and $7.10\times 10^{-3}$, respectively, over 100 iterations.
This indicates that the normal observations dominate the fitting procedure.

\subsection{Numerical Study of Convergence}
The proposed WGLasso objective function in \eqref{eq:obj-wgl-3rd} is nonconvex and nonsmooth, and the weights are a complicated function of the observations and the inverse covariance matrix.
To confirm the asymptotic properties derived in Section \ref{sec:asymptotic}, we use Monte Carlo simulations to investigate the convergence performance of the proposed estimator.
In the meantime, we also compare the proposed method with the glasso method, of which the objective function is convex and known to converge.
We utilize our previous three numerical simulation settings in which the true inverse covariance matrices are the identity matrix, banded-structured matrix, and random banded-structured matrix. We consider a lower dimension $p=55$ and a higher dimension $p=100$ under three outlier-to-signal ratios $0\%, 6\%$ and $10\%$, where the mean of the outliers is $\bm \mu=\bm 5$.

To study the performance of the two methods as the total sample size $n$ increases, we increase the sample size $n$ while fixing the outlier-to-signal ratio.
As expected in the asymptotic situations, the sample size $n$ will become much larger than the dimension $p$ eventually.
Here we set the largest sample size $n$ to be roughly at the scale of $p^2$.
In particular, for the lower dimension $p=55$, the value of $n$ ranges from 20 to $\mbox{2,500}$; for the higher dimension $p=100$, the value of $n$ ranges from $30$ to $\mbox{10,000}$.
To measure the distance between the estimate and the true inverse covariance matrix, we use the matrix Frobenius norm (Fnorm) as the metric.
All the results are based on $100$ independent replications.

Figures \ref{fig:converge-Identy}--\ref{fig:converge-RandAR1} depict the convergence performance results of the proposed method and the glasso methods under three numerical models, respectively.
It is seen that as the total sample size $n$ increases, the Fnorm distances of the two estimates from the true inverse covariance matrix decrease gradually and become stable in the later stages. Under the no-outlier situation, the glasso method yields a smaller distance to the true inverse covariance matrix than the WGLasso does. In the cases of positive outlier-to-signal ratios, the average distance between the WGLasso estimate and the true inverse covariance matrix is smaller than that of the glasso estimate. This suggests that the WGLasso estimate converges to a better position in the neighborhood of the true inverse covariance matrix when the data are noisy. The glasso estimate has a sharper decrease in the Fnorm value than the WGLasso estimate, especially at the beginning stages where the total sample size $n$ is much smaller than the dimension $p$. This illustrates the robustness of the WGLasso in a high-dimensional situation.

\begin{figure}
	\centering
	\begin{tabular}{p{1.8cm}|CC} \hline \hline
		outlier-to-signal ratio&$p=55$ & $p=100$  \\ \hline
		$0\%$&\subfloat{\includegraphics[width=6.7cm,height=6cm]{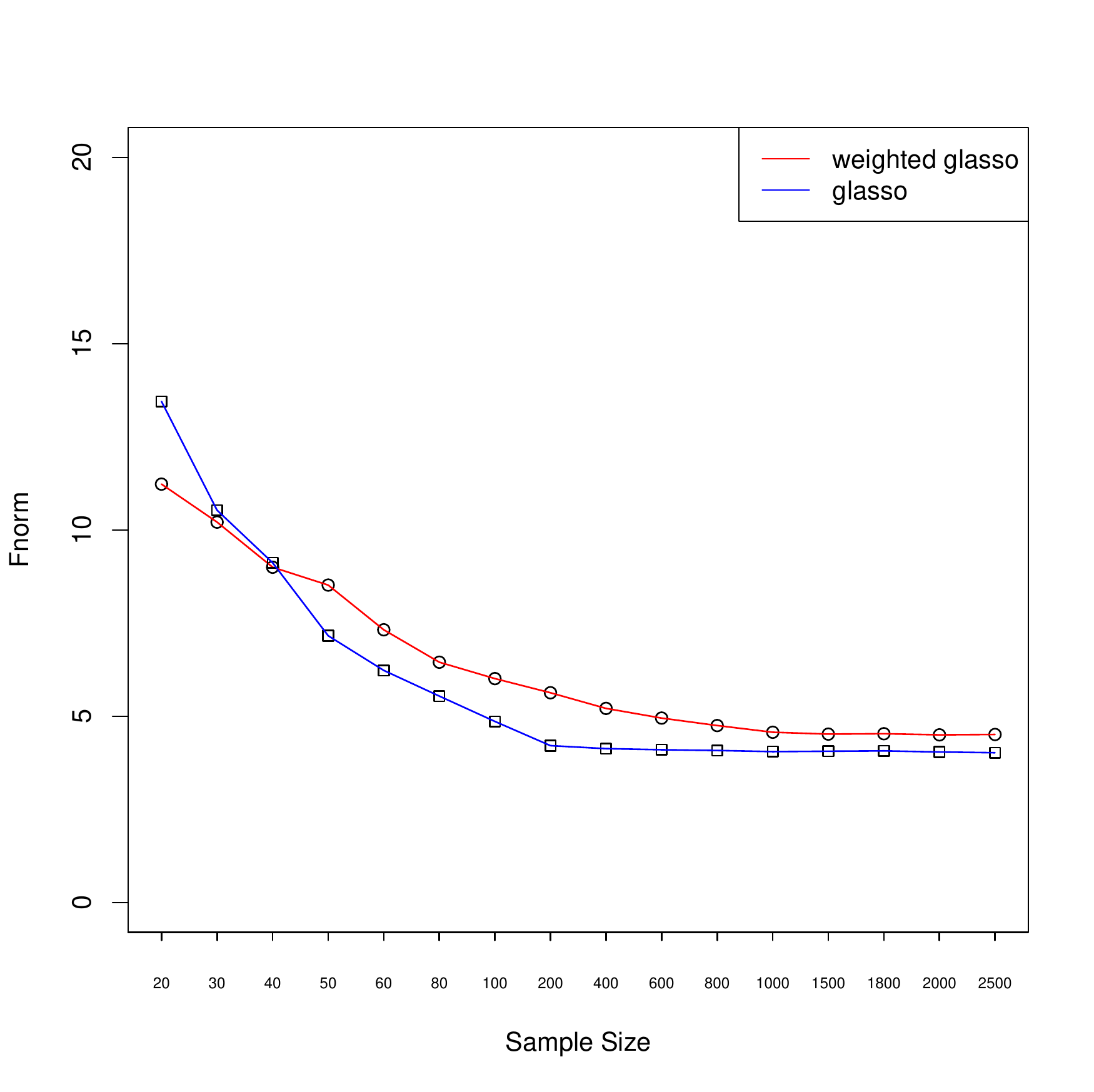}} &\subfloat{\includegraphics[width=6.7cm,height=6cm]{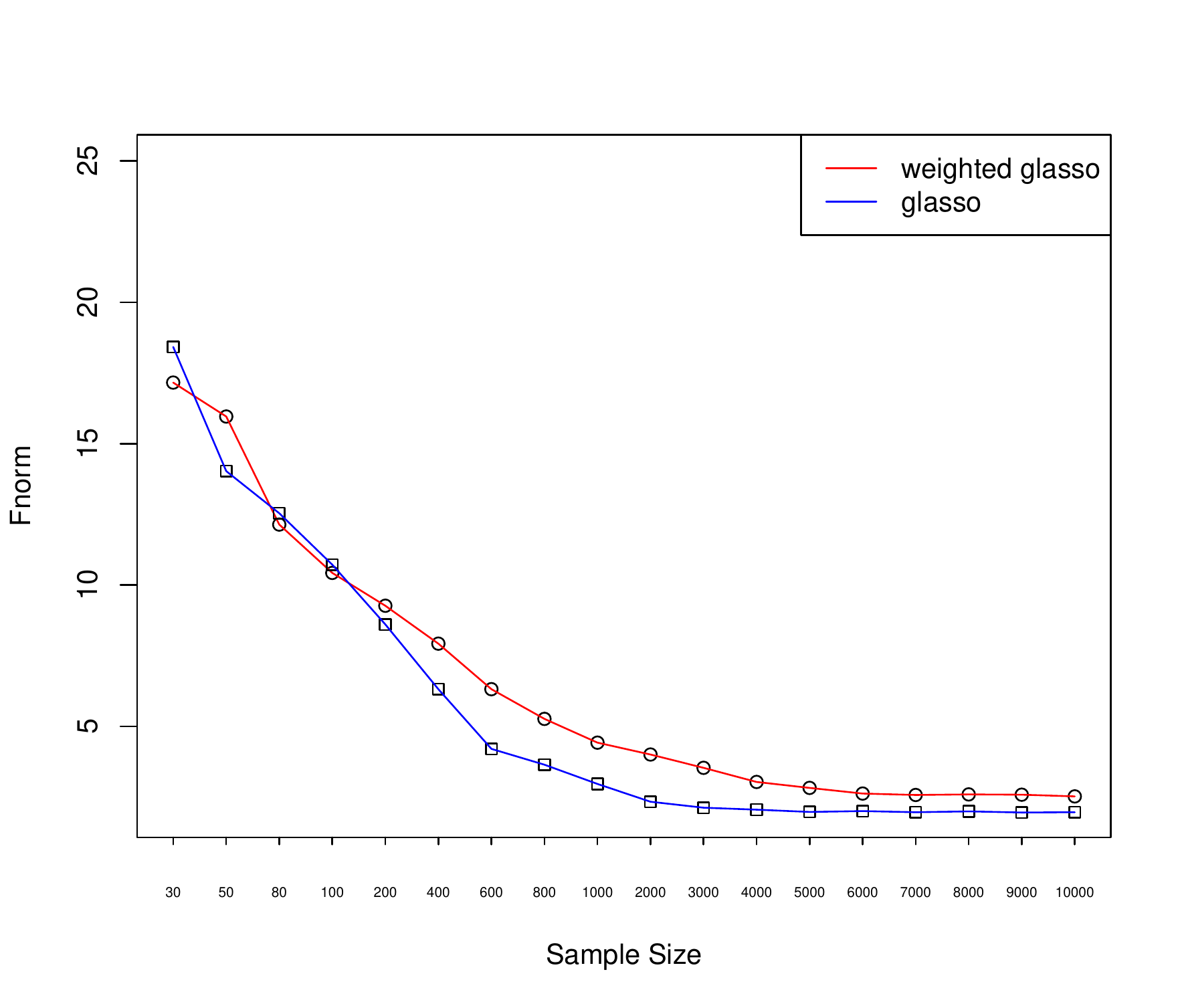}} \\ \hline
		$6\%$ & \subfloat{\includegraphics[width=6.7cm,height=6cm]{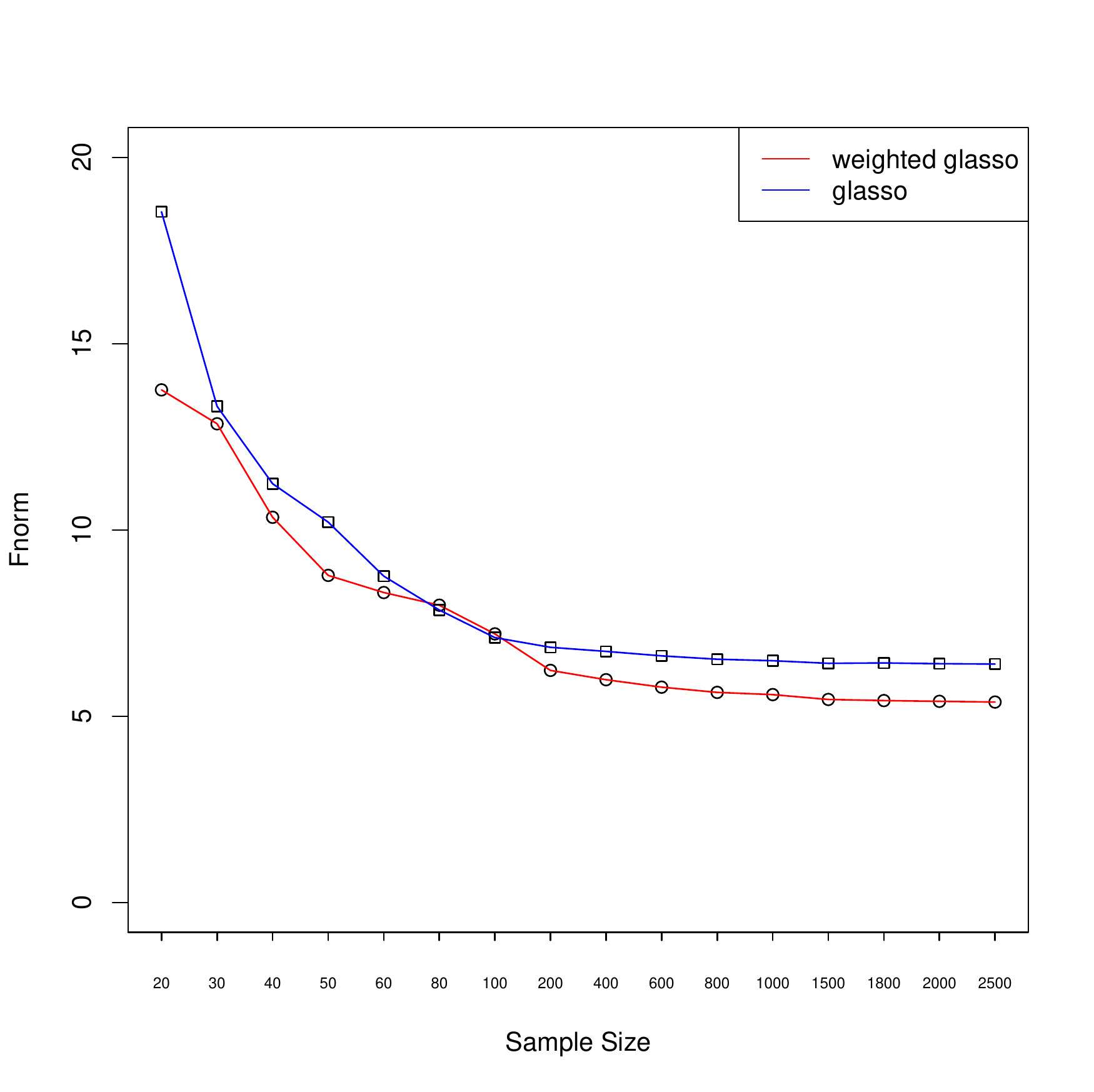}}&\subfloat{\includegraphics[width=6.7cm,height=6cm]{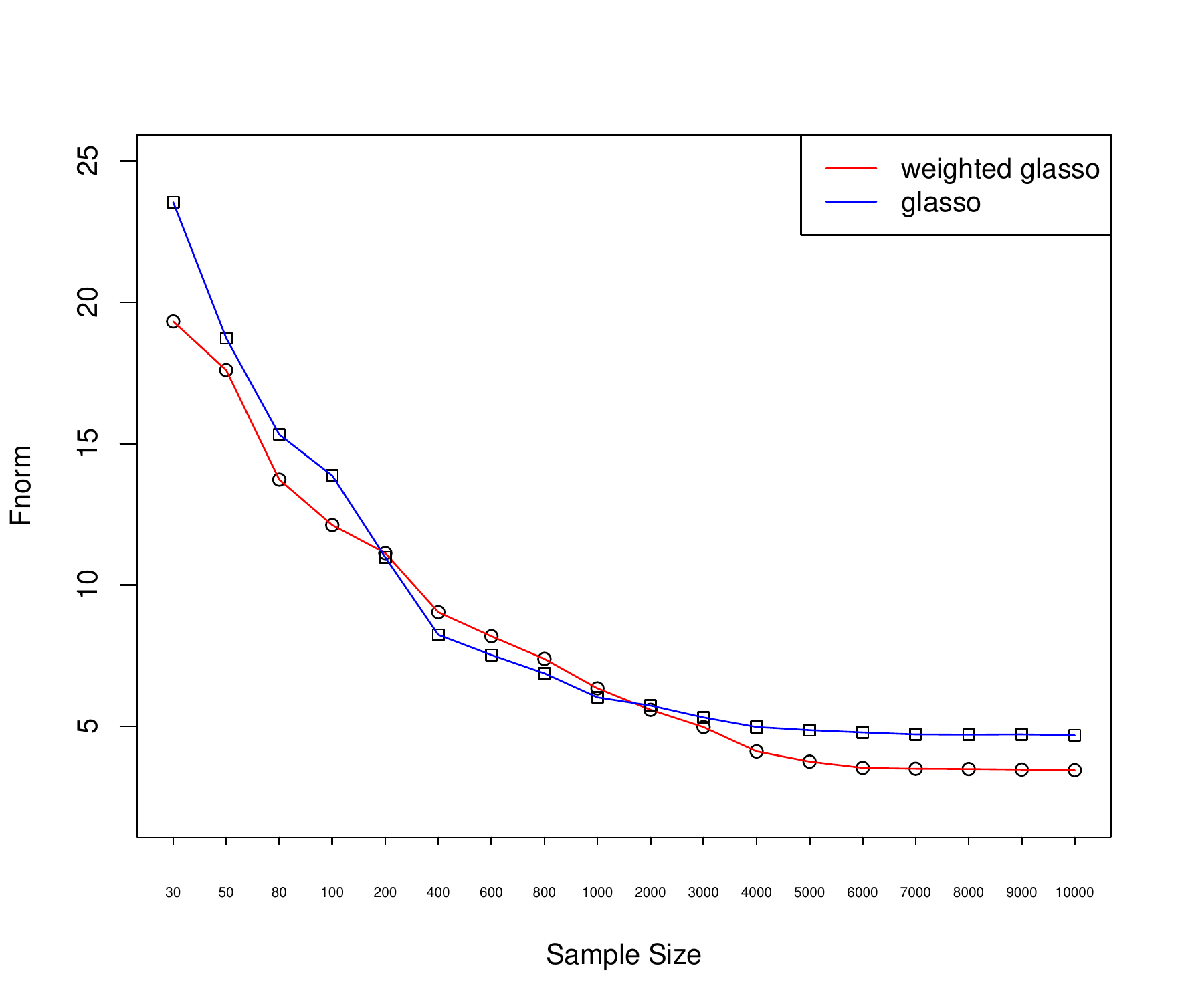}}\\ \hline
		$10\%$ & \includegraphics[width=6.7cm,height=6cm]{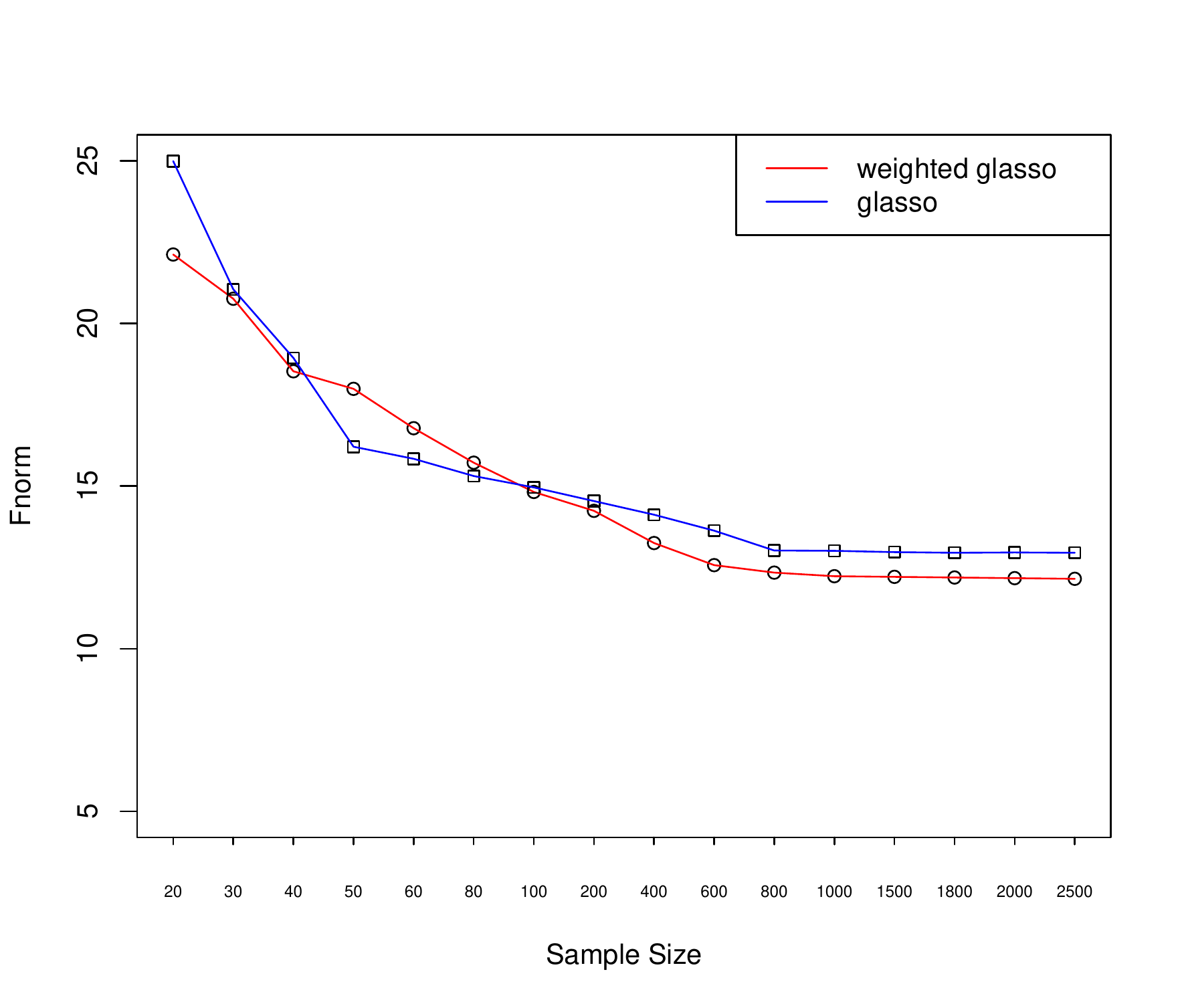}&\subfloat{\includegraphics[width=6.7cm,height=6cm]{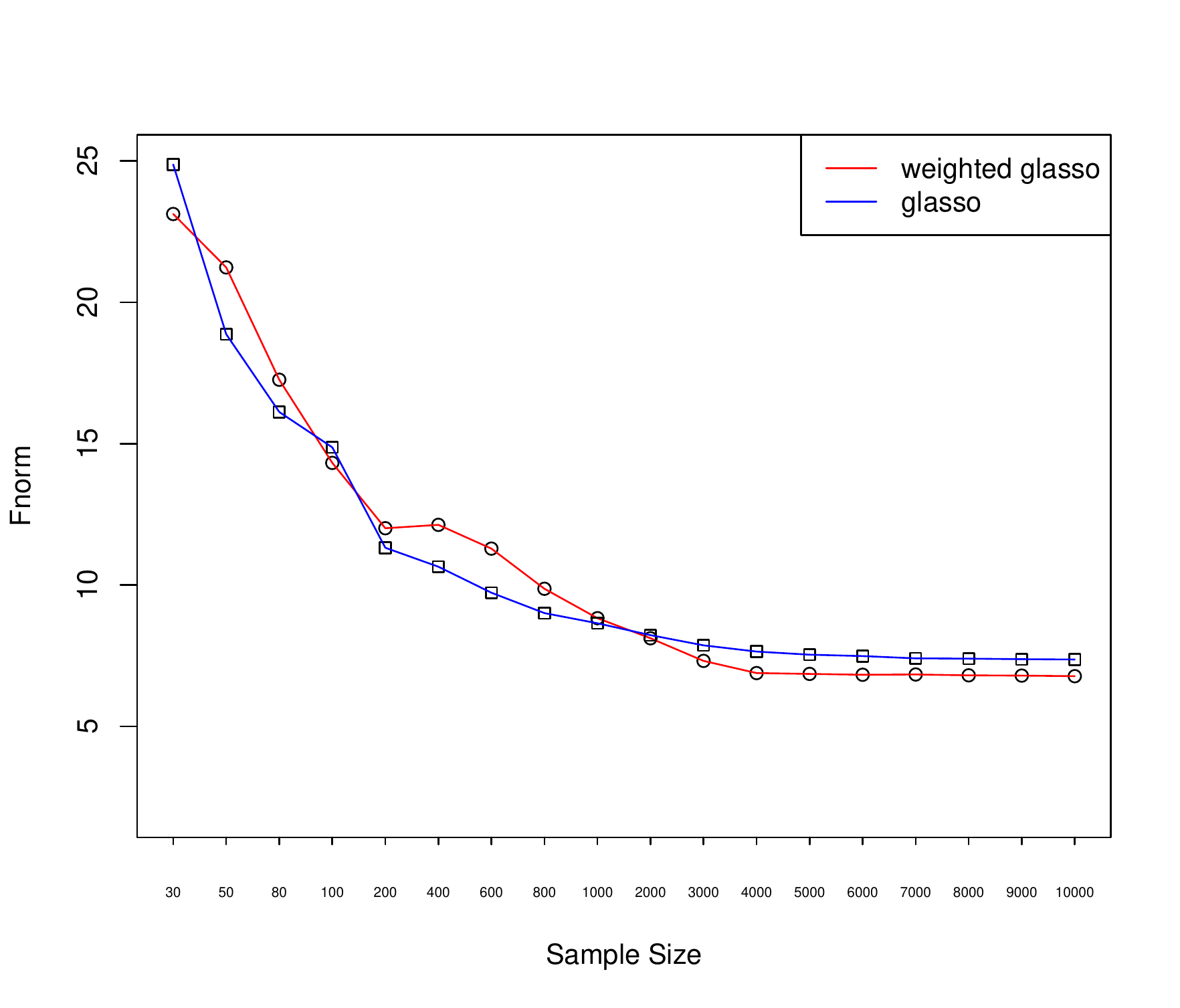}}\\ \hline \hline
	\end{tabular}
	\caption{Convergence of the WGLasso and glasso under the Model 1 setting} \label{fig:converge-Identy}
\end{figure}

\begin{figure}
	\centering
	\begin{tabular}{p{1.8cm}|CC} \hline \hline
		outlier-to-signal ratio&$p=55$ & $p=100$  \\ \hline
		$0\%$&\subfloat{\includegraphics[width=6.7cm,height=6cm]{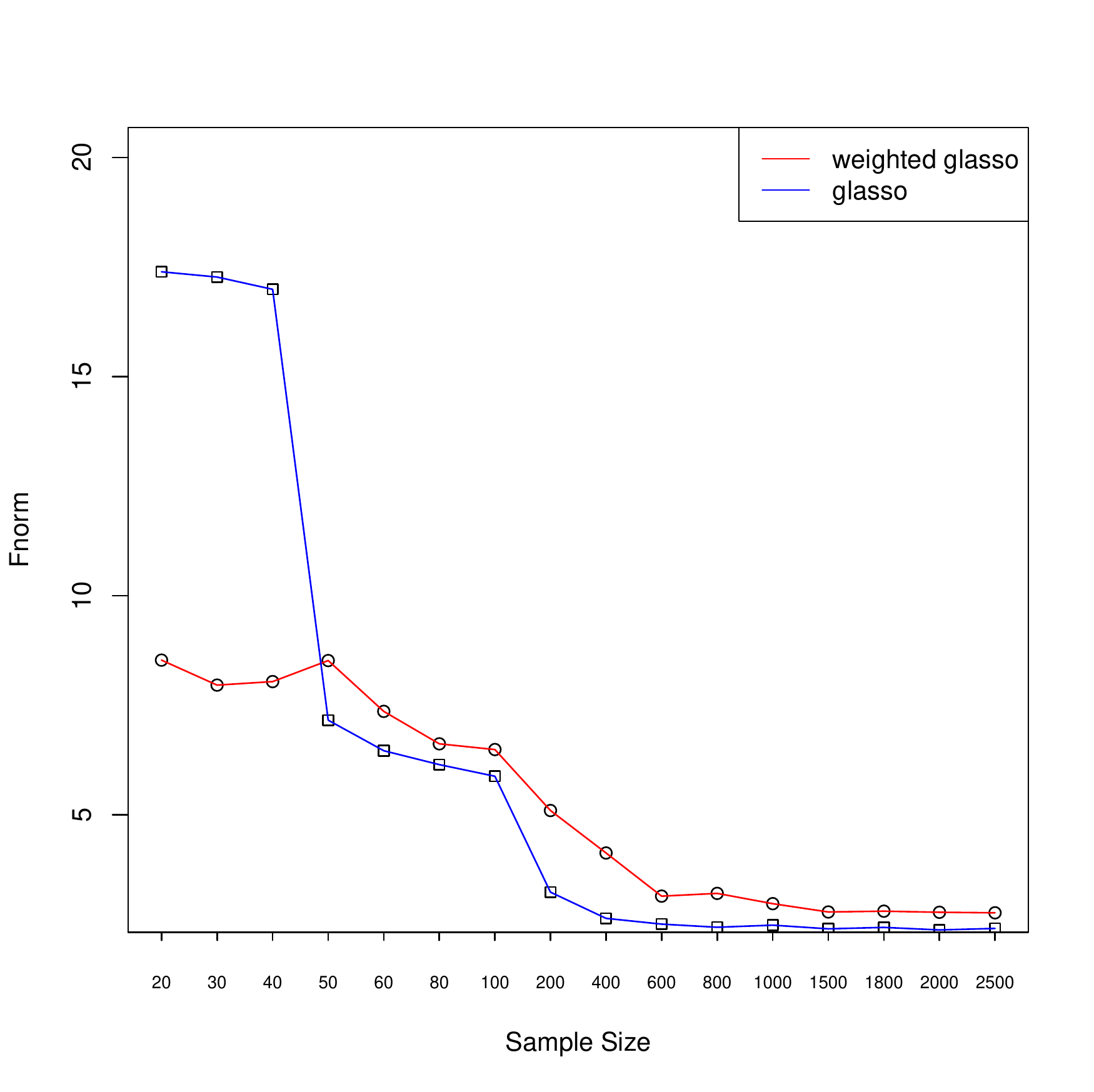}} &\subfloat{\includegraphics[width=6.7cm,height=6cm]{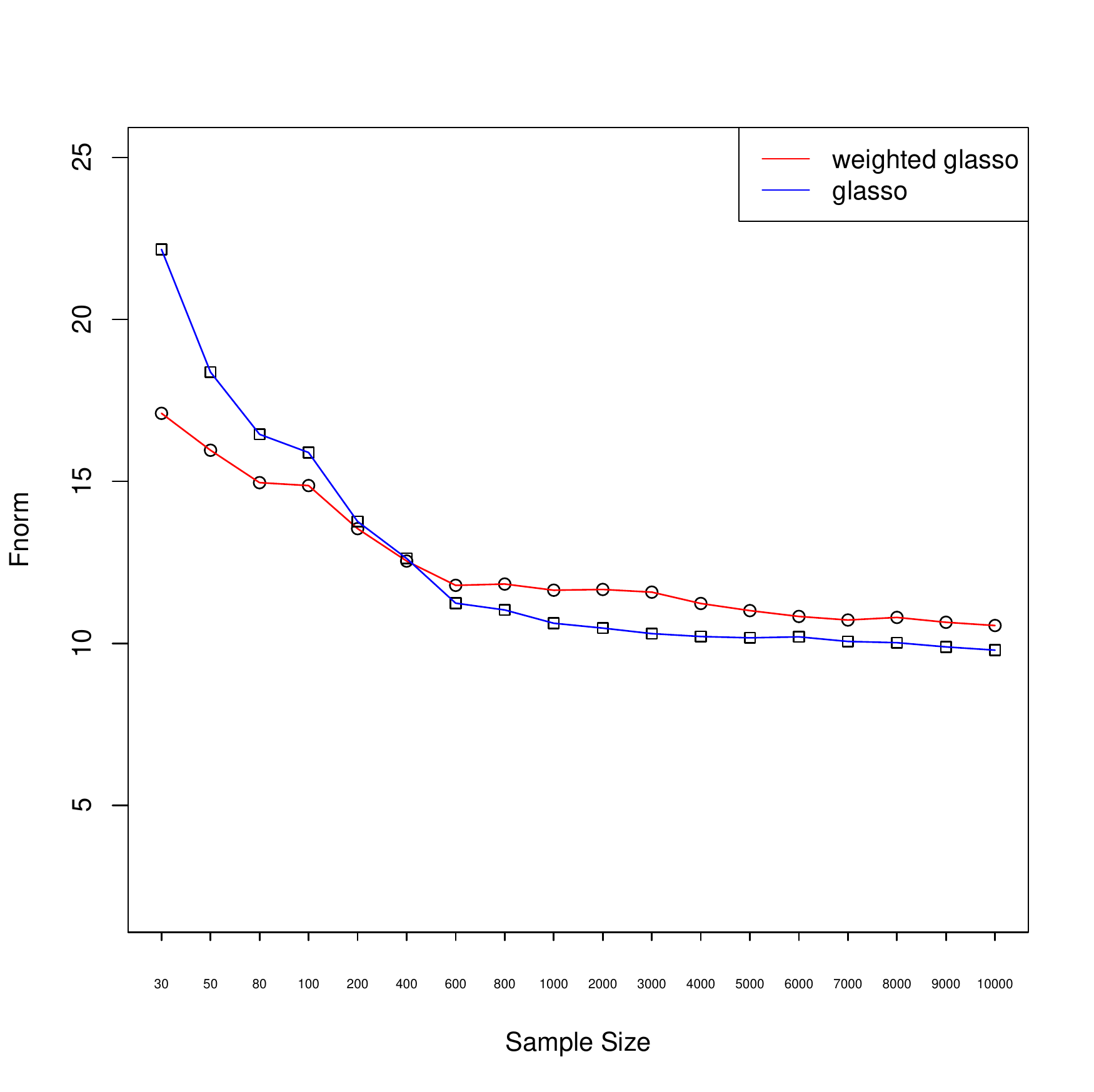}} \\ \hline
		$6\%$ & \subfloat{\includegraphics[width=6.7cm,height=6cm]{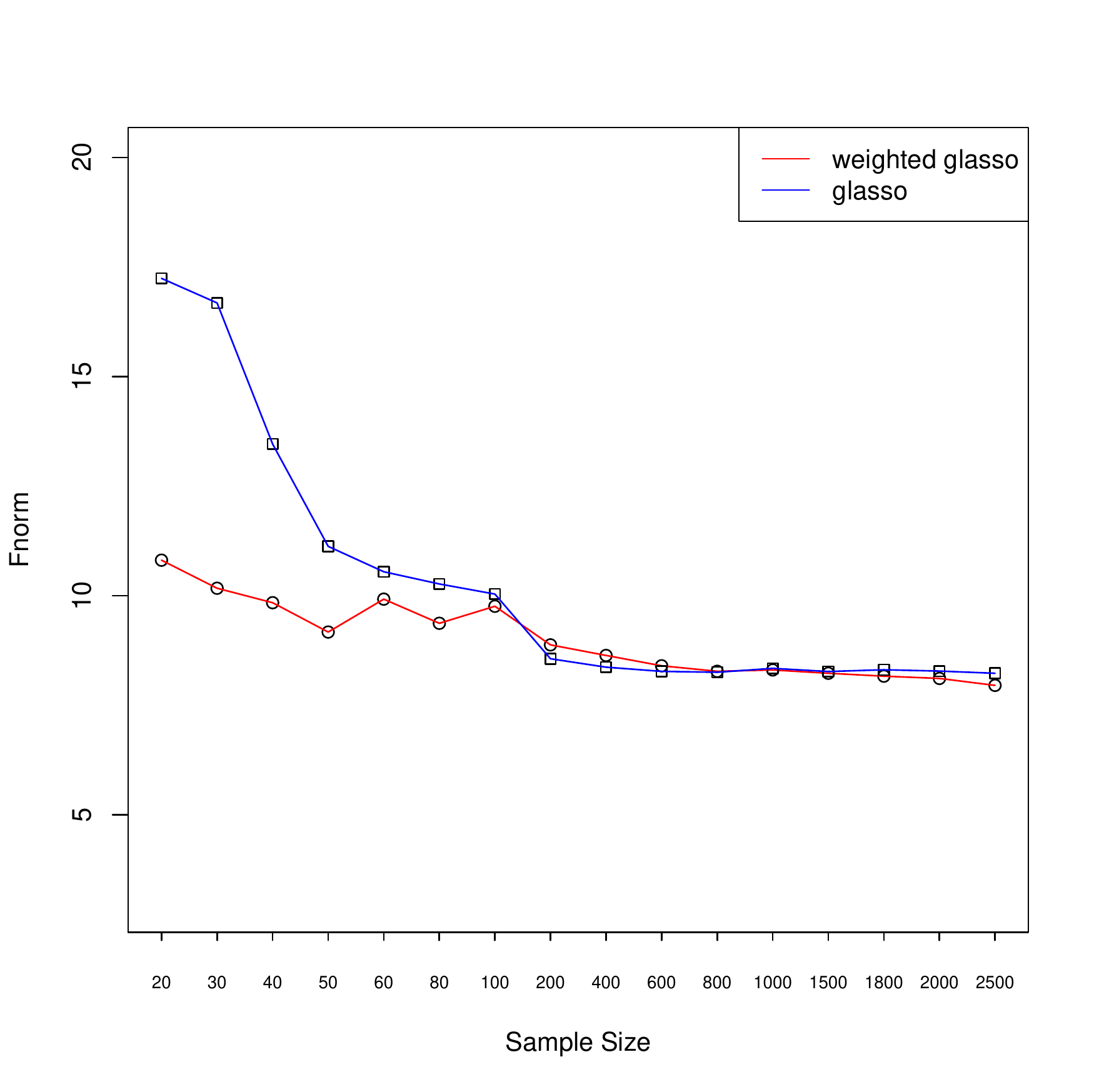}}&\subfloat{\includegraphics[width=6.7cm,height=6cm]{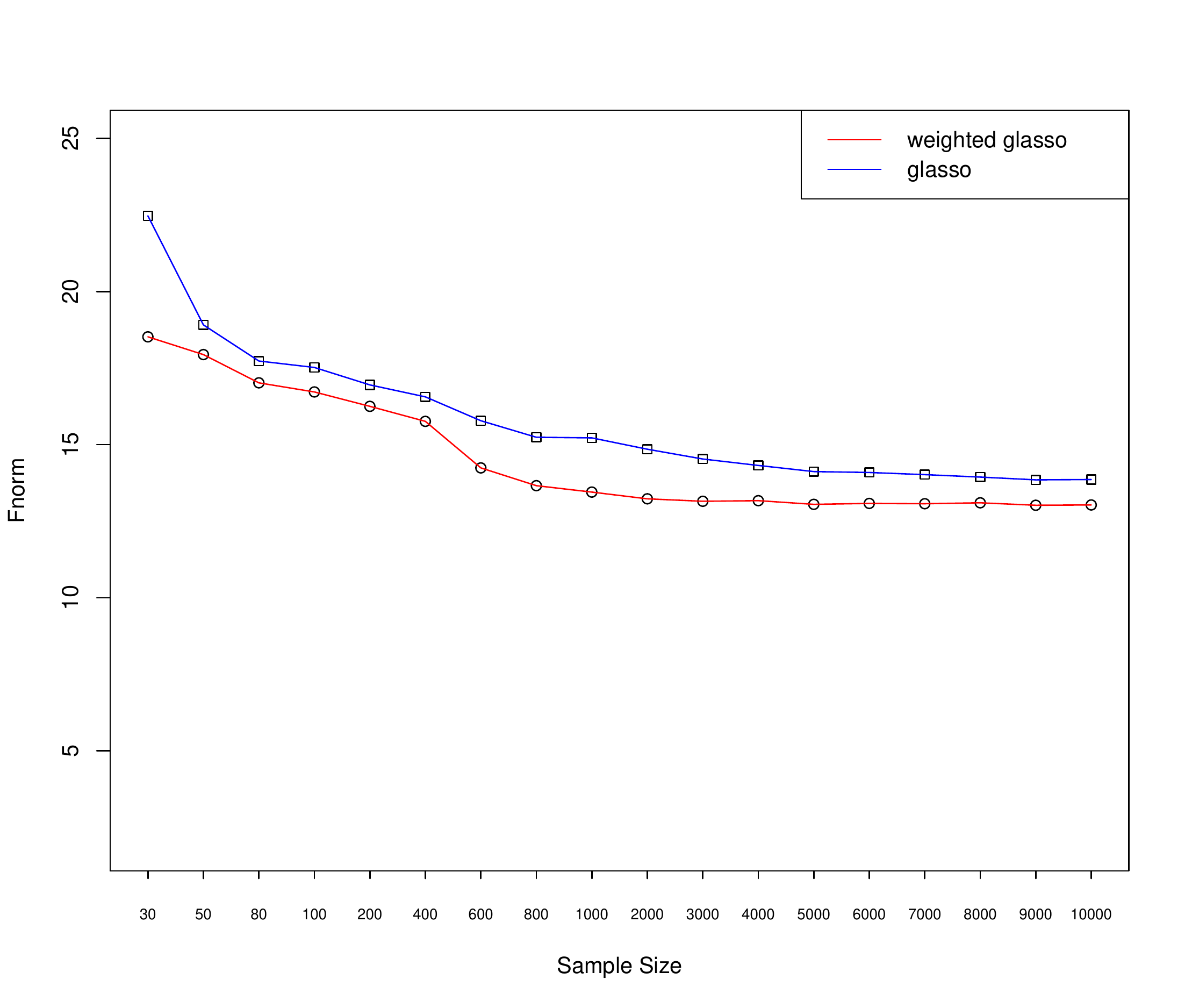}}\\ \hline
		$10\%$ & \includegraphics[width=6.7cm,height=6cm]{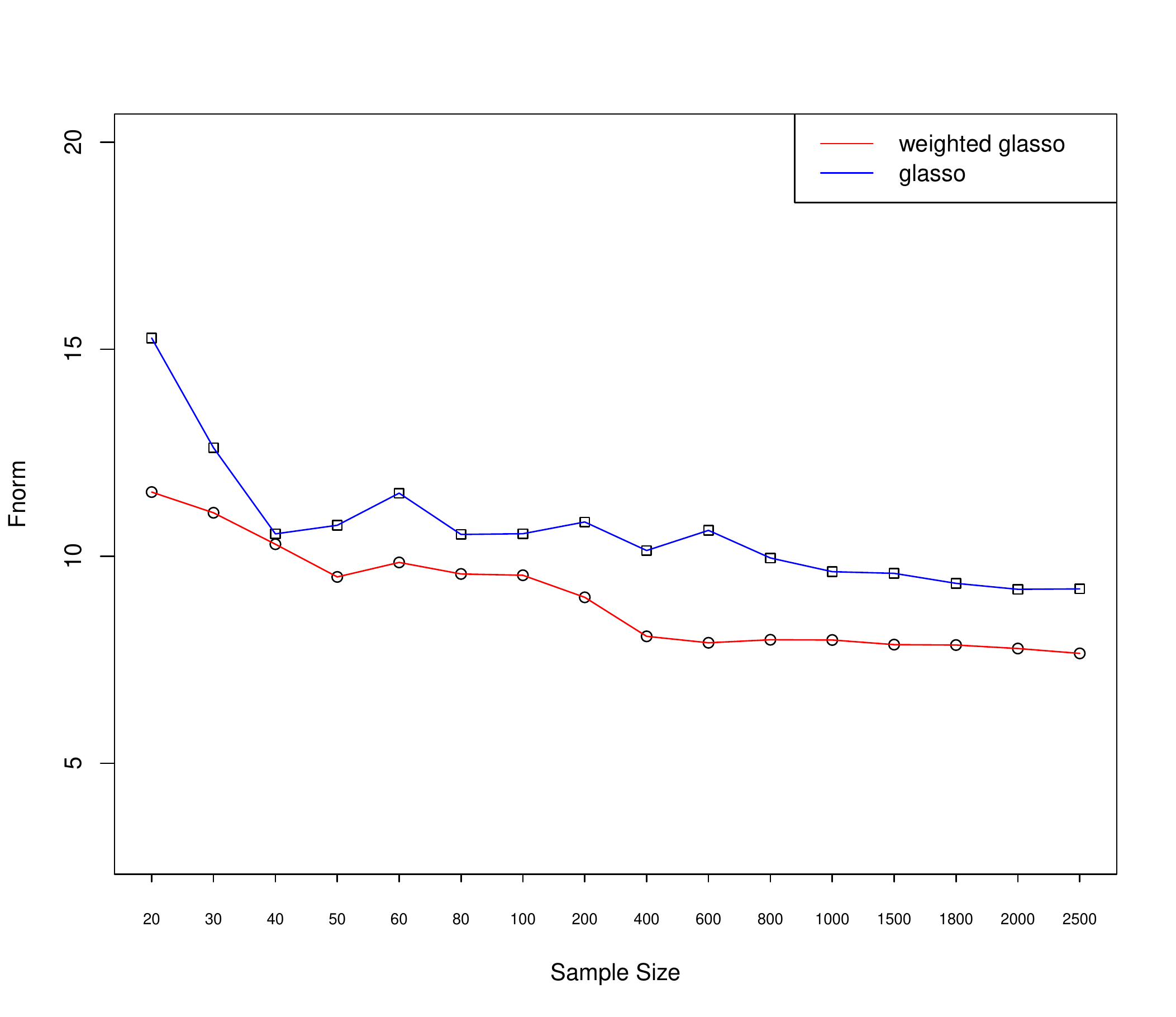}&\subfloat{\includegraphics[width=6.7cm,height=6cm]{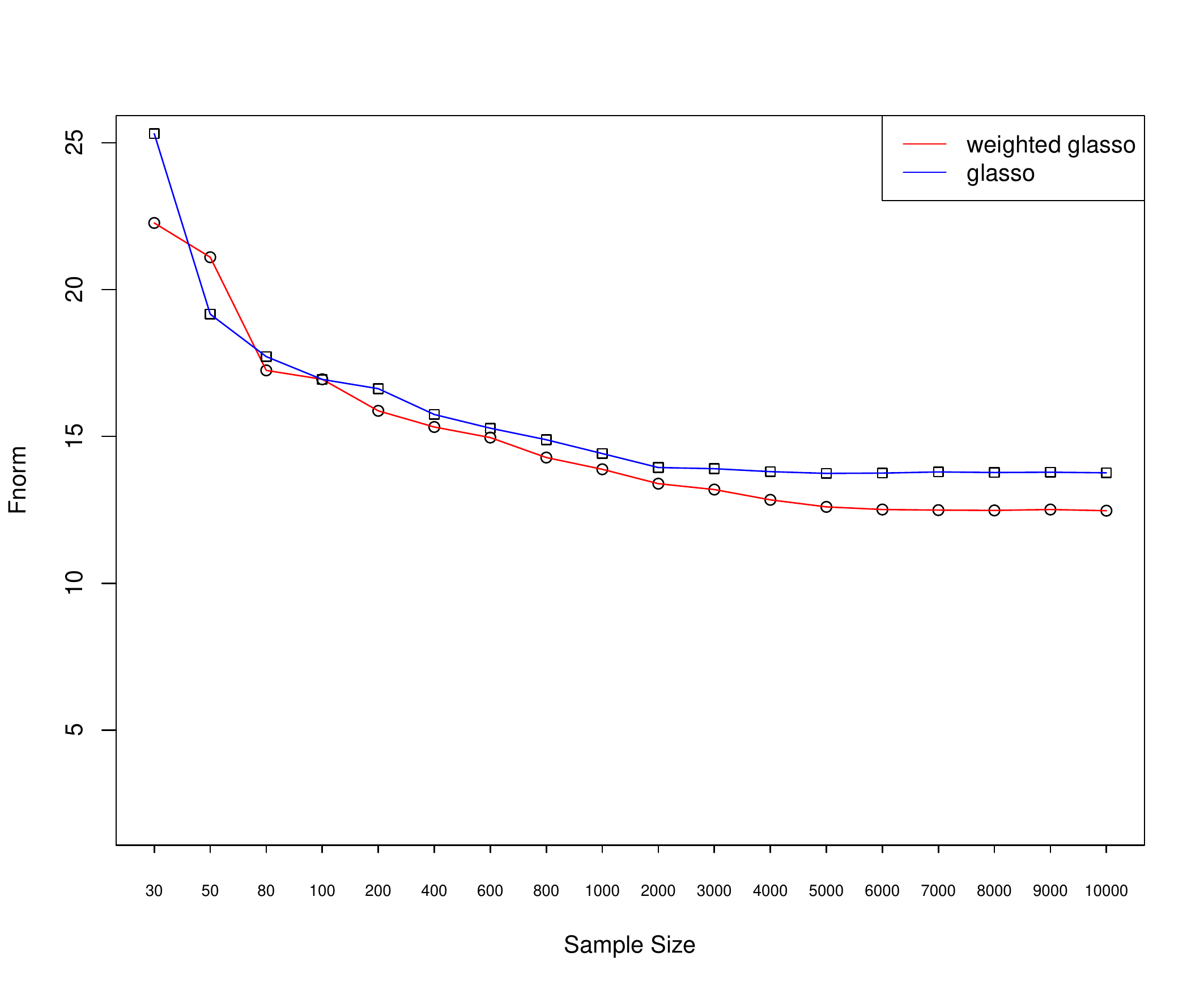}}\\ \hline \hline
	\end{tabular}
	\caption{Convergence of the WGLasso and glasso under the Model 2 setting} \label{fig:converge-AR1}
\end{figure}

\begin{figure}
	\centering
	\begin{tabular}{p{1.8cm}|CC} \hline \hline
		outlier-to-signal ratio&$p=55$ & $p=100$  \\ \hline
		$0\%$&\subfloat{\includegraphics[width=6.7cm,height=6cm]{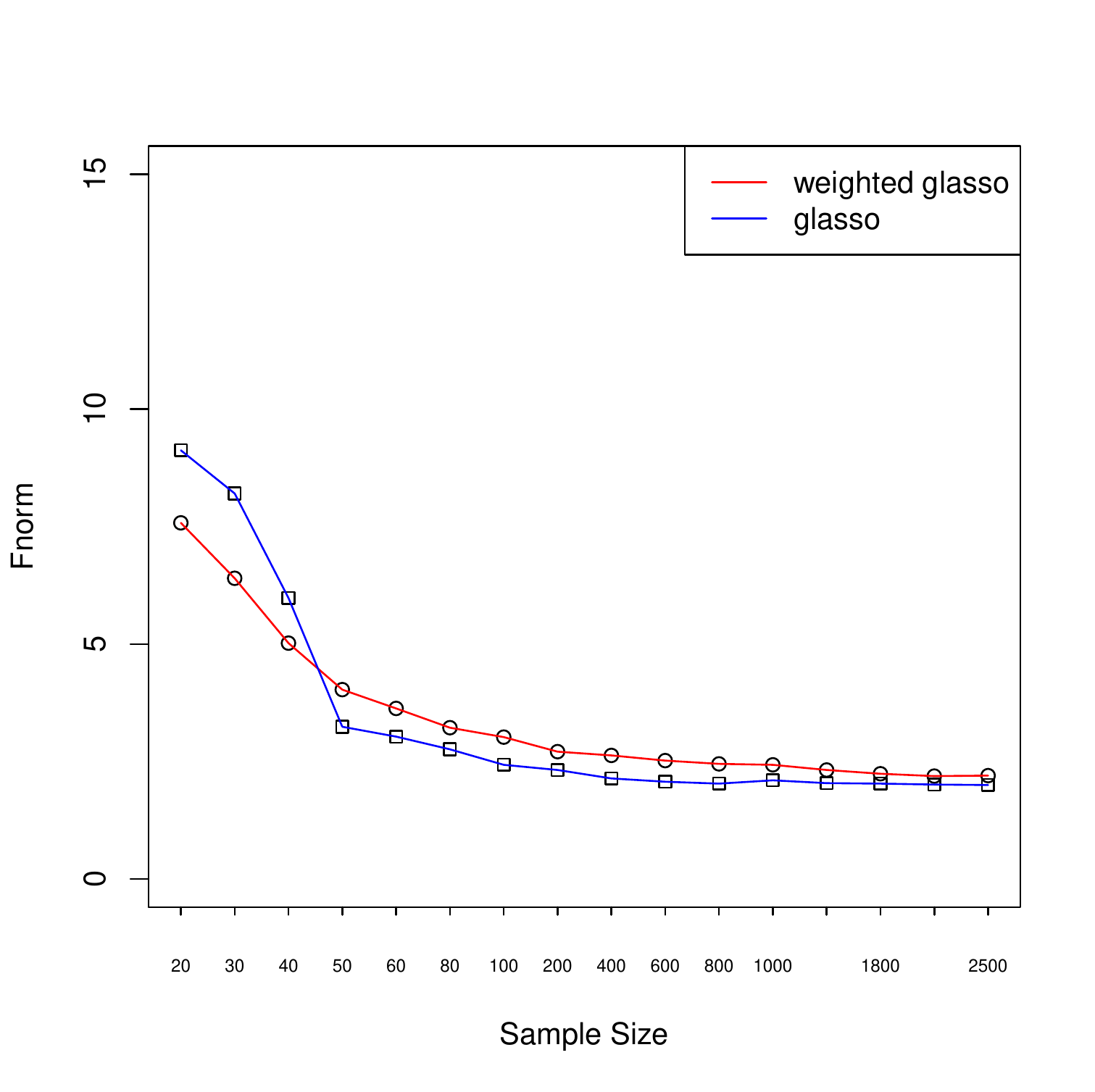}} &\subfloat{\includegraphics[width=6.7cm,height=6cm]{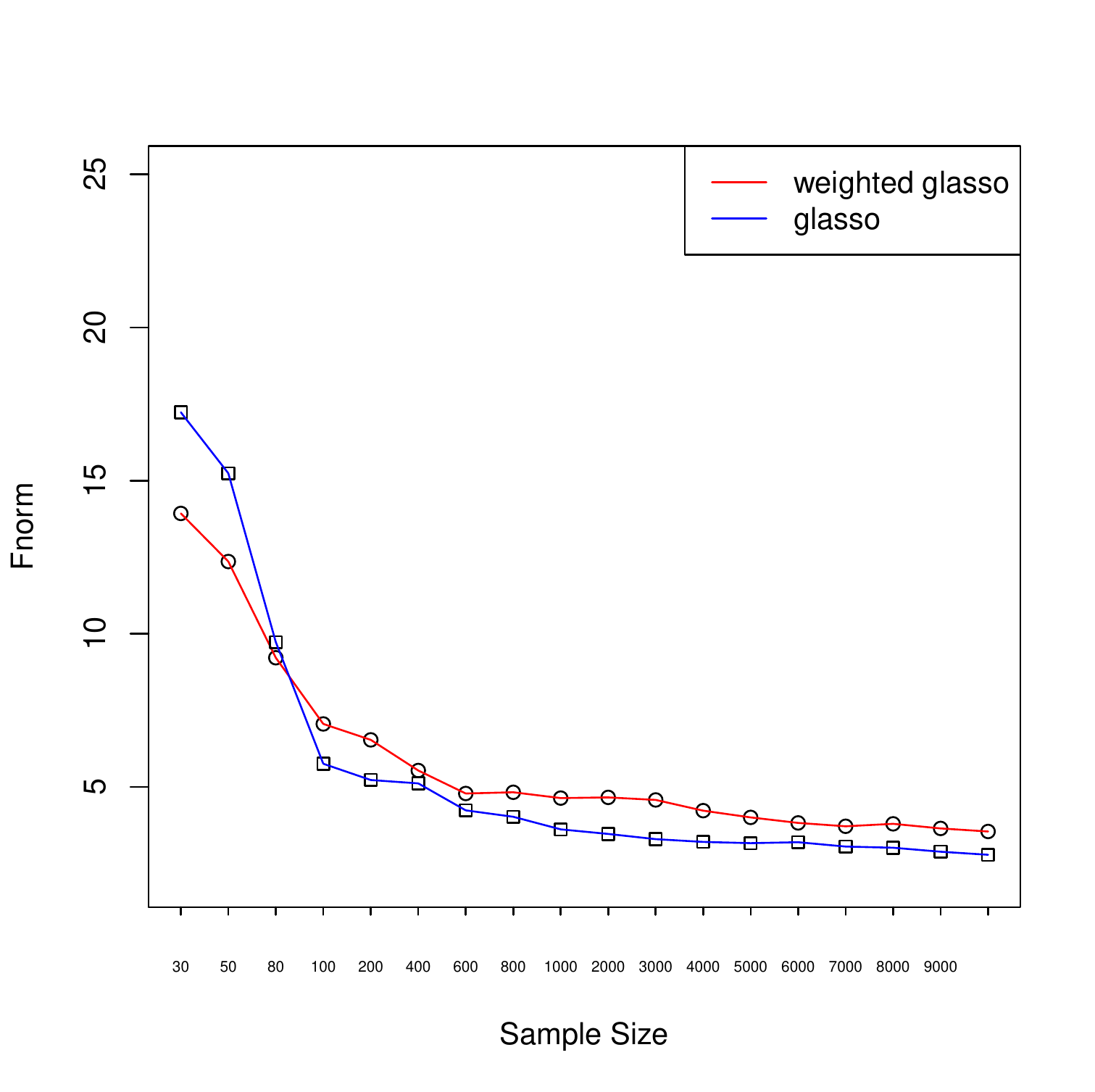}} \\ \hline
		$6\%$ & \subfloat{\includegraphics[width=6.7cm,height=6cm]{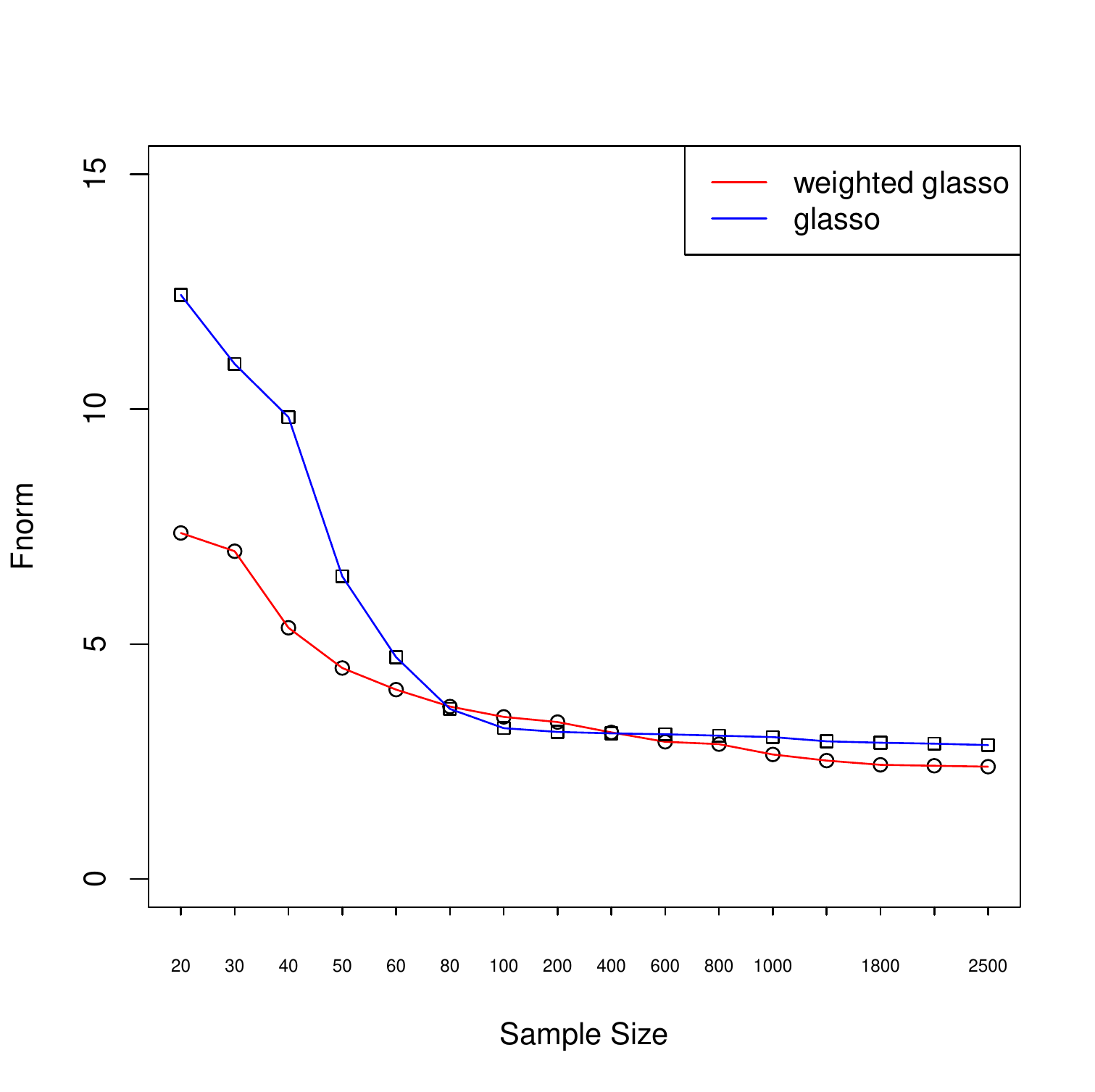}}&\subfloat{\includegraphics[width=6.7cm,height=6cm]{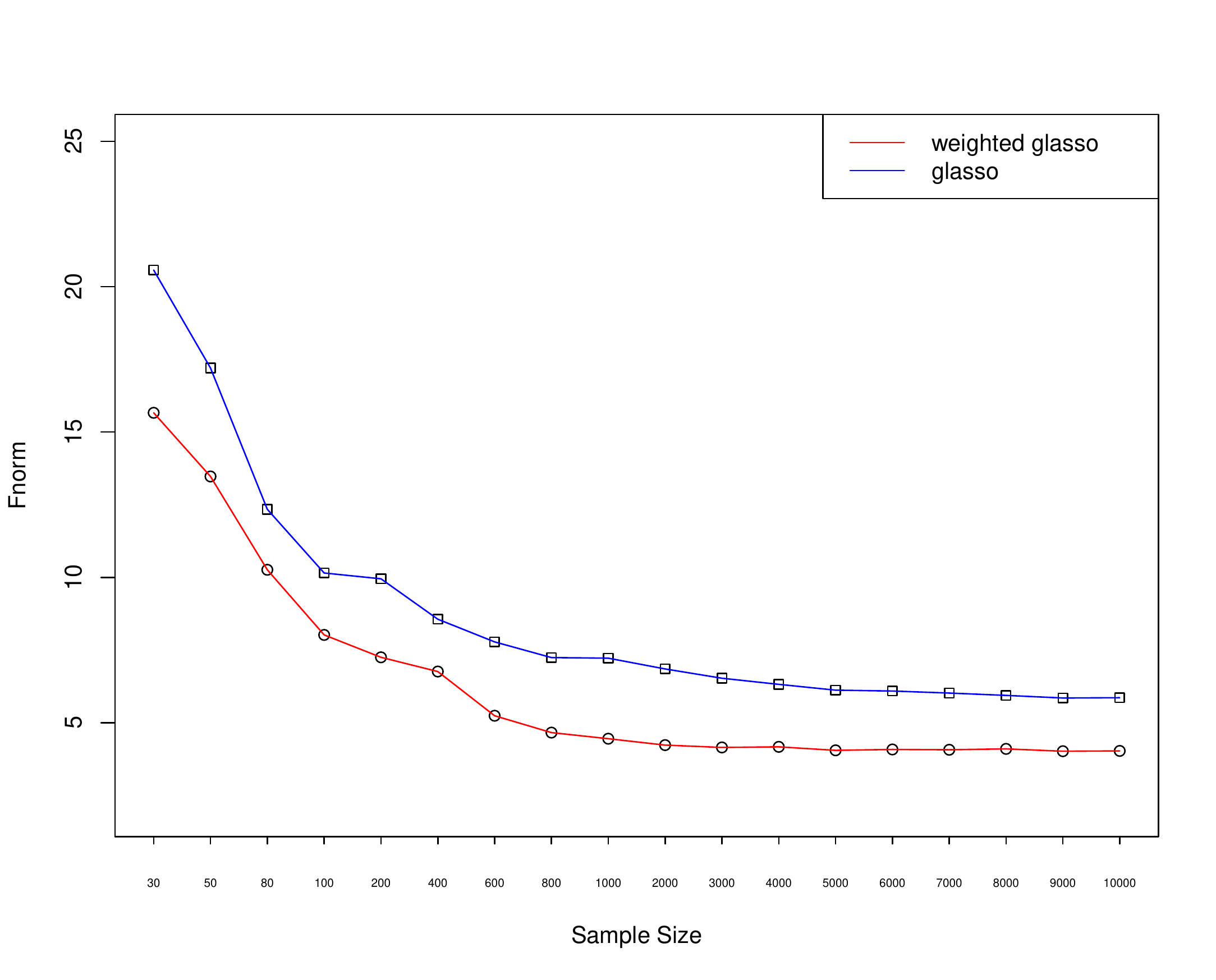}}\\ \hline
		$10\%$ & \includegraphics[width=6.7cm,height=6cm]{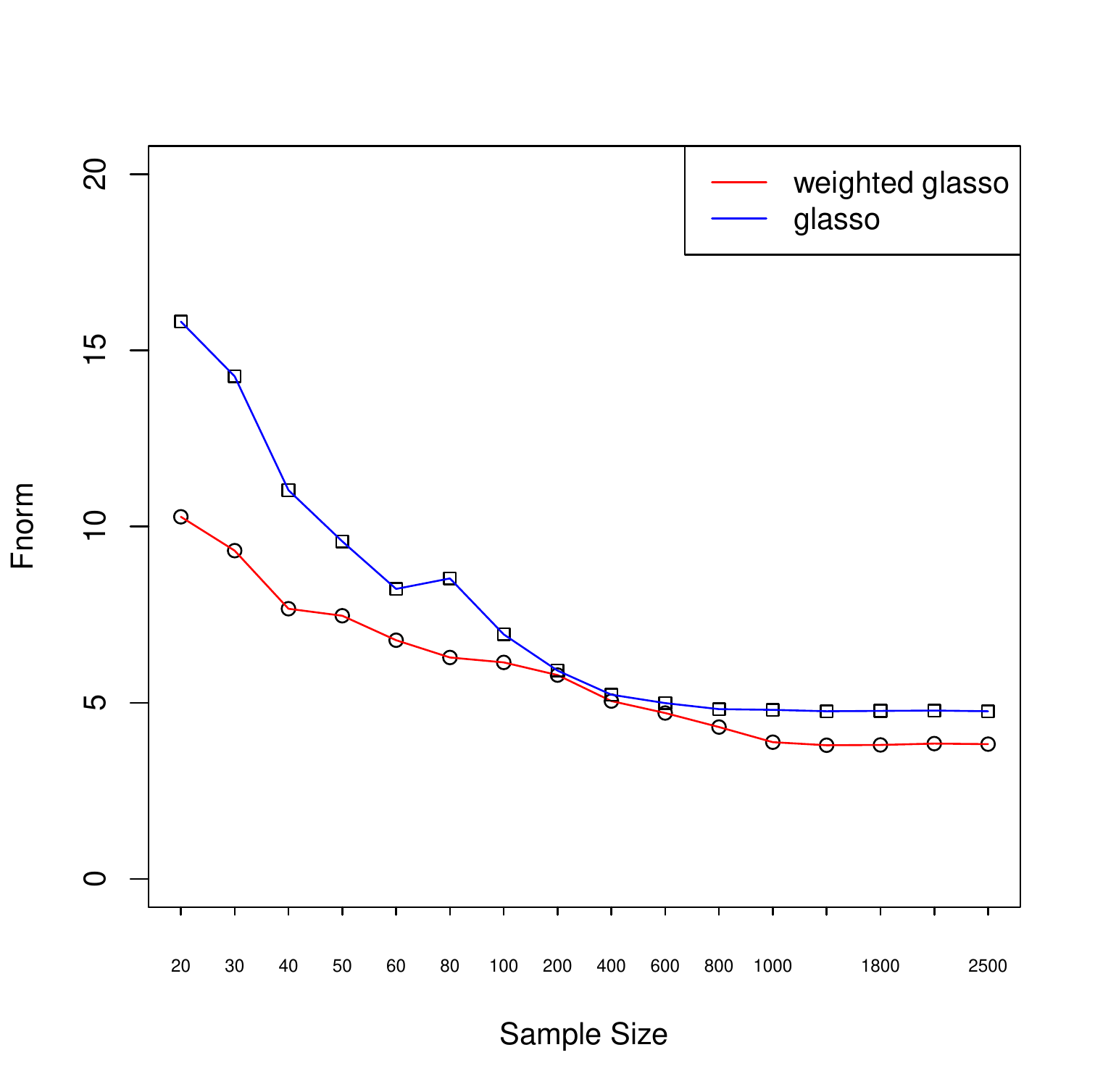}&\subfloat{\includegraphics[width=6.7cm,height=6cm]{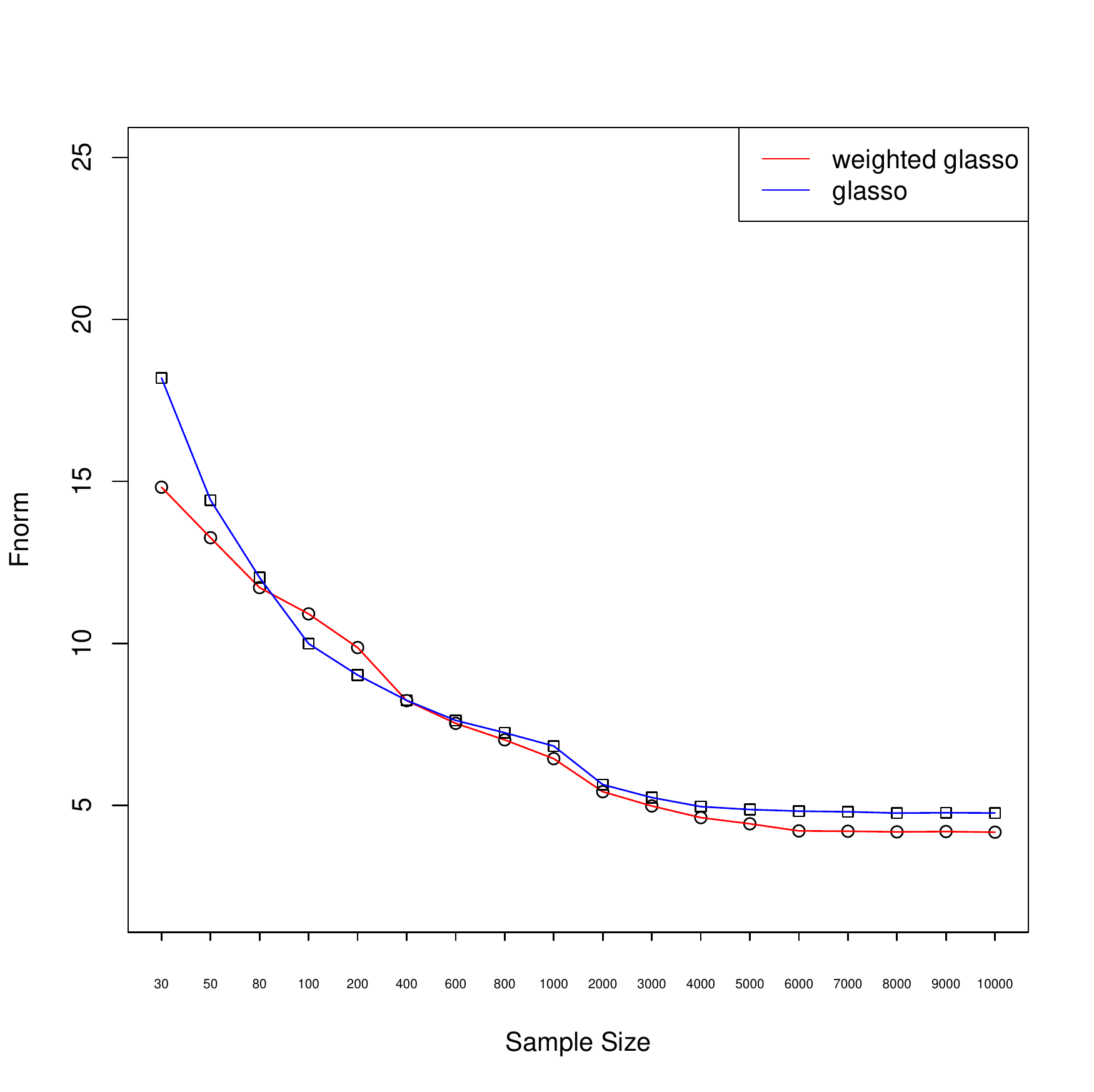}}\\ \hline \hline
	\end{tabular}
	\caption{Convergence of the WGLasso and glasso under the Model 3 setting} \label{fig:converge-RandAR1}
\end{figure}

From the above Monte Carlo simulations and  in comparison with the glasso estimate,
it appears that the estimate from the proposed method converges a bit more slowly to the true inverse covariance matrix, and has a smaller Fnorm distance to the true one when outliers are present.

\subsection{Numerical Study on Initial Estimators}
For the proposed method and other robust methods using adaptive weighting,
an initial estimator can play an important role in the final solution for both computation convergence and estimation and selection performances.
Note that for the proposed method, we take the inverse of the sample covariance matrix as an initial estimator.
When $p>n$, the sample covariance matrix is singular, and we add a small perturbation to its diagonal elements.

\begin{table}[hb]
\begin{center}
{\caption{Simulation results for the proposed method under different initial estimators under outlier-to-signal ratio = $8\%$}
\label{tab:robinit}
\footnotesize   \begin{tabular}{ cc | ccc| ccc |ccc  } % creating 27 columns
\hline \hline & & \multicolumn{3}{c|}{M1-I} & \multicolumn{3}{c|}{M2-I} & \multicolumn{3}{c}{M3-I} \\
[0.5ex] $p$ & Method & $F_{1}$ & $Fnorm$ & $KL$ & $F_{1}$ & $Fnorm$ & $KL$
& $F_{1}$ & $Fnorm$ & $KL$ \\
\hline
&Proposed  & 0.41 & 6.65 & 6.31 & 0.33 & 14.03 & 9.39 & 0.33 & 13.44 & 9.10  \\
&          &(0.05) &(0.79) & (0.74)  &(0.03) &(1.03) & (0.57) & (0.03)&(1.18)&(0.72)   \\
80        &robust initial  & 0.35  & 6.96 & 6.76 & 0.29 & 15.31 & 10.38 & 0.30 & 14.81 & 9.96\\
&          & (0.05) & (1.29) & (1.14) &  (0.03) & (2.22) & (1.62) & (0.03)& (1.83) & (1.17) \\
               &one step further  & 0.45  & 6.69 & 6.41 & 0.31 & 14.72 & 9.85 & 0.32 & 14.43 & 9.67\\
&          & (0.07) & (0.83) & (0.62) &  (0.03) & (1.51) & (0.91) & (0.03)& (1.10) & (0.63) \\
\hline
& & \multicolumn{3}{c|}{M1-II} & \multicolumn{3}{c|}{M2-II} & \multicolumn{3}{c}{M3-II} \\
[0.5ex]   &          & $F_{1}$ & $Fnorm$ & $KL$ & $F_{1}$ & $Fnorm$ & $KL$
& $F_{1}$ & $Fnorm$ &$KL$  \\
\hline
&Proposed  & 0.41 & 7.87 & 7.52 & 0.34 & 16.03 & 10.71 & 0.33 & 15.94 & 10.68\\
&          &(0.05) &(1.44) & (1.16)  &(0.03) &(1.77) &(1.11)   &(0.03) &(1.13) &(0.87)   \\
80        &robust initial  & 0.35  & 8.44 & 8.16 & 0.30 & 18.19 & 12.31 & 0.30 & 18.23 & 12.42\\
&          & (0.07) & (1.48) & (1.26) &  (0.03) & (2.07) & (1.53) & (0.03)& (2.57) & (1.82) \\
             &one step further & 0.41 & 8.65 & 8.07 & 0.33 & 16.90 & 11.35 & 0.32 & 17.55 & 11.80\\
&          & (0.03) & (2.05) & (1.65) &  (0.03) & (1.61) & (0.99) & (0.02)& (2.10) & (1.42) \\
\hline
\end{tabular}
}
\end{center}
\end{table}

To evaluate the role of the initial estimator for the proposed method, we consider two alternatives of initial estimators.
The first one, denoted as ``robust initial," uses the robust covariance matrix in \"{O}llerer  and Croux (2015) as an initial estimator for our Algorithm 1.
The second one, denoted as ``one-step further," considers using the estimator from currently proposed method as an initial estimator for Algorithm 1 (i.e., re-conducting the iteratively weighted graphical lasso method).
Table \ref{tab:robinit} reports the performance results under the setting of $8\%$ outliers and $p=80$ based on 100 simulations. The standard errors are in parentheses.

From the results (i.e., the mean value and standard error of performance measures) in the table,
it is seen that the proposed method under robust initial estimators generally have a comparable performance with the original proposed method using the sample covariance matrix as an initial estimator.

\section{Case Study} \label{sec:case}
In this section, we illustrate the merits of the proposed method through a real-world application on gene network inference for breast cancer data.

The estimation of inverse covariance matrix is widely used in gene network applications to identify important relationships among genes.
Here we adopt the gene expression data provided by Hess et al.\,(2006) to perform inference on the genetic networks.
This data set contains $133$ patients with stage I--III breast cancer, who were treated with chemotherapy. Two clusters are defined based on the patient responses
to the treatment: pathologic complete response (pCR) and residual
disease (not-pCR). Hess et al.\,(2006) and Natowicz et al.\,(2008)
identified $26$ key genes important for the treatment.
A detailed description of these genes is listed in Appendix \ref{sec:app2}.
As suggested by Ambroise et al.\,(2009), cases from the two classes of pCR (34 patients) and not-pCR (99 patients) do not have the same distribution since they were treated under different experimental conditions. Thus we apply the proposed inverse covariance matrix estimation procedure on each cluster separately.
The proposed method of estimating ${\bf \Omega}$ is compared with the glasso and the tlasso methods.
Here the LW method is not included for comparison since it would not produce a sparse graph for the gene network. The hub genes, those with most connected edeges are potentially most significant for breast cancer and would be emphasized in a further genetic study.

%\begin{sidewaysfigure}
\begin{figure}
\caption{Results of three methods in breast cancer gene
expressions data} \label{res-app}
 \centering
\begin{tabular}{V|c|c|c} \hline \hline
   &WGLasso & glasso & tlasso \\ \hline
 pCR class & \subfloat{\includegraphics[width=0.3\textwidth]{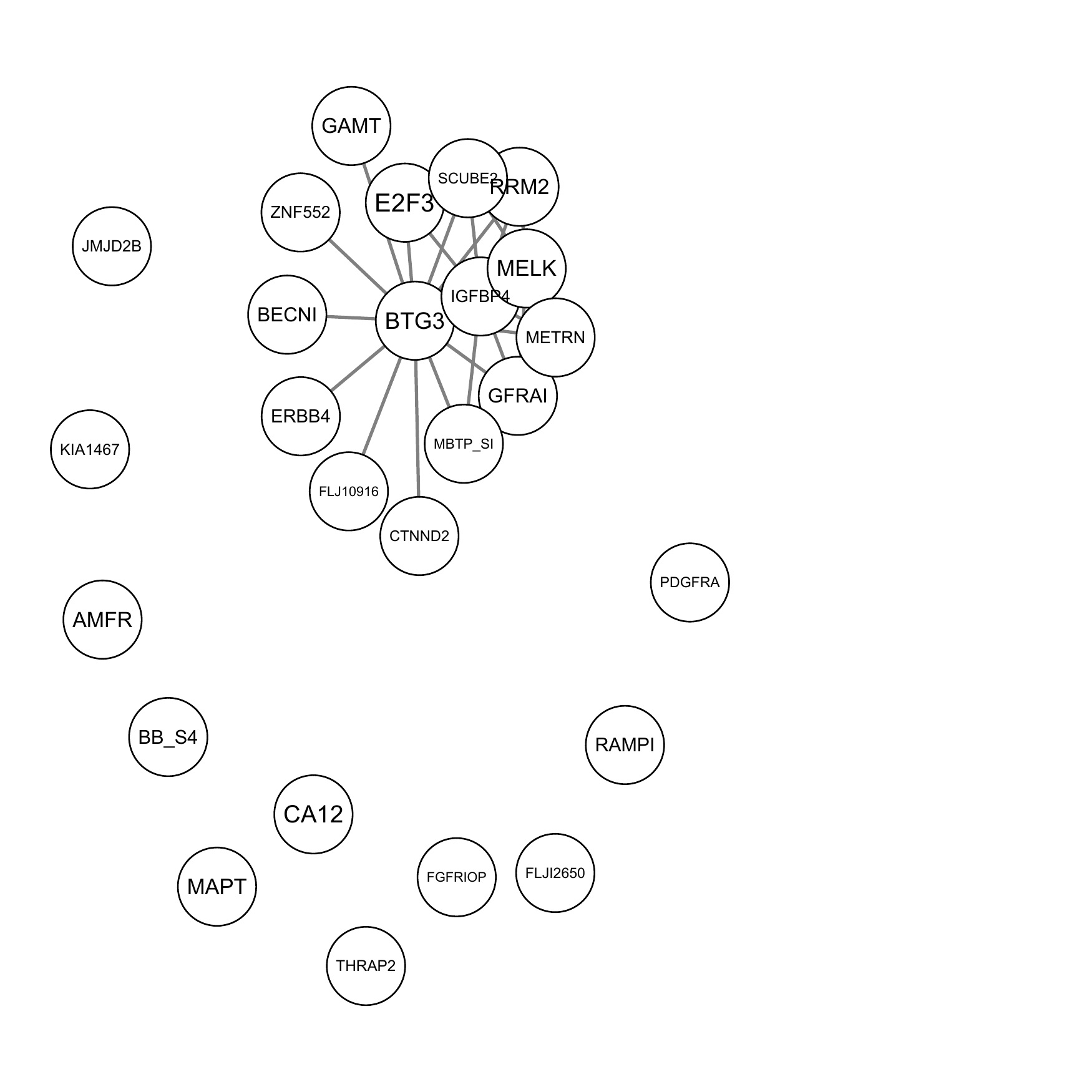}}
          &\subfloat{\includegraphics[width=0.3\textwidth]{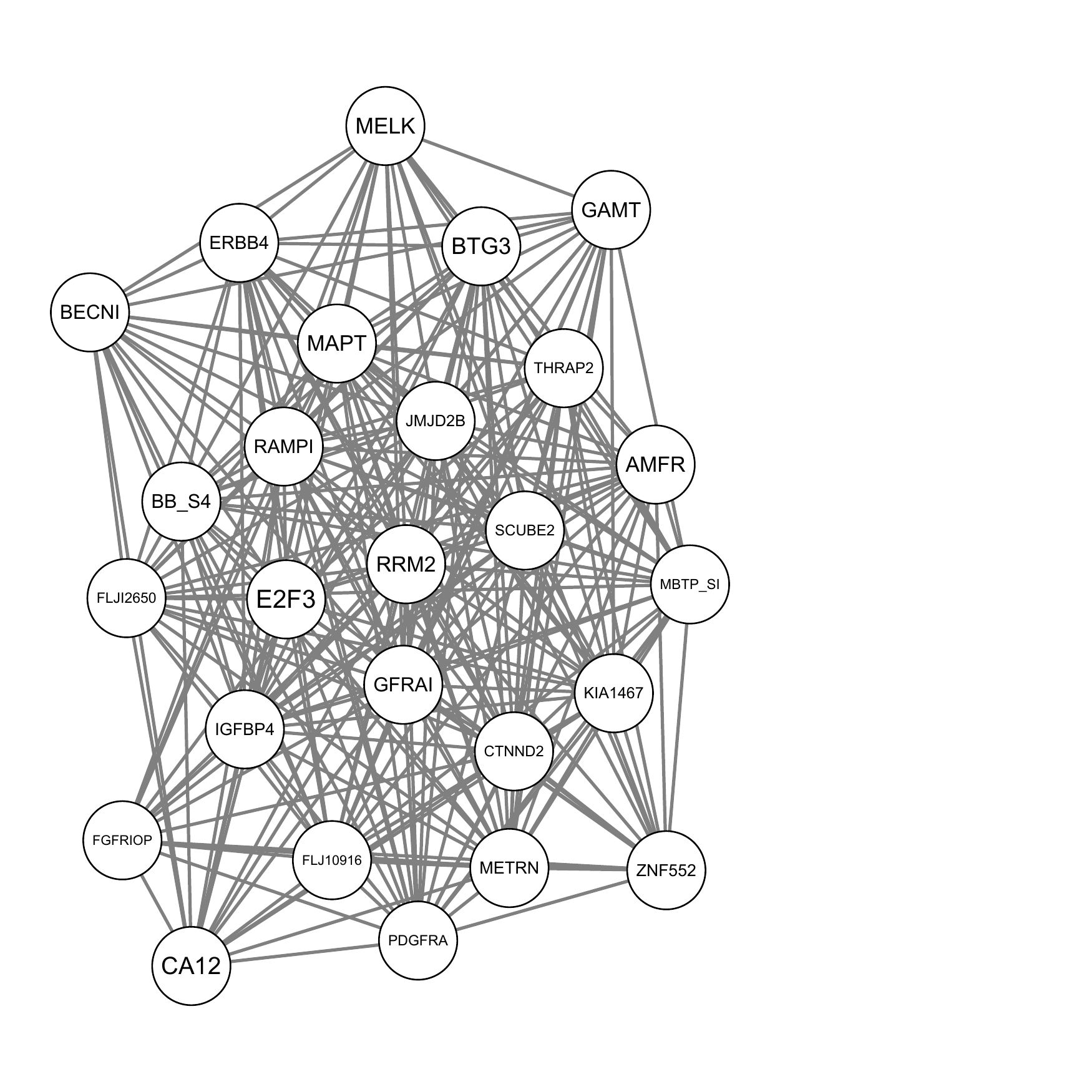}}
          & \subfloat{\includegraphics[width=0.3\textwidth]{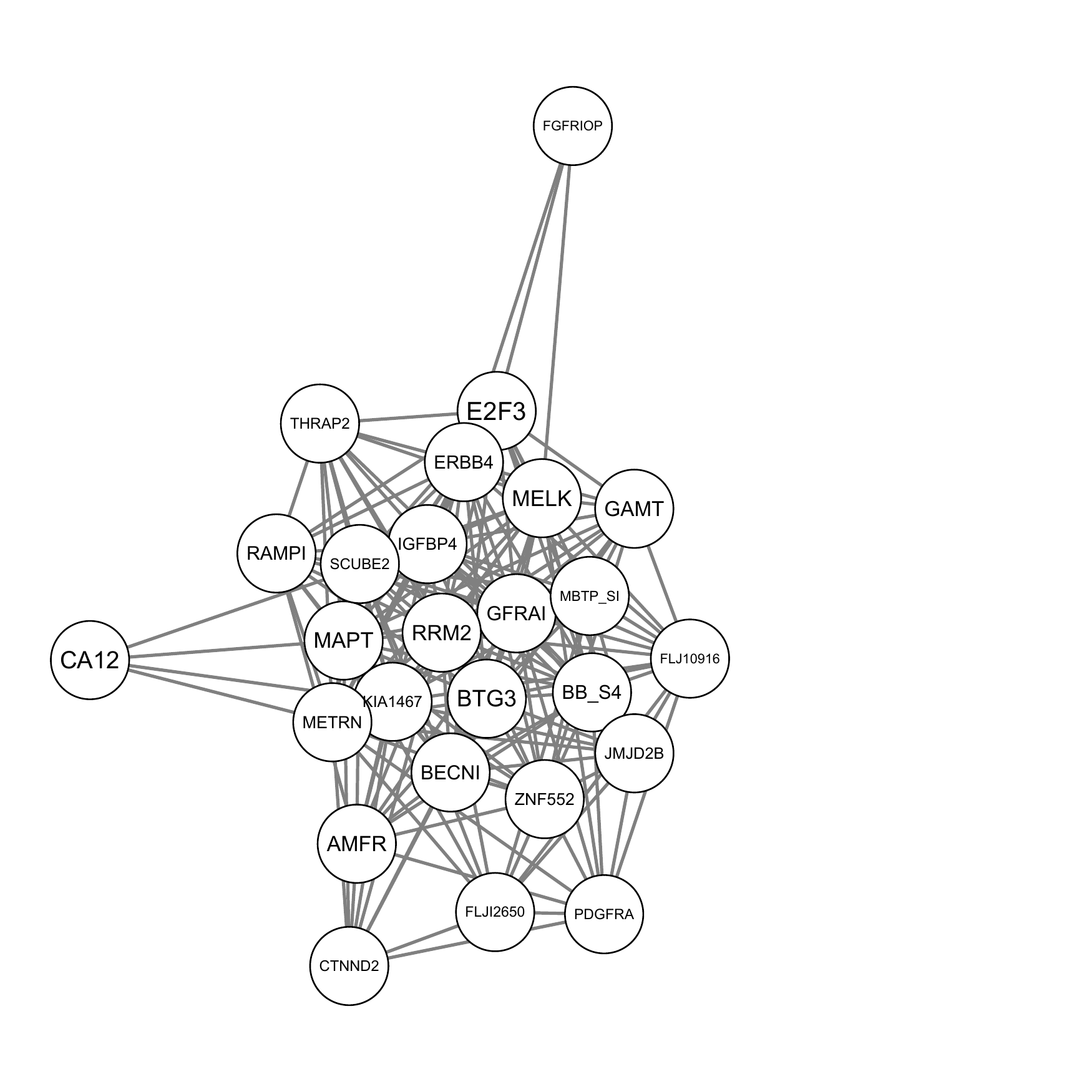}}
  \\ \hline
  not-pCR class  &
  \subfloat{\includegraphics[width=0.3\textwidth]{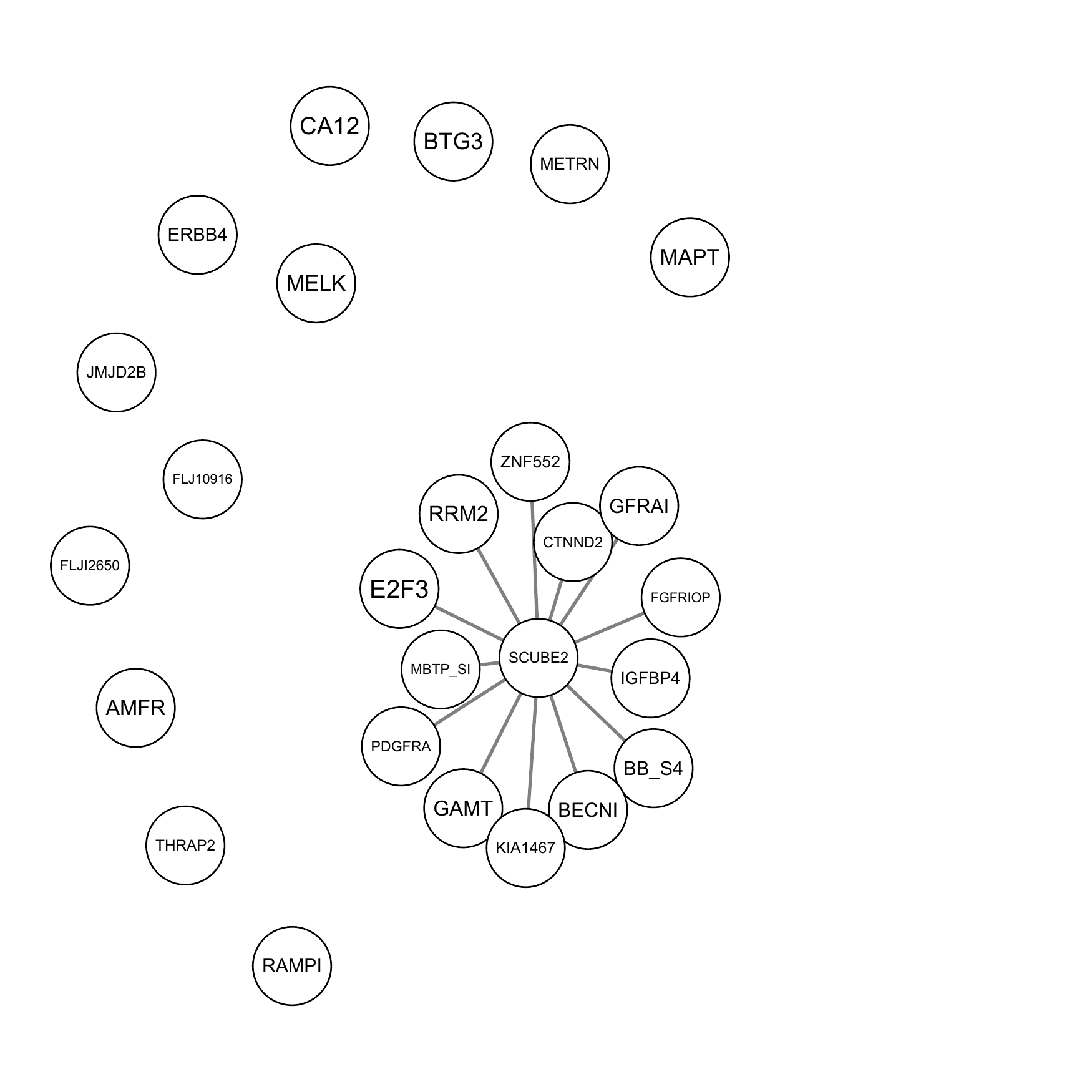}}
    &
     \subfloat{\includegraphics[width=0.3\textwidth]{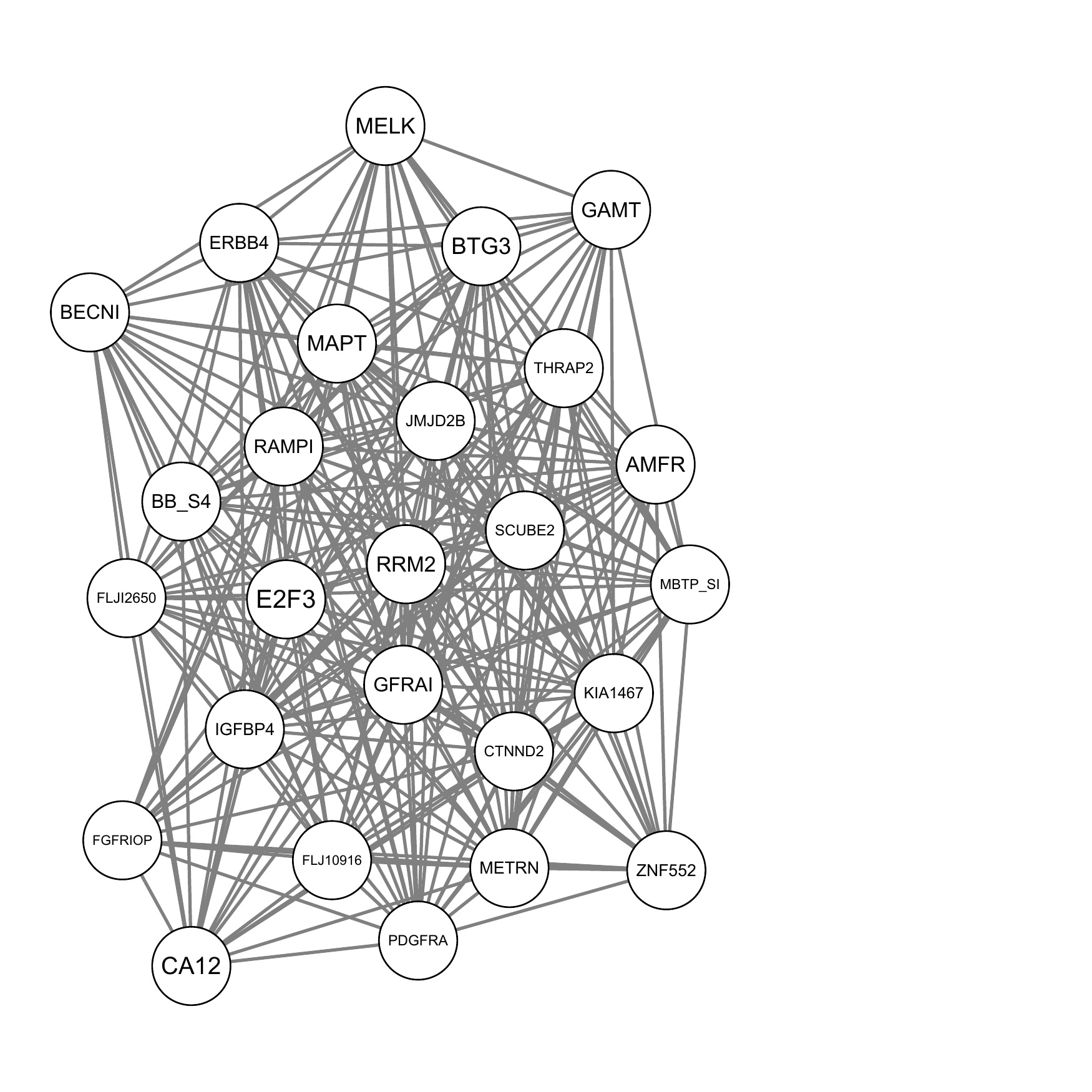}}
    &
      \subfloat{\includegraphics[width=0.3\textwidth]{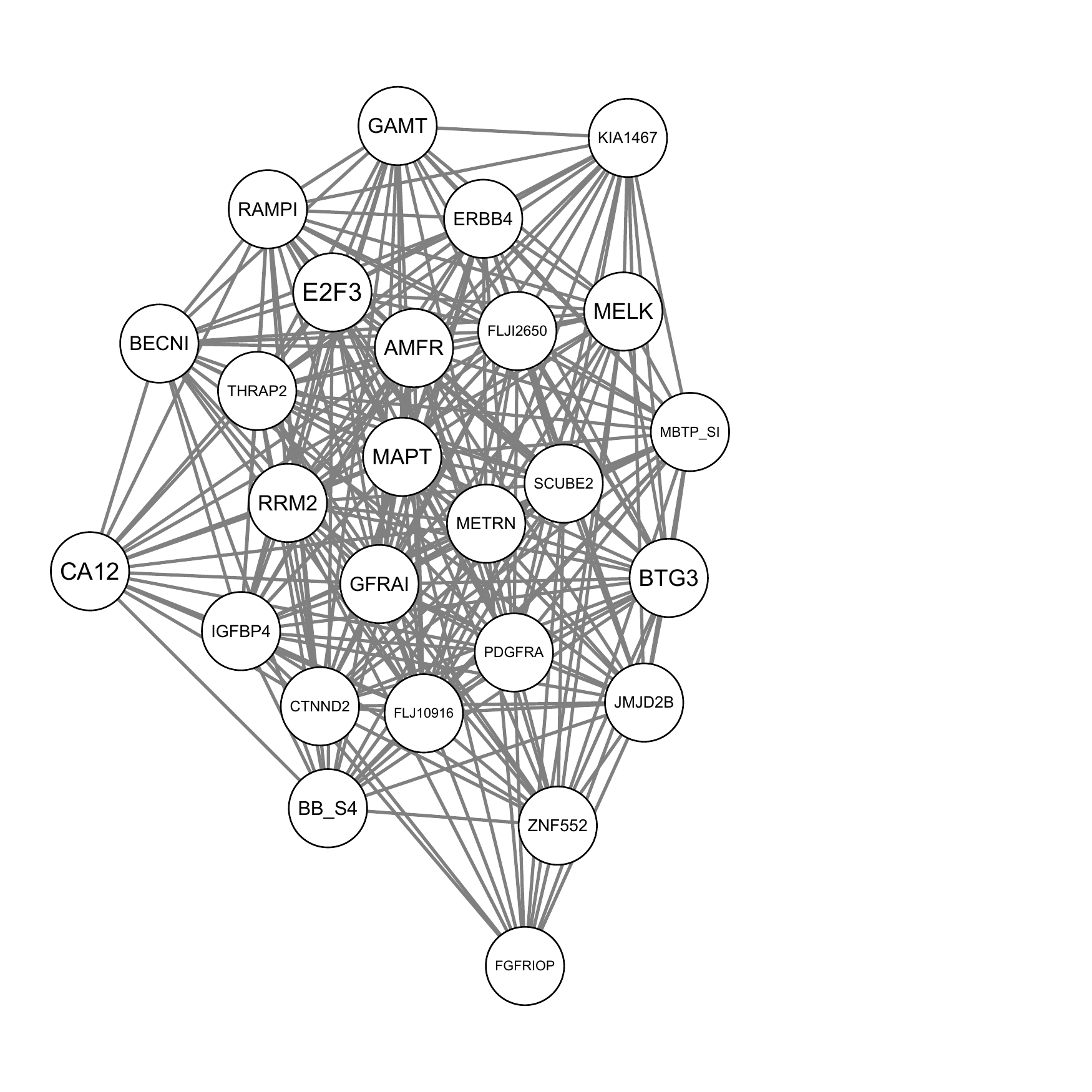}}
  \\  \hline
   %\\ \hline
   both classes  &
   \subfloat{\includegraphics[width=0.3\textwidth]{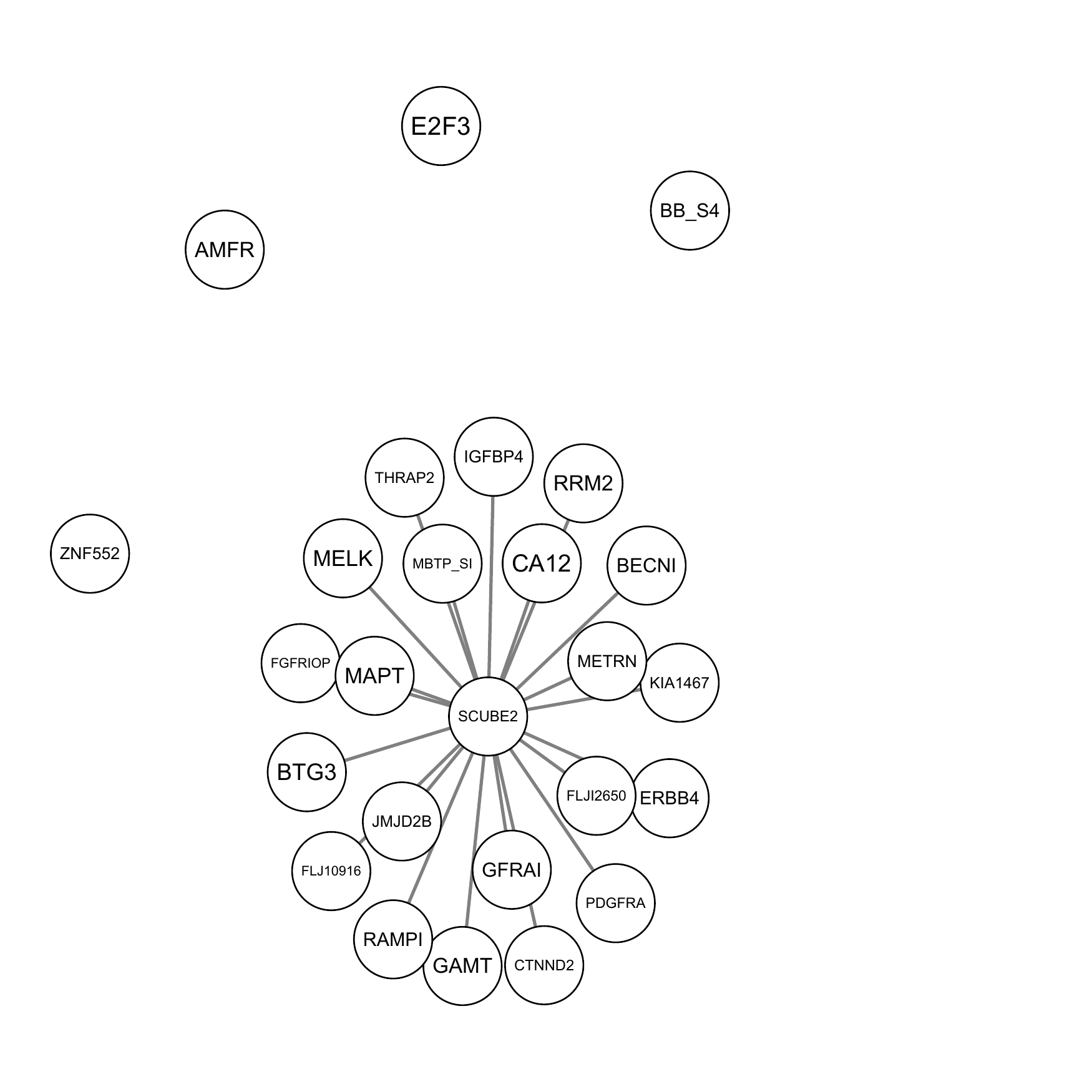}}
   &
   \subfloat{\includegraphics[width=0.3\textwidth]{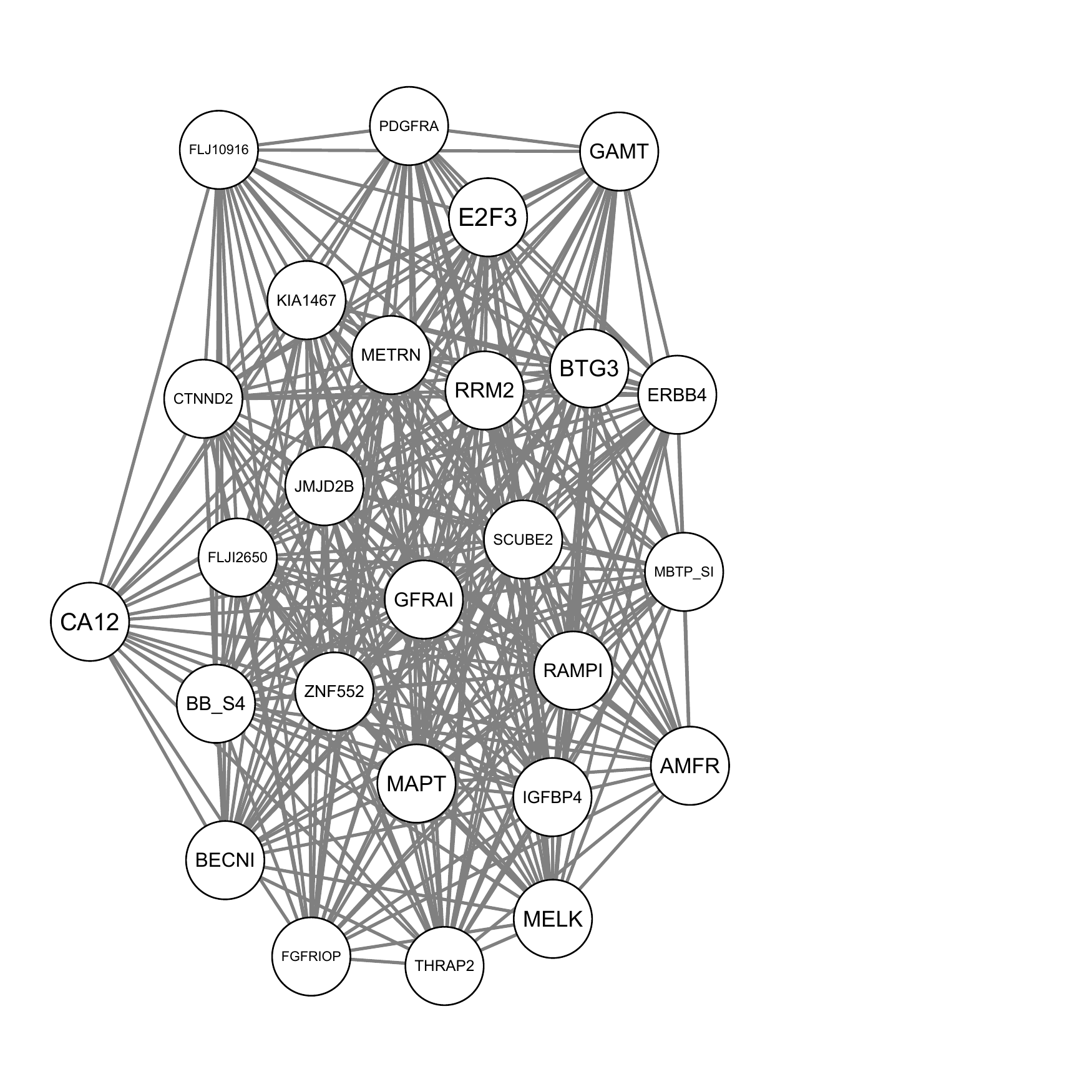}}
   &
   \subfloat{\includegraphics[width=0.3\textwidth]{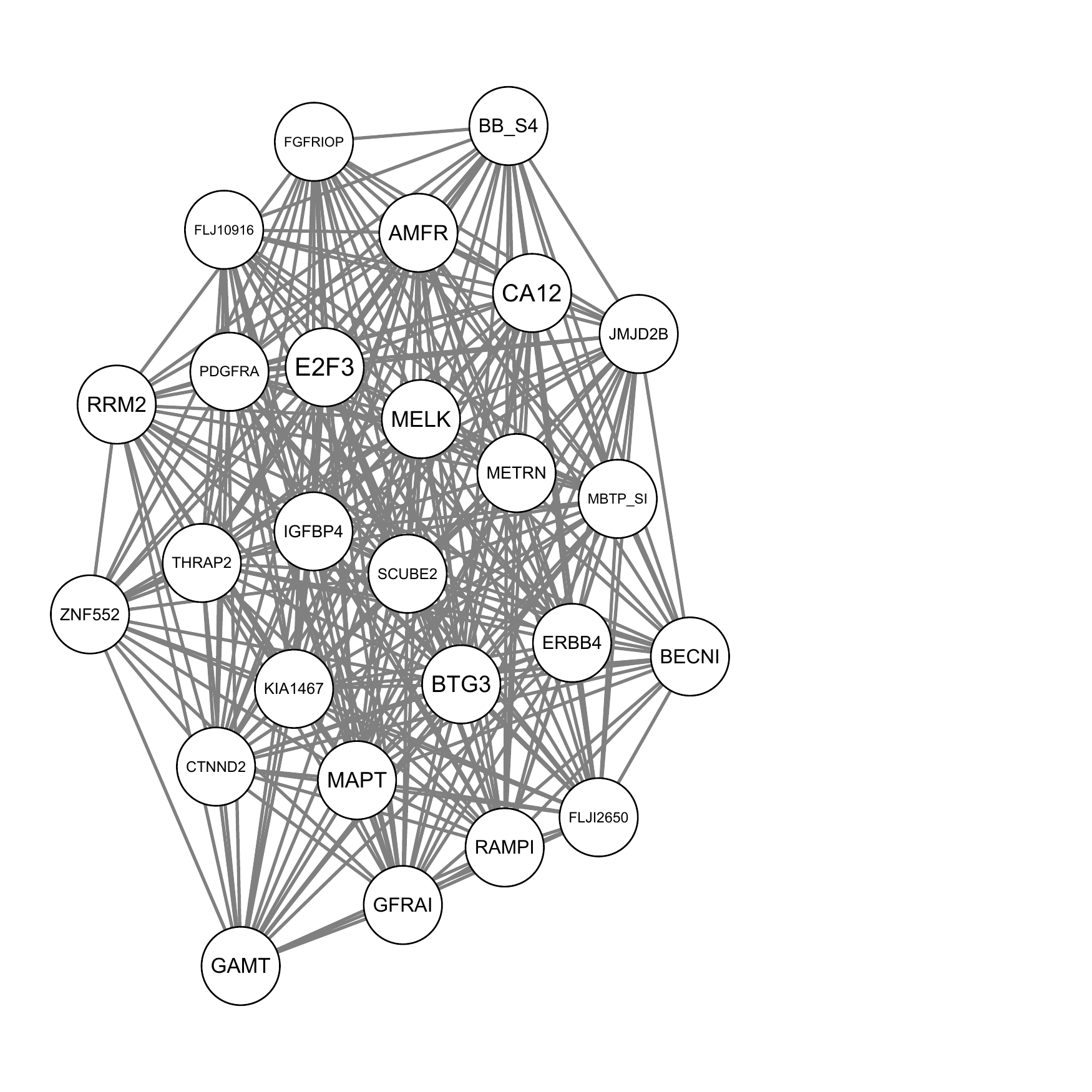}}
   \\  \hline \hline

 \end{tabular}
%\end{sidewaysfigure}
\end{figure}

\begin{sidewaysfigure}
\caption{Results of proposed WGLasso in gene
expressions data under three levels of penalty values} \label{res-app-wgl}
 \centering
\begin{tabular}{V| C |C |C}\hline \hline
  &low level ($\rho = 0.1$) & middle level ($\rho = 0.6$) & high level ($\rho = 0.9$) \\ \hline
  pCR class &\subfloat{\includegraphics[width=0.28\textwidth]{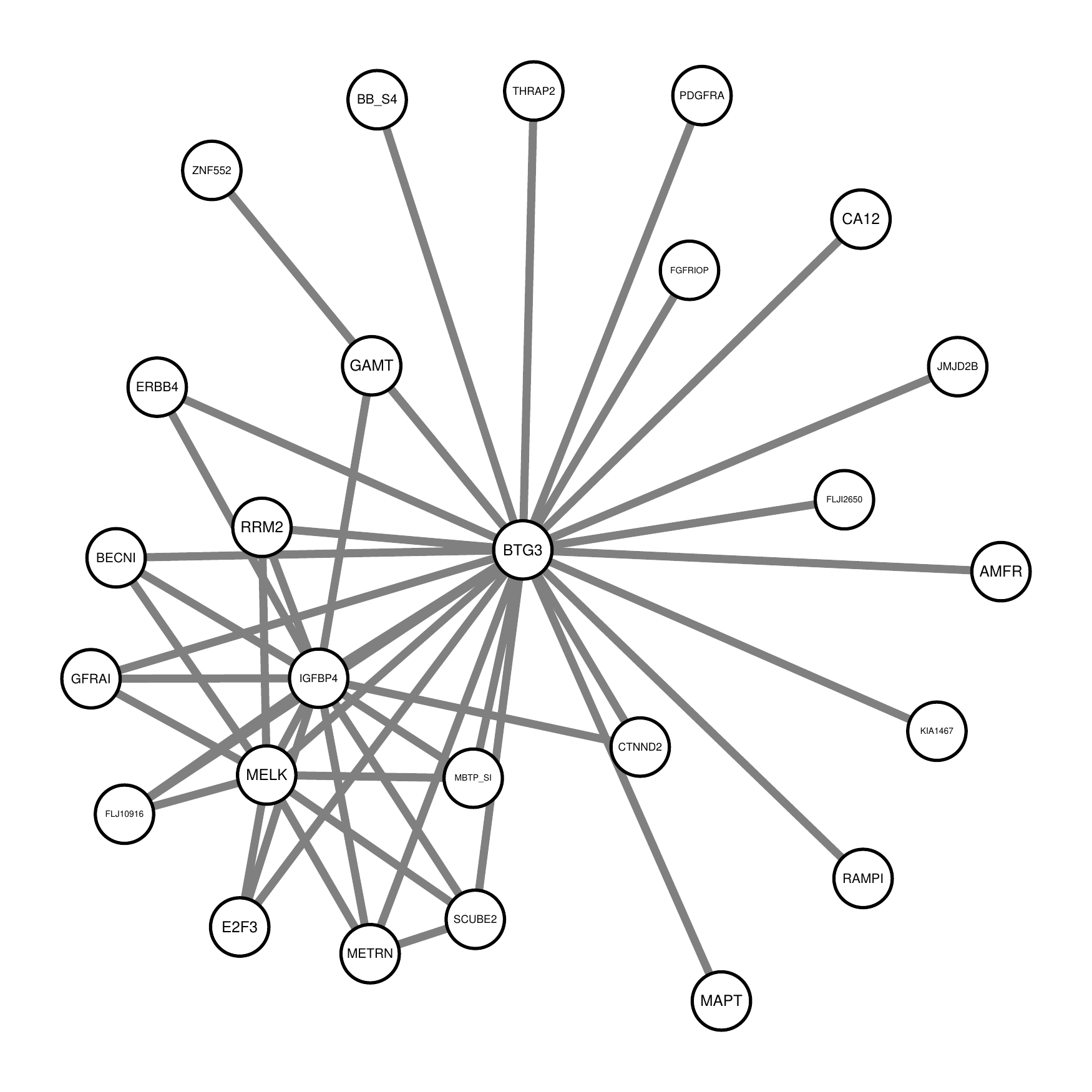}}
   &
    \subfloat{\includegraphics[width=0.28\textwidth]{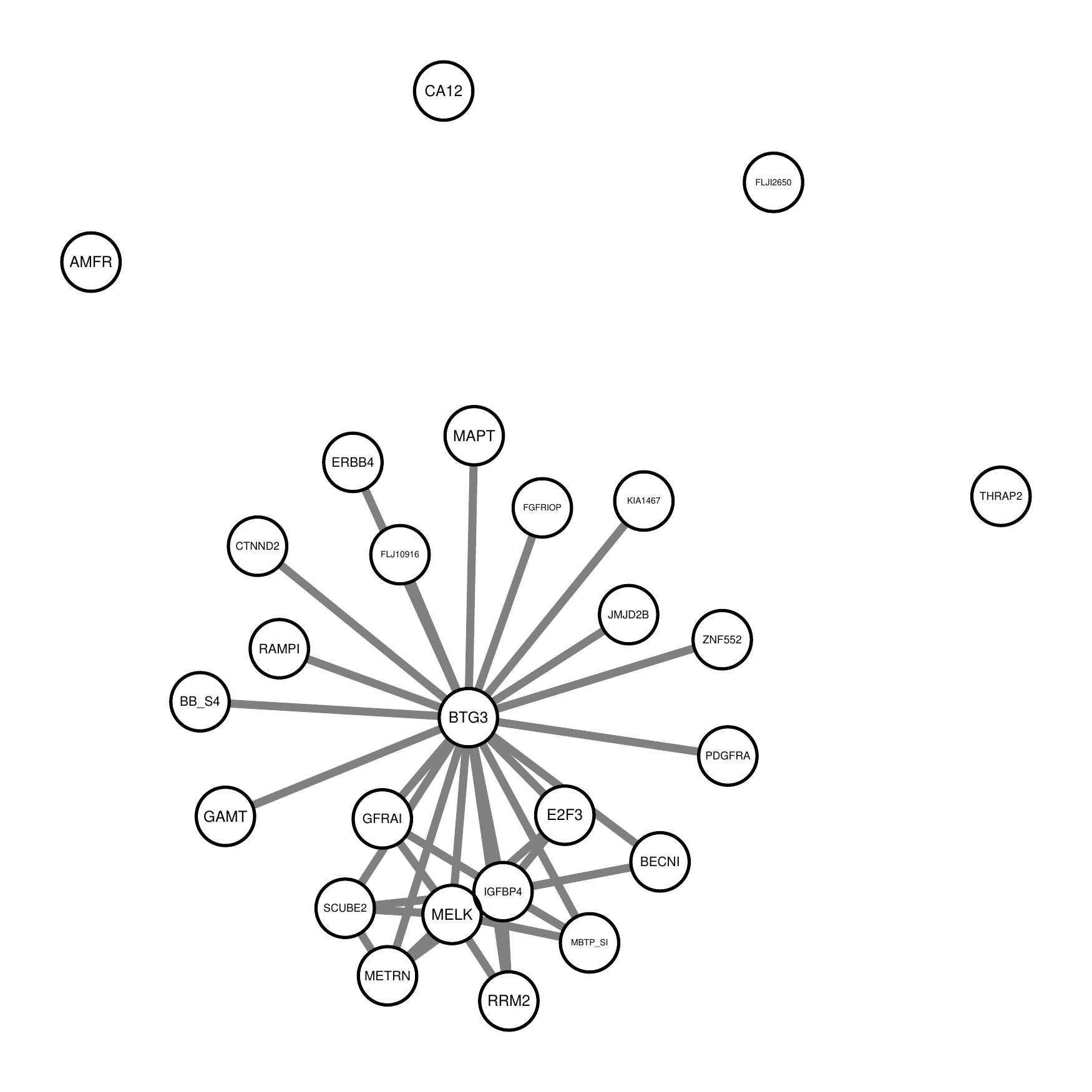}}
    &  \subfloat{\includegraphics[width=0.28\textwidth]{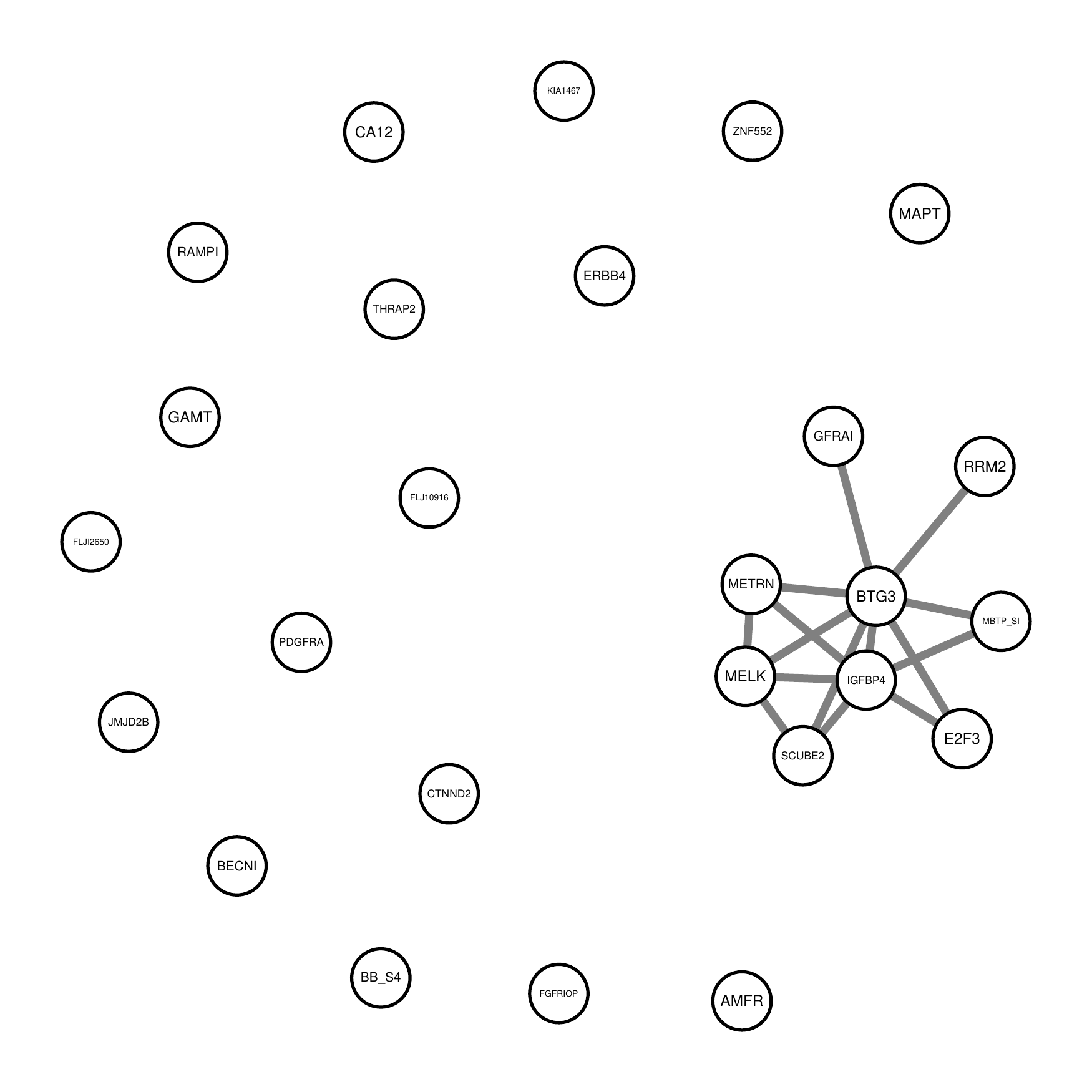}}
  \\ \hline
  not-pCR class  & \subfloat{\includegraphics[width=0.28\textwidth]{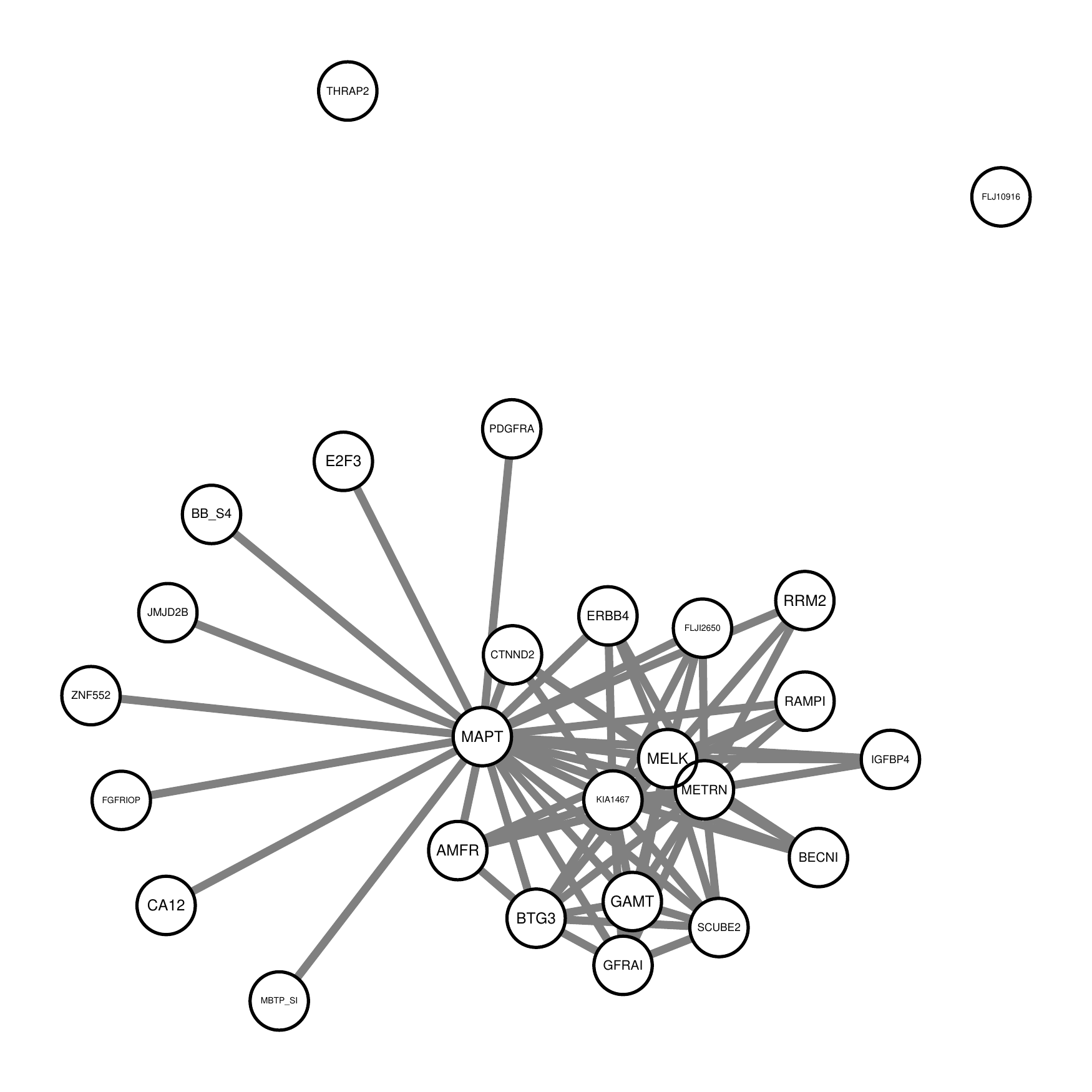}}
      &
       \subfloat{\includegraphics[width=0.28\textwidth]{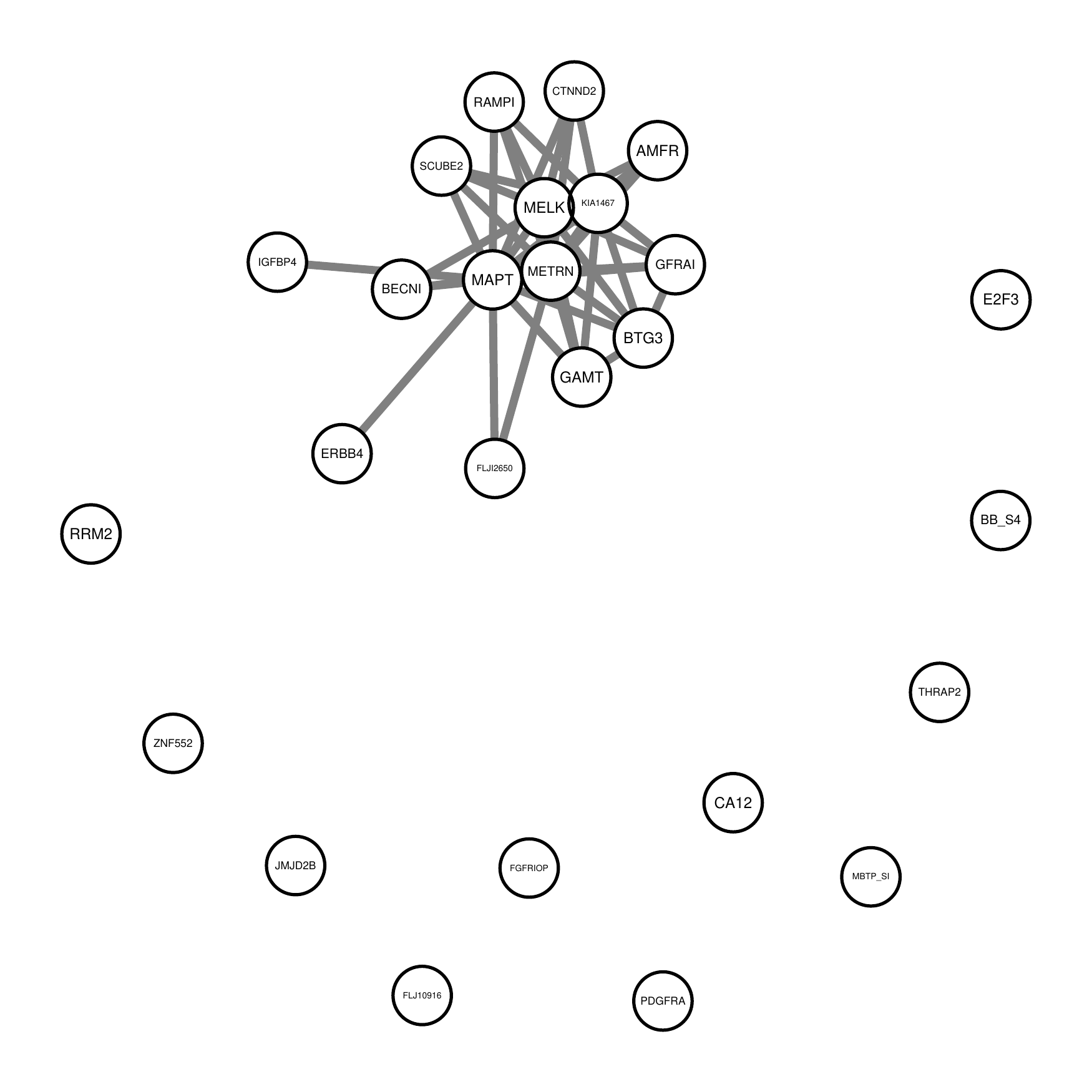}}
    &
      \subfloat{\includegraphics[width=0.28\textwidth]{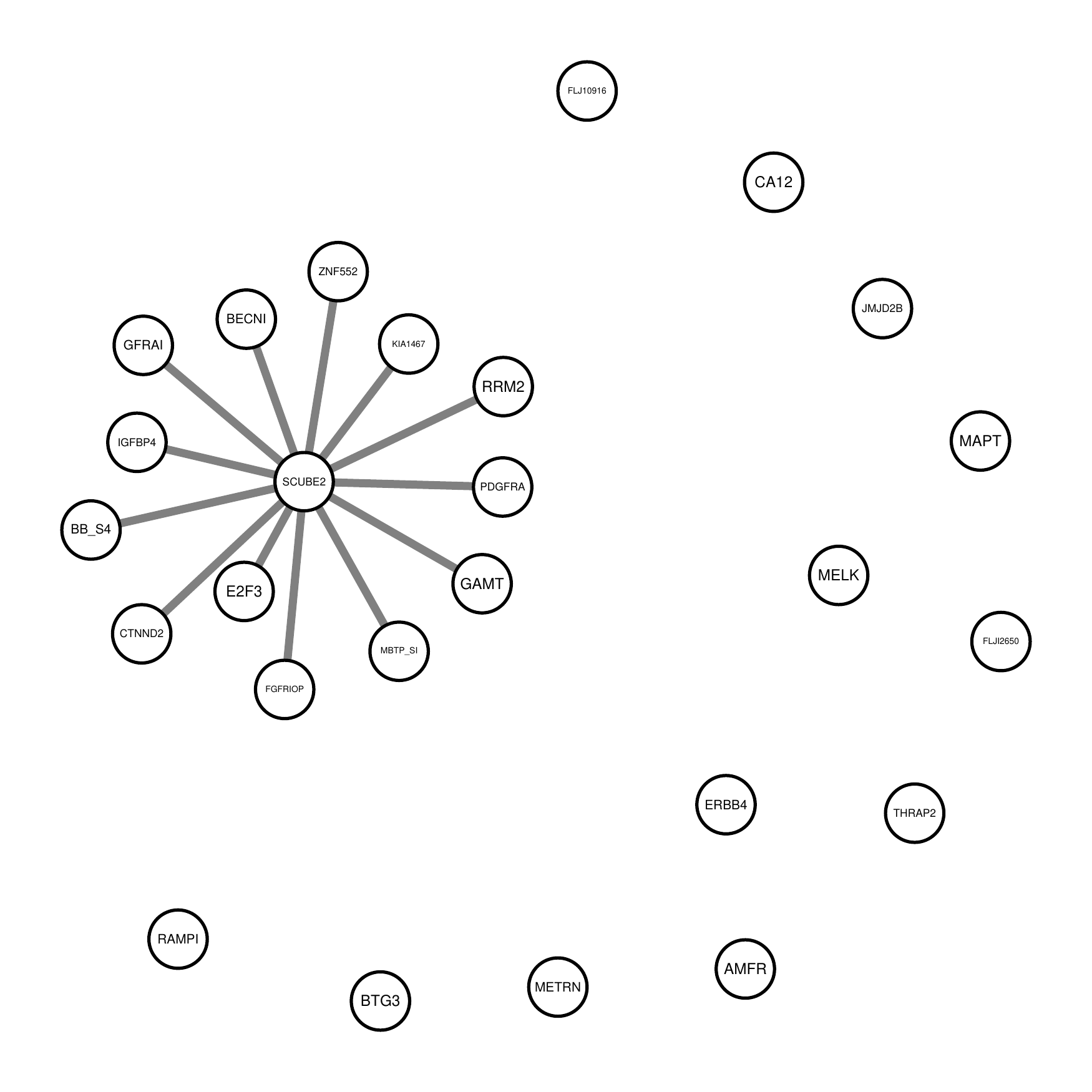}}
 \\  \hline \hline
\end{tabular}
\end{sidewaysfigure}

Figure \ref{res-app} shows the inferred networks of the three methods for studying pCR and not-pCR classes.
The penalty values are chosen by the revised cross-validation discussed in Section \ref{sec:selection}.
The results of the proposed WGLasso method indicate that the structure of the not-pCR network is more centered than that of the pCR-network.
All the edges in the network from not-pCR connect with SCUBE$2$, which implies its significant role in the study.
The network for the pCR case involves more genes and the pairwise relations are discovered between BTG$3$, METRN and MELK, also among BTG$3$, SCUBE$2$ and IGFBP$4$, to name a few.
Our findings confirm the results in Ambroise et al.\,(2009), where SCUBE$2$ is also identified as a hub gene for the not-pCR class and the cluster structure is discovered for the pCR class.
In contrast, results from the glasso and the tlasso methods tend to give less sparse structures, where more gene connections are identified in either the pCR or not-pCR networks.
Moreover, in this case study, we have an example of a mixed data set with non-pCR class as the majority portion (more than 75\%) of the data. To verify whether our proposed method can be well suited to identify the largest portion of the data that `matches' the model, we tried our method on the complete data (both pCP and non-pCR classes) as well. As expected, it successfully recovered the same findings as we applied our method to non-pCR class only, i.e., SCUBE$2$ was identified as the hub gene.

To further elaborate the usefulness of the proposed method in exploring the gene network structure, we report in Figure \ref{res-app-wgl} the estimated gene networks under \nx{low ($\rho = 0.1$), middle ($\rho = 0.6$) and high-level ($\rho = 0.9$)} values of the tuning parameter $\rho$ for both pCR and non-pCR.
For the pCR cluster, gene BTG$3$ consistently shows its importance to the network for all levels of penalty values.
This further validates the significance of BTG$3$, as indicated in  Ambroise et al.\,(2009).
For the non-pCR cluster, we observe that genes MAPT, METRN, and MELK tend to be more significant in the low and middle penalized networks, and gene SCUBE$2$ is identified as the key one in the high level network. This may provide additional information on potentially important genes when further investigating the network structure of the two classes.

\section{Discussion} \label{sec:concl}
In this work, we proposed robust estimation for sparse precision matrix by maximizing an $l_1$ regularized objective function with weighted sample covariance matrix. The robustness of our proposed objective function can be justified from a nonparametric perspective of the ISE criterion.
Adopting a majorization-minimization technique, we solved a relaxed optimization problem and developed an iteratively weighted graphical lasso estimation procedure.
Under such a formulation, the weights can be updated automatically in each iteration, which effectively reduces the influence of potential abnormal observations. This opens up a new possibility to better infer graphical models. Note that we adopt the $l_1$ penalty on the objective function to pursue the sparse structure of the graphical model.
Other penalties can also be considered into the proposed method, such as the elastic net (Zou and Hastie, 2005), which is a mixture of $l_{1}$ and $l_{2}$ penalties.

We show that the proposed estimator is asymptotically consistent, which is confirmed in the numerical experiments.
Additionally, an extensive simulation study is conducted to show that the proposed estimate is robust to abnormal observations in comparison with conventional methods.
In terms of computational efficiency, our simulation experience is that the computation time of the proposed algorithm is similar to that of the glasso method, but much faster than the tlasso method.
Note that the tlasso method appears to require more computational time because it uses the EM algorithm.

Finally, the current theoretical results are not very surprising since the approximated objective function can be considered as likelihood based weighted graphical lasso. However, the original optimization problem is not a likelihood based function as it considers the summation of probabilities. The theoretical investigation will be much more challenging. A possible direction is to connect the summation of probabilities to the density power divergence (Basu et al., 1998).
It will be also interesting to investigate asymptotic results on the direction of characterizing the break-down behavior of the estimator as in \"{O}llerer and Croux (2015).
However, in our method of weighted graphical lasso, weights keep changing over iterations, which requires some new tools for investigating the break-down properties.
Loh and Tan (2018) has investigated the break-down behavior of a one-step robust graphical lasso.
Their theoretical techniques could be helpful for further investigating theoretical properties of our proposed method.

\appendix
\section{The Proof of Theorem \ref{thm:consistency}}\label{sec:app1}
To prove Theorem \ref{thm:consistency}, we first establish the asymptotic property of the solution of the objective function in \eqref{eq:obj-wgl-3rd} as in Lemma \ref{thm:consistency2} below.
\begin{lemma}\label{thm:consistency2}
	Let $\tilde{\bm \Omega}$ denote the solution of the objective function in \eqref{eq:obj-wgl-3rd} with an initial estimator ${\bm \Omega}^{(0)}$ which is a consistent estimator of $\bm \Omega$.
	If $\sqrt{n}\rho \rightarrow \rho_0 \geq 0 $ as $n \rightarrow \infty$, \begin{equation*}
	\sqrt{n}(\tilde{\boldsymbol{\Omega}}-\boldsymbol{\Omega}) \rightarrow_d \argmin_{\bm{U}=\bm{U}^T} ~ V(\bm U),
	\end{equation*}
	where $V(\bm U)=\mbox{tr}(\textbf{U}\boldsymbol{\Sigma} \textbf{U} \boldsymbol{\Sigma})+\mbox{tr}(\textbf{U}\textbf{Q})+\rho_0\sum_{i,j}\{ u_{ij}\mbox{sign}(c_{ij})I(c_{ij}\neq 0)+|u_{ij}|I(c_{ij}=0) \}$, where $\textbf{Q}$ is a random symmetric $p \times p$ matrix such that $vec(\textbf{Q})\sim N(0, \boldsymbol{\Lambda})$, and $\boldsymbol{\Lambda}$ is such that $cov(q_{ij}, q_{i'j'})=\left(\frac{2}{\sqrt{3}}\right)^p cov(X^{(i)}X^{(j)}, X^{(i')}X^{(j')})$.
\end{lemma}
\begin{proof}
	Define $V_n(U)$ as
	\begin{eqnarray*}
		V_n(U)&=&-\log\left|\boldsymbol{\Omega}+\frac{\textbf{U}}{\sqrt{n}}\right|+\mbox{tr}\left\{\left(\boldsymbol{\Omega}+\frac{\textbf{U}}{\sqrt{n}}\right)\textbf{S}^*\right\}+\rho\left\|\boldsymbol{\Omega}+\frac{\textbf{U}}{\sqrt{n}}\right\|_1\\
		&& + \log|\boldsymbol{\Omega}|-\mbox{tr}(\boldsymbol{\Omega}\textbf{S}^*)-\rho\|\boldsymbol{\Omega}\|_1.
	\end{eqnarray*}
	Note that $\mbox{argmin }V_n(\bf U)=\sqrt{n}(\tilde{\boldsymbol{\Omega}}-\boldsymbol{\Omega})$.
	Using similar arguments as in the proof of Theorem 1 in Yuan and Lin (2007), it follows that
	\begin{eqnarray}\label{eq1}
	\log\left|\boldsymbol{\Omega}+\frac{\textbf{U}}{\sqrt{n}}\right|-\log|\boldsymbol{\Omega}| = \frac{\mbox{tr}(\textbf{U}\boldsymbol{\Sigma})}{\sqrt{n}}-\frac{\mbox{tr}(\textbf{U}\boldsymbol{\Sigma}\textbf{U}\boldsymbol{\Sigma})}{n}+o\left(\frac{1}{n}\right),
	\end{eqnarray}
	\begin{eqnarray}\label{eq2}
	\mbox{tr}\left\{\left(\boldsymbol{\Omega}+\frac{\textbf{U}}{\sqrt{n}}\right)\textbf{S}^*\right\}-\mbox{tr}(\boldsymbol{\Omega}\textbf{S}^*)=\mbox{tr}\left(\frac{\textbf{U}\textbf{S}^*}{\sqrt{n}}\right)=
	\frac{\mbox{tr}(\textbf{U}\boldsymbol{\Sigma})}{\sqrt{n}}+\frac{\mbox{tr}\{\textbf{U}(\textbf{S}^*-\boldsymbol{\Sigma})\}}{\sqrt{n}},
	\end{eqnarray}
	and
	\begin{eqnarray}\label{eq3}
	\rho\left\|\boldsymbol{\Omega}+\frac{\textbf{U}}{\sqrt{n}}\right\|_1 - \rho\|\boldsymbol{\Omega}\|_1 &=& \rho \sum_{i,j} \left(\left|c_{ij}+\frac{u_{ij}}{\sqrt{n}}\right|-|c_{ij}|\right) \nonumber \\
	&&=\frac{\rho}{\sqrt{n}}\sum_{i,j}\{ u_{ij}\mbox{sign}(c_{ij})I(c_{ij}\neq 0)+|u_{ij}|I(c_{ij}=0) \},
	\end{eqnarray}
	where $c_{ij}$ is the ($i$, $j$)th entry of $\boldsymbol{\Omega}$. Note that equation \eqref{eq3} holds only when $n$ is sufficiently large. Then, combining \eqref{eq1}, \eqref{eq2}, and \eqref{eq3}, we have
	\begin{equation*}
	nV_n(\textbf{U})=\mbox{tr}(\textbf{U}\boldsymbol{\Sigma}\textbf{U}\boldsymbol{\Sigma})+\mbox{tr}(\textbf{U}\textbf{Q}_n)+\sqrt{n}\rho\sum_{i,j}\{ u_{ij}\mbox{sign}(c_{ij})I(c_{ij}\neq 0)+|u_{ij}|I(c_{ij}=0) \}+o(1),
	\end{equation*}
	where $\textbf{Q}_n=\sqrt{n}(\textbf{S}^*-\boldsymbol{\Sigma})$. Let $q_{n,ij}$ denote the element of the matrix $\textbf{Q}_n$. Then,
	\begin{equation*}
	cov(q_{n,ij}, q_{n,i'j'})= \frac{\sum_{k=1}^n w_k^2}{n} cov(X^{(i)}X^{(j)}, X^{(i')}X^{(j')}),
	\end{equation*}
	where $\frac{\sum_{k=1}^n w_k^2}{n} \rightarrow \left(\frac{2}{\sqrt{3}}\right)^p$ as $n \rightarrow \infty$ due to Lemma \ref{lem:w} below. Also note that $E(\textbf{S}^*)=\boldsymbol{\Sigma}$. Therefore, $\textbf{Q}_n \rightarrow_d N(0, \boldsymbol{\Lambda})$ by the central limit theorem, and thus $nV_n(\textbf{U}) \rightarrow_d V(\bf{U})$. Since $nV_n(\textbf{U})$ and $V(\bf U)$ are both convex and $V(\bf U)$ has a unique minimum, $\mbox{argmin }nV_n(\textbf{U}) \rightarrow_d \mbox{argmin }V(\textbf{U})$. Since $\mbox{argmin }V_n(\textbf U)=\sqrt{n}(\tilde{\boldsymbol{\Omega}}-\boldsymbol{\Omega})$, it follows that $\sqrt{n}(\hat{\boldsymbol{\Omega}}-\boldsymbol{\Omega}) \rightarrow_d \mbox{argmin } V(\textbf{U})$.
\end{proof}

A careful examination of the proof of Lemma \ref{thm:consistency2} shows that the required condition to generalize the result for solving a single WGLasso problem in \eqref{eq:obj-wgl-3rd} to a sequence of WGLasso problems in Algorithm \ref{wgl} is the consistency of $\hat{\bf \Omega}$ in \textbf{Step 3}. The proof of the consistency of $\hat{\bf \Omega}$ is straightforward from the result of Lemma \ref{thm:consistency2}.

\begin{lemma}\label{lem:w}
	With an initial estimator ${\bf \Omega}^{(0)}$ which is a consistent estimator of $\boldsymbol{\Omega}$, we have
	\begin{equation*}
	\frac{1}{n}\sum_{i=1}^n w_i^2 \rightarrow \left(\frac{2}{\sqrt{3}}\right)^p \mbox{ as } n \rightarrow \infty,
	\end{equation*}
	where $w_{i} = f({\bf x}_{i}) / [\frac{1}{n} \sum_{j=1}^{n} f({\bf x}_{j})]$, where $f({\bf x}_{i})$ is the pdf of $N_p(\bm 0, {{\bf \Omega}^{(0)}}^{-1})$ when ${\bf x}={\bf x}_i$.
\end{lemma}
\begin{proof}
	\begin{equation*}\label{eq:w1}
	\frac{1}{n}\sum_{i=1}^n w_i^2=\frac{1}{n}\sum_{i=1}^n  \frac{(f(\textbf{x}_i))^2}{[\frac{1}{n}\sum_{j=1}^n f(\textbf{x}_j)]^2}
	=\frac{1}{[\frac{1}{n}\sum_{j=1}^n f(\textbf{x}_j)]^2}\left(\frac{1}{n}\sum_{i=1}^n (f(\textbf{x}_i))^2\right).
	\end{equation*}
	By the law of large numbers and ${\bf \Omega}^{(0)} \rightarrow {\bf \Omega}$,
	\begin{equation*}\label{eq:w2}
	\frac{1}{n}\sum_{j=1}^n f(\textbf{x}_j)
	%\rightarrow \mathbb{E}f(\textbf{x}|\Omega^{(0)})
	\rightarrow \mathbb{E}f(\textbf{x}|{\bf \Omega}) =\int_{\mathbb{R}^p} [f(\textbf{x})]^2 d\textbf{x}=2^{-p}\pi^{-p/2}|\boldsymbol{\Sigma}|^{-1/2}
	\end{equation*}
	and
	\begin{equation*}\label{eq:w3}
	\frac{1}{n}\sum_{i=1}^n (f(\textbf{x}_i))^2
	%\rightarrow \mathbb{E}(f(\textbf{x})|\Omega^{(0)})^2
	\rightarrow \mathbb{E}(f(\textbf{x})|{\bf \Omega})^2 =\int_{\mathbb{R}^p} [f(\textbf{x})]^3 d\textbf{x}=2^{-p}~ 3^{-p/2}\pi^{-p}|\boldsymbol{\Sigma}|^{-1}.
	\end{equation*}
	Therefore,
	\begin{equation*}
	\frac{1}{n}\sum_{i=1}^n w_i^2 \rightarrow (2^{-p}\pi^{-p/2}|\boldsymbol{\Sigma}|^{-1/2})^{-2}~2^{-p}~ 3^{-p/2}\pi^{-p}|\boldsymbol{\Sigma}|^{-1}=2^p~ 3^{-p/2}
	=\left(\frac{2}{\sqrt{3}}\right)^p.
	\end{equation*}
\end{proof}

\clearpage

\section{The 26 Key Genes used in Section 4.1}\label{sec:app2}
\begin{table}[htbp]
\centering %\caption{Genes in the Network Inference}
\label{tab:genes-desp}
\footnotesize
\begin{tabular}{rl} \hline
 Gene Symbol &  Gene Name \\ \hline
AMFR &  Autocrine motility factor receptor \\
BBS4 &  Bardet-Biedl syndrome 4 \\
BECN1 &  Beclin 1 (coiled-coil, myosin-like BCL2 interacting protein) \\
BTG3 &  BTG family, member 3 \\
CA12 &  Carbonic anhydrase XII \\
CTNND2 &  Catenin, delta 2 \\
E2F3 &  E2F transcription factor 3 \\
ERBB4 &  Verba erythroblastic leukemia viral oncogene homolog 4(avian) \\
FGFR1OP &  FGFR1 oncogene partner \\
FLJ10916 &  Hypothetical protein FLJ10916 \\
FLJ12650 &  Hypothetical protein FLJ12650 \\
GAMT &  Guanidinoacetate N-methyl transferase \\
GFRA1 &  GDNF family receptor 1 \\
IGFBP4 &  Insulin-like growth factor binding protein 4 \\
JMJD2B &  Jumonji domain containing 2B \\
KIAA1467 &  KIAA1467 protein \\
MAPT &  Microtubule-associated protein \\
MBTP-S1 &  Hypothetical protein \\
      MELK &  Maternal embryonic leucine zipper kinase \\
     MTRN &  Meteorin, glial cell differentiation regulator \\
    PDGFRA &  Human clone 23, 948 mRNA sequence \\
     RAMP1 &  Receptor (calcitonin) activity modifying protein 1 \\
      RRM2 &  Ribonucleotide reductase M2 polypeptide \\
    SCUBE2 &  Signal peptide, CUB domain EGF-like 2 \\
    THRAP2 &  Thyroid hormone receptor associated protein 2 \\
    ZNF552 &  Zinc finger protein 552 \\ \hline
\end{tabular}
\end{table}

\clearpage


\begin{thebibliography}{99}
\bibitem[\protect\citeauthoryear{Ambroise, Chiquet and Matias}{Ambroise et~al.}{2009}]{ambroise:09}
Ambroise, C., Chiquet, J., and Matias, C. (2009).
\newblock ``Inferring Sparse Gaussian Graphical Models with Latent Structure''.
\newblock {\em Electronic Journal of Statistics\/} {\bf 3}:205--238.

\bibitem[\protect\citeauthoryear{Basu, Harris and Jones}{Basu et~al.}{1998}]{basu:98}
Basu, A., Harris, I., Hjort, N., and Jones, M. (1998).
\newblock ``Robust and Efficient Estimation by Minimising a Density Power Divergence''.
\newblock {\em Biometrika\/} {\bf 85}:549--559.

\bibitem[\protect\citeauthoryear{Banerjee,  El Ghaoui and Natsoulis}{Banerjee et~al.}{2006}]{banerjee:06}
Banerjee, O., El~Ghaoui, L., and Georges, N. (2006).
\newblock ``Convex Optimization Techniques for Fitting Sparse Gaussian Graphical Models''.
\newblock {\em Proceedings of the 23rd International Conference on Machine Learning\/}: 89--96.

%\bibitem[\protect\citeauthoryear{Bickel and Levina}{Bickel and Levina}{2008}]{bickel:08}
%Bickel, P. J., and Levina, E. (2008).
%\newblock ``Regularized Estimation of Large Covariance Matrices''.
%\newblock {\em Annals of Statistics\/} \textbf{36(1)}:199--227.

\bibitem[\protect\citeauthoryear{Boyd and Vandenberghe}{Boyd and Vandenberghe}{2004}]{boyd:04}
Boyd, S., and Vandenberghe, L. (2004).
\newblock {\em Convex Optimization\/}.
\newblock Cambridge University Press, New York.

\bibitem[\protect\citeauthoryear{Cai,Liu and Luo}{Cai et~al.}{2011}]{cai:2011}
Cai, T., Liu, W., and Luo, X. (2011).
\newblock ``A Constrained L1 Minimization Approach to Sparse Precision Matrix Estimation''.
\newblock {\em Journal of American Statistical Association} {\bf 106}:594--607.

\bibitem[\protect\citeauthoryear{Chi and Scott}{Chi and Scott}{2012}]{chi:2012}
Chi, E.~C., and Scott, W. (2014).
\newblock ``Robust Parametric Classification and Variable Selection by a Minimum Distance Criterion''.
\newblock {\em Journal of Computational and Graphical Statistics\/}~{\bf 23(1)}:111--128.

\bibitem[\protect\citeauthoryear{Davis and Goadrich}{Davis and Goadrich}{2006}]{davis:2006}
Davis, J., and Goadrich, M. (2006).
\newblock ``The Relationship between Precision-Recall and ROC Curves''.
\newblock {\em Proceedings of the 23rd International Conference on Machine Learning\/}: 233--240.

%\bibitem[\protect\citeauthoryear{Dempster, Laird and Rubin}{Dempster et~al.}{1977}]{dempster:77}
%Dempster, A.~P., Laird, N.~M., and Rubin, D.~B. (1977).
%\newblock ``Maximum Likelihood from Incomplete Data via the EM Algorithm''.
%\newblock {\em Journal of the Royal Statistical Society, Series B\/}~{\bf 39(1)}:1--38.

%\bibitem[\protect\citeauthoryear{Dey}{Dey}{1988}]{dey:88}
%Dey, D.~K. (1988).
%\newblock ``Simultaneous Estimation of Eigenvalues''.
%\newblock {\em Annals of the Institute of Statistical Mathematics\/}~{\bf 40(1)}:137--147.

%\bibitem[\protect\citeauthoryear{Fan and Gijbels}{Fan and Gijbels}{1996}]{fan:96}
%Fan, J., and Gijbels, I. (1996).
%\newblock {\em Local Polynomial Modelling and Its Applications\/}.
%\newblock London: Chapman \& Hall.

\bibitem[\protect\citeauthoryear{Deng and Yuan}{Deng and Yuan}{2009}]{deng:2009}
Deng, X., and Yuan, M. (2009).
\newblock ``Large Gaussian Covariance Matrix Estimation with Markov Structures''.
\newblock {\em Journal of Computational and Graphical Statistics\/}~ {\bf 18(3)}: 640--657.

\bibitem[\protect\citeauthoryear{Finegold and Drton}{Finegold and Drton}{2011}]{finegold:11}
Finegold, M., and Drton, M. (2011).
\newblock ``Robust Graphical Modeling of Gene Networks using Classical and
Alternative $t$-distributions''.
\newblock {\em Annals of Applied Statistics\/}~{\bf 5(2A)}: 1057--1080.


\bibitem[\protect\citeauthoryear{Friedman, Hastie and Tibshirani}{Friedman et~al.}{2008}]{friedman:08}
Friedman, J., Hastie, T., and Tibshirani, R. (2008).
\newblock ``Sparse Inverse Covariance Estimation with the Graphical Lasso''.
\newblock {\em Biostatistics\/}~{\bf 9(3)}:432--441.

%\bibitem[\protect\citeauthoryear{Haff}{Haff}{1977}]{haff:77}
%Haff, L.~R. (1977).
%\newblock ``Minimax Estimators for A Multinormal Precision Matrix''.
%\newblock  {\em Journal of Multivariate Analysis \/}~{\bf 7(3)}:374--385.


\bibitem[\protect\citeauthoryear{Hess et~al.}{Hess et~al.}{2006}]{hess:06}
Hess, K.~R., Anderson, K., Symmans, W.~F., Valero, V., Ibrahim, N., Mejia, J.~A., Booser, D., Theriault, R.~L., Buzdar, U., Dempsey, P.~J., Rouzier, R., Sneige, N.,
Ross, J.~S., Vidaurre, T., G{\'o}mez, H.~L., Hortobagyi, G.~N., and Pustzai, L. (2006).
\newblock ``Pharmacogenomic Predictor of Sensitivity to Preoperative Chemotherapy with Paclitaxel and
Uorouracil, Doxorubicin, and Cyclophosphamide in Breast Cancer''.
\newblock {\em Journal of Clinical Oncology\/}~{\bf 24(26)}:4236--4244.


\bibitem[\protect\citeauthoryear{Horn and Johnson}{Horn and Johnson}{1990}]{horn:90}
Horn, R.~A., and Johnson, C.~R. (1996).
\newblock {\em Matrix Analysis\/}.
\newblock Cambridge University Press, New York.

\bibitem[\protect\citeauthoryear{Huang et~al.}{Huang et~al.}{2006}]{huang:06}
Huang, J.~Z., Liu, N., Pourahmadi, M., and Liu, L. (2006).
\newblock ``Covariance Matrix Selection and Estimation via Penalised Normal Likelihood''.
\newblock {\em Biometrika\/}~{\bf 93(1)}:85--98.

\bibitem[\protect\citeauthoryear{Hunter and Lange}{Hunter and Lange}{2004}]{hunter:04}
Hunter, D.~R., and Lange, K. (2004).
\newblock ``A Tutorial on MM Algorithms''.
\newblock {\em The American Statistician\/}~{\bf 58(1)}:30--37.

\bibitem[\protect\citeauthoryear{Kang and Deng}{Kang and Deng}{2020}]{kang:2020}
Kang, X., and Deng, X. (2020).
\newblock ``An Improved Modified Cholesky Decomposition Approach for Precision Matrix Estimation''.
\newblock {\em Journal of Statistical Computation and Simulation\/}~ {\bf 90(3)}: 443--464.

\bibitem[\protect\citeauthoryear{Lam and Fan}{Lam and Fan}{2009}]{LamFan:09}
Lam, C., and Fan, J. (2009).
\newblock ``Sparsistency and Rates of Convergence in Large Covariance Matrices Estimation".
\newblock {\em The Annals of Statistics \/}~{\bf 37}: 4254–4278.

\bibitem[\protect\citeauthoryear{Lauritzen}{Lauritzen}{1996}]{lauritzen:96}
Lauritzen, S.~L. (1996).
\newblock``Graphical Models''.
\newblock Oxford University Press, New York.

\bibitem[\protect\citeauthoryear{Liu}{Liu}{2004}]{liu:12}
Liu, H., Han, F., Yuan, M., Lafferty, J., and Larry Wasserman L. (2004).
\newblock ``High-Dimensional Semiparametric Gaussian Copula Graphical Models''.
\newblock {\em The Annals of Statistics\/}~{\bf 40(4)}:2293--2326.

\bibitem[\protect\citeauthoryear{Ledoit and Wolf}{Ledoit and Wolf}{2004}]{ledoit:04}
Ledoit, O., and Wolf, M. (2004).
\newblock ``A Well-Conditioned Estimator For Large Dimensional Covariance Matrices''.
\newblock {\em Journal of Multivariate Analysis\/}~{\bf 88(2)}:365--411.

\bibitem[\protect\citeauthoryear{Levina,Rothman and Zhu}{Levina,Rothman and Zhu}{2008}]{levina:08}
Levina, E., Rothman, A.~J., and Zhu, J. (2008).
\newblock ``Sparse Estimation of Large Covariance Matrices via A Nested Lasso Penalty''.
\newblock{\em Annals of Applied Statistics\/}~{\bf 2(1)}:245--263.




\bibitem[\protect\citeauthoryear{Po-Ling and Tan}{Loh and Tan}{2018}]{Poling:18}
Loh, P. L.,  and Tan, X. L. (2018).
\newblock ``High-dimensional robust precision matrix estimation: Cellwise corruption under $\epsilon $-contamination''.
\newblock{\em Electronic Journal of Statistics\/}~{\bf 12(1)}:1429--1467.

\bibitem[\protect\citeauthoryear{Markowitz}{Markowitz}{1952}]{mark:52}
Markowitz, H. (1952).
\newblock ``Portfolio Selection''.
\newblock {\em The Journal of Finance\/}~{\bf 7(1)}:77--91.


\bibitem[\protect\citeauthoryear{Meinshausen}{Meinshausen}{2006}]{mein:06}
Meinshausen, N., B{\"u}hlmann, P. (2006).
\newblock ``High Dimensional Graphs and Variable Selection with the Lasso''.
\newblock {\em Annals of Statistics\/}~{\bf 34(3)}:1436--1462.


\bibitem[\protect\citeauthoryear{Miyamura and Kano}{Miyamura and Kano}{2006}]{miya:06}
Miyamura, M., and Kano, Y. (2006).
\newblock ``Robust Gaussian Graphical Modeling''.
\newblock {\em Journal of Multivariate Analysis\/}~{\bf 97(7)}:1525--1550.


\bibitem[\protect\citeauthoryear{Natowicz et~al.}{Natowicz et~al.}{2008}]{nato:08}
Natowicz, R., Incitti, R., Horta, E.~G., Charles, B., Guinot, P., Yan, K., Coutant, C.,
Andr{\'e}, F., Pusztai, R., and Rouzier, L. (2008).
\newblock ``Prediction of the Outcome of A Preoperative Chemotherapy in Breast Cancer using DNA Probes that Provide Information on Both Complete and Incomplete Response''.
\newblock{\em BMC Bioinformatics \/}~{\bf 9(149)}:1--17.


\bibitem[\protect\citeauthoryear{Pe{\'n}a and J. Prieto}{Pe{\'n}a and J. Prieto}{2001}]{pena:01}
Pe{\'n}a, D., and Prieto, F.~J. (2001).
\newblock ``Multivariate Outlier Detection and Robust Covariance Matrix Estimation''.
\newblock {\em Technometrics\/}~{\bf 43(3)}:286--310.

%\bibitem[\protect\citeauthoryear{Peng, Wang, Zhou and Zhu}{Peng et~al.}{2009}]{peng:09}
%Peng, J., Wang, P., Zhou, N., and Zhu, J. (2009).
%\newblock ``Partial Correlation Estimation by Joint Sparse Regression Models''.
%\newblock {\em Journal of the American Statistical Association\/}~{\bf 104(486)}:735--746.

\bibitem[\protect\citeauthoryear{Perron}{Perron}{1992}]{perron:92}
Perron, F. (1992).
\newblock ``Minimax Estimators of A Covariance Matrix''.
\newblock  {\em Journal of Multivariate Analysis \/}~{\bf 43(1)}:16--28.

\bibitem[\protect\citeauthoryear{Pourahmadi}{Pourahmadi}{2000}]{pour:2000}
Pourahmadi, M. (2000).
\newblock ``Maximum Likelihood Estimation Of Generalized Linear Models For Multivariate Normal Covariance Matrix''.
\newblock{\em Biometrika\/}~{\bf 87(2)}:425–-435.

%\bibitem[\protect\citeauthoryear{Rothman, Bickel, Levina and Zhu}{Rothman et~al.}{2008}]{rothman:08}
%Rothman, A., Bickel, P., Levina, E., and Zhu, J. (2008).
%\newblock ``Sparse Permutation Invariant Covariance Estimation''.
%\newblock {\em Electronic Journal of Statistics\/}~{\bf 2}:494--515.

\bibitem[\protect\citeauthoryear{Ollerer and Croux}{Ollerer and Croux}{2015}]{ollerer:15}
\"{O}llerer, V., and Croux, C. (2015).
\newblock ``Robust high-dimensional precision matrix estimation''.
\newblock {\em Modern Nonparametric, Robust and Multivariate Methods\/}:325-350.
\newblock  Springer: New York.

\bibitem[\protect\citeauthoryear{Rudemo}{Rudemo}{1982}]{rudemo:82}
Rudemo, M. (1982).
\newblock ``Empirical Choice of Histograms and Kernel Density Estimators''.
\newblock {\em Scandinavian Journal of Statistics\/}~{\bf 9(2)}:65--78.

\bibitem[\protect\citeauthoryear{Sharpe}{Sharpe}{1994}]{sharpe:94}
Sharpe, W.~F. (1994).
\newblock ``The Sharpe Ratio''.
\newblock {\em The Journal of Portfolio Management \/}~{\bf 21(1)}:49--58.

\bibitem[\protect\citeauthoryear{Scott}{Scott}{2001}]{scott:01}
Scott, D.~W. (2001).
\newblock ``Parametric Statistical Modeling by Minimum Integrated Square Error''.
\newblock {\em Technometrics\/}~{\bf 43(3)}:274--285.

\bibitem[\protect\citeauthoryear{Sun and Li}{Sun and Li}{2012}]{sun:12}
Sun, H., and Li, H. (2012).
\newblock ``Robust Gaussian Graphical Modeling via $l_1$ Penalization''.
\newblock {\em Biometrics\/}~{\bf 68(4)}:1197--1206.

\bibitem[\protect\citeauthoryear{Tarr}{Tarr}{2016}]{tarr:16}
Tarr, G.,  M\"{u}ller, S., and Weber, N.~C. (2016).
\newblock ``Robust estimation of precision matrices under cellwise contamination''.
\newblock {\em Journal Computational Statistics \& Data Analysis}~{\bf 93(C)}:404--420.


\bibitem[\protect\citeauthoryear{Tibshirani}{Tibshirani}{1996}]{tibshirani:96}
Tibshirani, R. (1996).
\newblock ``Regression Shrinkage and Selection via the Lasso''.
\newblock {\em Journal of the Royal Statistical Society, Series B}~{\bf 58(1)}:267--288.

\bibitem[\protect\citeauthoryear{Whittaker}{Whittaker}{1990}]{whittaker:90}
Whittaker, J. (1990).
\newblock``Graphical Models in Applied Multivariate Statistics''.
\newblock John Wiley \& Sons: New York.

\bibitem[\protect\citeauthoryear{Won}{Won}{2013}]{won:13}
Won J-H, Lim J, Kim S-J, and Rajaratnam B. (2013).
\newblock ``Condition Number Regularized Covariance Estimation''.
\newblock {\em Journal of the Royal Statistical Society Series B\/}.
{\bf 75(3)}: 427--450.

%\bibitem[\protect\citeauthoryear{Wong, Carter and Kohn}{Wong, Carter and Kohn}{2003}]{wong:03}
%Wong, F., Carter, C., and Kohn, R. (2003).
%\newblock ``Efficient Estimation of Covariance Selection Models''.
%\newblock {\em Biometrika\/}~{\bf 90(4)}:809–-830.

%\bibitem[\protect\citeauthoryear{Wu and Pourahmadi}{Wu and Pourahmadi}{2003}]{wu:03}
%Wu, W. B., and Pourahmadi, M. (2003).
%\newblock ``Nonparametric Estimation of Large Covariance Matrices of Longitudinal Data''.
%\newblock {\em Biometrika\/}~{\bf 90(4)}:831–-844.

\bibitem[\protect\citeauthoryear{Yuan and Lin}{Yuan and Lin}{2007}]{yuan:07}
Yuan, M., and Lin, Y. (2007).
\newblock ``Model Selection and Estimation in the Gaussian Graphical Model''.
\newblock {\em Biometrika\/}~{\bf 94(1)}:19--35.

\bibitem[\protect\citeauthoryear{Zou and Hastie}{Zou and Hastie}{2005}]{zou:05}
Zou, H., and Hastie, T. (2005).
\newblock ``Regularization and Variable Selection via the Elastic Net''.
\newblock {\em Journal of Royal Statistical Society, Series B\/}~{\bf 67(2)}:301--320.

\bibitem[\protect\citeauthoryear{Zou and Li}{Zou and Li}{2008}]{zou:08}
Zou, H., and Li, R. (2008).
\newblock ``One-step Sparse Estimates in Nonconcave Penalized Likelihood Models''.
\newblock  {\em Annals of Statistics\/}~{\bf 36(4)}: 1509--1533.

\end{thebibliography}
\end{document}